\documentclass[11pt,reqno]{amsart}
\usepackage{amssymb}
\usepackage{amsfonts}
\usepackage{amsmath}
\usepackage{stmaryrd}
\usepackage{physics}
\usepackage{braket}
\usepackage{dsfont}
\usepackage{bbold}
\usepackage{graphicx}
\usepackage{relsize}
\usepackage{makecell}
\usepackage[table,xcdraw]{xcolor}
\usepackage{enumerate}
\usepackage[pagebackref, colorlinks = true, linkcolor = blue, urlcolor  = blue, citecolor = red]{hyperref}
\usepackage[margin=1in]{geometry}

\usepackage{enumitem}
\usepackage{mathtools}
\usepackage{graphbox}

\usepackage{comment}
\usepackage{float}
\usepackage[capitalize]{cleveref}
\usepackage{booktabs}
\usepackage[skins]{tcolorbox}
\usepackage{bbding}
\usepackage{nicematrix}
\NiceMatrixOptions{cell-space-limits = 2pt}
\usepackage{caption} 
\usepackage{tikz}
\usetikzlibrary{arrows.meta, bending, chains}

\newcommand{\stargraph}[2]{\begin{tikzpicture}
    \node[circle,fill=black] at (360:0mm) (center) {};
    \foreach \n in {1,...,#1}{
         \node[circle,fill=black,inner sep=1pt] at ({\n*360/#1}:#2cm) (n\n) {};
        \draw (center)--(n\n);
        \node at (0,-#2*1.5) {}; % delete line to remove label
    }
\end{tikzpicture}}

\newcommand{\upwardsubseteq}{\rotatebox{90}{\(\subseteq\)}}

\newcommand{\upwardeq}{\rotatebox{90}{\(=\)}}

\newcommand{\completegraph}[2]{%
\begin{tikzpicture}[baseline=(current bounding box.center)]
    \foreach \n in {1,...,#1}{
        \node[circle,fill=black,inner sep=1pt] 
            (n\n) at ({360*\n/#1}:#2cm) {};
    }
    \foreach \i in {1,...,#1}{
        \foreach \j in {1,...,#1}{
            \ifnum\i<\j
                \draw  (n\i)--(n\j);
            \fi
        }
    }
\end{tikzpicture}}

\newcommand{\tabfigure}[2]{\raisebox{-.5\height}{\includegraphics[#1]{#2}}}

\renewcommand{\epsilon}{\varepsilon}
\renewcommand{\phi}{\varphi}

\newcommand{\MLDUI}[1]{\mathcal{M}_{#1}(\mathbb{C})^{\times 2}_{\mathbb{C}^{#1}}}

\newcommand{\R}[1]{\mathbb{R}^{#1}}
\newcommand{\C}[1]{\mathbb{C}^{#1}}

\newcommand{\M}[1]{\mathcal{M}_{#1}(\mathbb{C})}
\newcommand{\Mreal}[1]{\mathcal{M}_{#1}(\mathbb{R})}
\newcommand{\Msa}[1]{\mathcal{M}^{\mathrm{sa}}_{#1}(\mathbb{C})}
\newcommand{\Mrealsa}[1]{\mathcal{M}^{\mathrm{sa}}_{#1}(\mathbb{R})}

\newcommand{\sa}{\mathrm{sa}}

\newcommand{\PhiDUC}[2]{\Phi^{\mathsf{DUC}}_{{#1},{#2}}}
\newcommand{\PhiCDUC}[2]{\Phi^{\mathsf{CDUC}}_{{#1},{#2}}}
\newcommand{\PhiDUCCDUC}[2]{\Phi^{\mathsf{(C)DUC}}_{{#1},{#2}}}

\newcommand{\XLDUI}[2]{X^{\mathsf{LDUI}}_{{#1},{#2}}}
\newcommand{\XCLDUI}[2]{X^{\mathsf{CLDUI}}_{{#1},{#2}}}
\newcommand{\XLDUICLDUI}[2]{X^{\mathsf{(C)LDUI}}_{{#1},{#2}}}

\DeclareMathOperator{\diag}{diag}

\DeclareMathOperator{\Aut}{Aut}
\DeclareMathOperator{\Paley}{Paley}

\newcommand{\PPT}{\mathsf{PPT}}
\DeclareMathOperator{\EWP}{\mathsf{EWP}}
\newcommand{\PSDR}{\mathsf{PSD}^{\mathbb R}}
\newcommand{\PSDC}{\mathsf{PSD}^{\mathbb C}}
\DeclareMathOperator{\COPCP}{\mathsf{PCOP}}
\DeclareMathOperator{\CLDUI}{\mathsf{CLDUI}}
\DeclareMathOperator{\LDUI}{\mathsf{LDUI}}
\DeclareMathOperator{\COP}{\mathsf{COP}}
\DeclareMathOperator{\PCP}{\mathsf{PCP}}
\DeclareMathOperator{\DUC}{\mathsf{DUC}}
\DeclareMathOperator{\PDEC}{\mathsf{PDEC}}
\newcommand{\CP}{\mathsf{CP}}
\newcommand{\DNN}{\mathsf{DNN}}
\newcommand{\SPN}{\mathsf{SPN}}
\newcommand{\PDNN}{\mathsf{PDNN}}

\newtheorem{theorem}{Theorem}[section]
\newtheorem{proposition}[theorem]{Proposition}
\newtheorem{corollary}[theorem]{Corollary}
\newtheorem{lemma}[theorem]{Lemma}

\newtheorem*{conjecture*}{Conjecture}

\theoremstyle{definition}
\newtheorem{definition}[theorem]{Definition}
\theoremstyle{definition}
\newtheorem{remark}[theorem]{Remark}
\theoremstyle{definition}
\newtheorem{example}[theorem]{Example}
\theoremstyle{definition}
\newtheorem{question}[theorem]{Question}

\newcommand\vertarrowbox[3][6ex]{
  \begin{array}[t]{@{}c@{}} #2 \\
  \left\uparrow\vcenter{\hrule height #1}\right.\kern-\nulldelimiterspace\\
  \makebox[0pt]{\scriptsize#3}
  \end{array}
}

\newcommand{\id}{\mathrm{id}}
\newcommand{\la}{\langle}
\newcommand{\ra}{\rangle}

\newcommand{\K}{\mathsf{K}}

\newcommand{\Real}{\mathbb{R}}
\newcommand{\Comp}{\mathbb{C}}

\author{Aabhas Gulati}
\email{aabhas.gulati@math.univ-toulouse.fr}
\address{Institut de Mathématiques, Université de Toulouse, UPS, France.}

\author{Ion Nechita}
\email{nechita@irsamc.ups-tlse.fr}
\address{Laboratoire de Physique Th\'eorique, Universit\'e de Toulouse, CNRS, UPS, France}

\author{Sang-Jun Park}
\email{spark@irsamc.ups-tlse.fr}
\address{Laboratoire de Physique Th\'eorique, Universit\'e de Toulouse, CNRS, UPS, France}

\title[Positive maps and extendibility hierarchies from copositive matrices]{Positive maps and extendibility hierarchies\\from copositive matrices}

\begin{document}
\begin{abstract}
The characterization of positive, non-completely positive linear maps is a central problem in operator algebras and quantum information theory, where such maps serve as entanglement witnesses. This work introduces and systematically studies a new convex cone of \emph{pairwise copositive matrices}, denoted $\COPCP_n$. We establish that this cone is dual to the cone of pairwise completely positive matrices and, critically, provides a complete characterization for the positivity of the broad and physically relevant class of covariant maps. We provide a way to systematically lift matrices from the classical cone of copositive matrices, $\COP_n$, to the new pairwise cone $\COPCP_n$, thereby creating a powerful bridge between the well-studied theory of copositive forms and the structure of positive maps. We develop an analogous framework for decomposable maps, introducing the \emph{cone $\mathsf{PDEC}_n$ of pairwise decomposable matrices}. {For several families of linear maps having diagonal unitary symmetry such as generalized Choi maps, we characterize membership in these cones using simple properties of the parameters of the maps.}

As a primary application of this framework, we define a novel family of linear maps $\Phi_t^G$ parameterized by a graph $G$ and a real parameter $t$. We derive exact thresholds on $t$ that determine when these maps are positive, decomposable, or completely positive, linking these properties to fundamental graph-theoretic parameters. This construction yields vast \emph{new families of positive indecomposable maps}, for which we provide explicit examples derived from infinite classes of graphs, most notably rank 3 strongly regular graphs such as Paley graphs. 

On the dual side, we investigate the entanglement properties of large classes of symmetric states, such as the (mixture of) Dicke states. We prove that the \emph{sum-of-squares (SOS) hierarchies} used in polynomial optimization to approximate the cone of copositive matrices correspond precisely to dual cones of witnesses for different levels of the \emph{PPT bosonic extendibility hierarchy}. {In the setting of the DPS hierarchy for separability, we construct a large family of \textit{boundary} entanglement witnesses that are not certifiable by any level of the PPT bosonic extendibility hierarchy, answering a long standing open question from \cite{DPS04}.} Leveraging the duality, we also provide an explicit construction of bipartite (mixture of) Dicke states that are simultaneously entangled and $\mathcal{K}_r$-PPT bosonic extendible for any desired hierarchy level $r \geq 2$ and local dimension $n \geq 5$.
\end{abstract}

\maketitle

\tableofcontents

\section{Introduction} \label{sec:introduction}

The study of positive linear maps between matrix algebras is a cornerstone of modern functional analysis and operator theory \cite{stormer1963positive, choi1975completely, woronowicz1976positive}. A linear map $\Phi:\M{n} \to \M{n}$ is called \emph{positive} if it maps positive semidefinite matrices to positive semidefinite matrices. A stronger condition is that of \emph{complete positivity} \cite{paulsen2002completely}, where the amplification $\Phi \otimes \id_k$ remains positive for all ancilla dimensions $k$. While completely positive maps admit a simple structural characterization via the Choi-Kraus representation, the set of maps that are positive but not completely positive possesses a far more intricate structure.

Originating in operator algebra, the distinction between positivity and complete positivity turns out to have profound implications for \emph{quantum entanglement theory} \cite{horodecki2009quantum}. The celebrated \emph{Horodecki criterion} \cite{horodecki1996separability} established that a bipartite quantum state $\rho$ is entangled if and only if there exists a positive but not completely positive map $\Phi$ such that $[\Phi \otimes \id](\rho)$ is not positive semidefinite. Consequently, such maps $\Phi$ are synonymous with \emph{entanglement witnesses} \cite{terhal2000bell}, and the challenge of constructing and classifying them is one of the foremost problems in quantum information science \cite{nielsen2010quantum, aubrun2017alice, watrous2018theory}. Within the set of positive maps lies the important convex cone of \emph{decomposable maps}, that can be written as the sum of a completely positive map and a completely co-positive map (the composition of a completely positive map with the matrix transpose). Maps that are positive but not decomposable, known as \emph{indecomposable maps}, are of particular interest as they can detect subtle forms of entanglement, such as that found in \emph{PPT} (positive partial transpose) entangled states \cite{horodecki1997separability}. While it was shown by Størmer and Woronowicz that every positive map $\M{2} \to \M{2,3}$ is known to be decomposable \cite{stormer1963positive, woronowicz1976positive}, Choi \cite{choi1975} provided the first example of a positive indecomposable map $\M{3} \to \M{3}$, {which was later generalized in various related forms (see, e.g., \cite{choi1977,tanahashi1988,osaka1991,kye1992,cho1992,ha2003class,scala2024optimality}).} However, a general theory has remained largely out of reach. {\emph{Unextendible Product Bases} (UPBs) underpin a range of methods for constructing positive indecomposable maps in quantum information theory. DiVincenzo et al. \cite{divincenzo2003unextendible} define UPBs using orthogonality graphs, tiling patterns, and algebraic constructions; they show that the complement of a UPB yields a bound entangled state with a positive partial transpose that furnishes indecomposable maps. Terhal \cite{terhal2001family} leverages these UPBs to generate positive maps that resist decomposition by standard methods. Bravyi \cite{Bravyi2004Unextendible} classifies three‐qubit UPBs and demonstrates that uncompletable product bases yield bound entangled states, while Alon and Lovász \cite{alon2001unextendible} employ graph‐theoretic techniques to construct multipartite UPBs and derive bounds from orthogonality graphs. Several studies focus on explicit algebraic constructions of positive indecomposable maps: Chruściński and Kossakowski \cite{chruscinski2006class, Chruscinski2007On} introduce classes of maps based on cyclic bistochastic matrices and generalizations of the Choi map, as well as atomic maps characterized by their action on rank one projectors. Sarbicki and Chruściński \cite{Sarbicki2012class} construct exposed indecomposable maps in $2n \times 2n$ matrices using reduction maps, generalizing the Robertson map. Ha and Kye \cite{Ha2016Construction} present parameterized families of exposed indecomposable maps, with explicit conditions for exposedness and the bi-spanning property. Collins, Hayden, and the second author \cite{collins2015random} employ random matrix and free probability techniques, constructing $k$-positive maps via free convolution and random matrix models. Marciniak and Rutkowski \cite{Marciniak2017Merging} introduce a geometric ``merging'' technique, systematically combining two positive maps with additional operators and functionals to produce new positive, often exposed, indecomposable maps. Müller-Hermes \cite{Muller2018Decomposability} uses symmetrization and tensor product techniques, particularly with Werner maps, to analyze decomposability under tensor powers and construct new non-decomposable maps. Siudzińska \cite{Siudzinska2022Indecomposability} uses symmetric measurements (positive operator-valued measures) and operator bases to construct positive maps and indecomposable entanglement witnesses. 
}

Other than or beyond the PPT criterion, the study of bipartite entanglement has also motivated the study of the notion of \emph{extendibility} \cite{DPS04,christandl2007one,navascues2009power,doherty2014entanglement,harrow2017improved}. In particular, the notion of \emph{bosonic (or Bose-symmetric) extendibility} plays a central role in the study of quantum entanglement, providing one of the most powerful and systematic separability criteria. Introduced in the seminal work of Doherty, Parrilo, and Spedalieri \cite{DPS04}, the hierarchy of $r$-extendible states forms an increasing sequence of convex sets that converges to the set of separable states. This completeness ensures that every entangled state eventually fails some level of the hierarchy. Moreover, each level can be efficiently tested using semidefinite programming (SDP), making the DPS hierarchy not only theoretically elegant but also computationally tractable.

A stronger version of this idea arises when one focuses on bipartite \emph{bosonic} systems. In this setting, one may further require that the Bose-symmetric extensions simultaneously include both original subsystems, leading to what is called \emph{$r$-PPT bosonic exchangeability} \cite{christandl2007one} or \emph{PPT extendible with respect to the complete graph $\mathcal{K}_r$} \cite{ACG+23+}. This refinement yields a sharper and still complete hierarchy of tests that interpolate between the PPT and separability conditions. Thanks to the monogamy of entanglement \cite{terhal2004entanglement}, a state passing all these tests behaves as if it could be consistently extended to many identical copies while maintaining only a weak amount of entanglement across any bipartition. Despite the natural formulation of this criterion, explicitly constructing entangled states that satisfy such strong extendibility conditions remains a challenging problem, even for the simplest cases \cite{DPS04}.

A powerful approach to understanding cones of maps and states is through \emph{duality}. The cone of positive maps is dual to the cone of separable states, while the cone of decomposable maps is dual to the cone of PPT states \cite{eom2000duality, skowronek2009cones, kye2023compositions}. Duality also connects map properties to cones of matrices. The cone of \emph{completely positive matrices} $\CP_n$ \cite{berman2003completely}, consisting of matrices that can be factored as $XX^\top$ for some entrywise non-negative (rectangular) matrix $X$, has as its dual the cone of \emph{copositive matrices} $\COP_n$ \cite{shaked2021copositive}. A real symmetric matrix $M$ is in $\COP_n$ if $\la x, Mx \ra \geq 0$ for all non-negative vectors $x\in \R{n}_+$. Just like the set of positive maps, the cone $\COP_n$ is notoriously difficult to characterize, with membership being a co-NP-complete problem \cite{murty1987some}. The subset of $\COP_n$ comprising matrices that are a sum of a positive semidefinite and an entrywise non-negative matrix, denoted $\SPN_n$, is tractable and corresponds to the first level of a sum-of-squares (SOS) hierarchy that approximates $\COP_n$. The strict containment $\COP_n \supsetneq \SPN_n$ for $n \geq 5$, exemplified by the famous Horn matrix \cite{shaked2021copositive}, mirrors the distinction between positive and decomposable maps. Let us point out here that the use of copositive matrices for entanglement detection dates back to \cite{marconi2021entangled}, see also \cite{marconi2023entanglement,romero2025multipartite, marconi2025symmetric}. On the dual side, the connection between completely positive matrices and separability dates back to \cite{yu2016separability, tura2018separability}, see also \cite{johnston2019pairwise,singh2021diagonal,singh2021entanglement,berman2023completely,gulati2025entanglement,gulati2025witnessing}. Finally, the dual of (PPT-) bosonic extendibility hierarchy admits a natural interpretation in terms of sum-of-squares (SOS) representations of Hermitian forms \cite{DPS04,fang2021sum}. This duality connects the feasibility of an $r$-extendibility test to the existence of SOS certificates of block-positivity for corresponding entanglement witnesses. In this context, a fundamental open question was raised in \cite[Sec. VI]{DPS04}, which we state informally, 
\begin{center}
    Can the cone of entanglement witnesses (i.e.~block-positive operators or elements of the dual of the separable set) be obtained as the union of the witnesses of these DPS levels?
\end{center}
We refer to \cref{sec:DPS} for further details. Any counterexample to this question could be used to detect entanglement of quantum states that are $r$-PPT bosonic extendible for arbitrarily high $r\geq 1$. While partial results are known (e.g., inclusion of its interior), a full resolution remains open, reflecting the deep interplay between convex geometry, operator algebras, and optimization in the characterization of quantum entanglement.

On the other hand, a fruitful strategy for finding PPT entangled states and positive indecomposable maps is to impose \emph{symmetries} on the states and maps under consideration \cite{vollbrecht2001entanglement,eggeling2001separability,singh2021diagonal,park2024universal}. A particularly rich and physically relevant class is that of \emph{diagonal unitary covariant} (DUC) maps \cite{liu2015unitary, singh2021diagonal}, which commute with conjugations by diagonal unitary matrices. Such maps, and their close relatives the \emph{conjugate diagonal unitary covariant} (CDUC) maps, have recently been the subject of intense study due to their structured algebraic properties \cite{singh2022ppt, kennedy2018composition, singh2024ergodic}. Recent work established a fundamental isomorphism between the space of DUC/CDUC maps and the vector space of \emph{pairs of matrices} $(A,B)$ sharing a common diagonal \cite{johnston2019pairwise, singh2021diagonal}. Within this framework, a cone of pairwise completely positive matrices $\PCP_n$ was introduced, which was shown to correspond precisely to the class of \emph{entanglement breaking} DUC/CDUC maps. This breakthrough provided a complete description of the strongest form of positivity for these maps. However, a corresponding characterization for general positivity was left as a significant open problem. This gap motivates the central questions of this paper:
\smallskip
\begin{enumerate}
\renewcommand{\labelenumi}{\boxed{{\rm Q}\arabic{enumi}}}
\setlength{\itemsep}{.5em}
    \item Can we define and characterize a cone of matrix pairs that corresponds precisely to the set of all positive DUC/CDUC maps?
    \item Can we leverage the well-understood theory of classical matrix cones, such as $\COP_n$ and $\SPN_n$, to systematically construct elements in these new pairwise cones, similarly to the correspondence between the $\CP_n$ and the $\PCP_n$ cones? 
    \item Can such a framework be applied to generate vast, explicit families of positive indecomposable maps, thereby providing new tools for entanglement detection?
    \item Can these findings be further leveraged to construct families of entangled quantum states that nevertheless display separability-like properties, such as being PPT or possessing high-order symmetric extensions? More precisely, are there explicit examples of entangled bosonic states that are $\mathcal{K}_r$-PPT bosonic extendible? On the dual side, can we construct entanglement witnesses that admit no SOS certificate of any finite order?
    
\end{enumerate}

This paper provides affirmative answers to all these questions. We introduce a comprehensive framework for understanding positivity and decomposability of DUC/CDUC maps by defining and analyzing novel cones of matrix pairs. Our main contributions are as follows:
\begin{itemize}\setlength{\itemsep}{.5em}
    \item \underline{The Pairwise Copositive Cone}: We introduce in \cref{sec:COPCP} the cone of \emph{pairwise copositive matrices}, denoted $\COPCP_n$. Our first main result is that this cone is the dual of the $\PCP_n$ cone and provides the sought-after characterization: a DUC/CDUC map parameterized by a pair $(A,B)$ is positive if and only if $(A,B) \in \COPCP_n$. This completely resolves the problem of characterizing positivity for this class of maps, answering \boxed{{\rm Q}1}.

    \item \underline{A bridge between the classical and pairwise cones}: We establish a powerful constructive result, Theorem \ref{thm:COPCP-from-COP}, which provides a canonical way to ``lift'' any matrix $M \in \COP_n$ to a family of pairs in $\COPCP_n$. This theorem enables the transfer results about copositive matrices into the domain of operator algebra and quantum information theory. These results, discussed in \cref{sec:positivity}, and the results discussed in the next point (see \cref{sec:decomposability}) answer \boxed{{\rm Q}2}.

    \item \underline{Characterization of (in)decomposability}: We extend our analysis to decomposability by introducing the cone of \emph{pairwise decomposable matrices}, $\PDEC_n$. We show this cone precisely characterizes decomposable DUC/CDUC maps and prove an analogous constructive theorem (Theorem \ref{thm:PDEC-from-SPN}) that links $\PDEC_n$ to the tractable cone $\SPN_n$. In particular, for each matrix in $\COP_n \setminus \SPN_n$, we construct an \emph{$(n^2 - n)/2$-dimensional family of positive indecomposable linear maps} (Corollary~\ref{cor:PosIndecomp}). Interestingly, we further show that when the underlying copositive matrix is \emph{extremal}, the lifted family translates directly into a family of \emph{optimal} entanglement witnesses, yielding an $(n^2 - n)/2$-dimensional collection of “edge” objects (Corollary~\ref{cor:OptimalEW}). The gap between $\COPCP_n$ and $\PDEC_n$ thus corresponds to the domain of indecomposable positive maps, and our results bridge this gap to the one between $\COP_n$ and $\SPN_n$, which has been studied extensively in the literature. {Markov-Choi maps (discussed in \cref{sec:XJ}) correspond to LDUI maps having a special form. We characterize positivity, decomposability, and other properties of such maps in Proposition \ref{prop:equivalent-MC-positive}. In Propositions \ref{prop:MC-PDEC1} and \ref{prop:MC-PDEC2}, we further provide new conditions for identifying positive indecomposable linear maps within this family. As a consequence, in Section \ref{sec:MC1} and Examples \ref{ex:MC-PosIndecomp1}, \ref{ex:MC-PosIndecomp2}, and \ref{ex:MC-PosIndecomp3}, we present several examples of indecomposable Markov–Choi maps that significantly extend the known classes described in \cite{kye1992,cho1992,ha2003class}.
    }

    \item \underline{Positive maps from graphs}: As a primary application of our framework, we introduce in \cref{sec:graphs} a one-parameter family of linear maps $\Phi_t^G$ associated with any simple graph $G$. These maps act like the completely depolarizing channel on the diagonal of the input, and as a scaled Schur product on the off-diagonal entries of the input corresponding to edges in $G$; a similar construction has been proposed in \cite{kennedy2018composition}. We derive sharp, analytic thresholds for the parameter $t$ that separate the regimes of complete positivity, decomposability, and positivity. These thresholds are expressed in terms of fundamental graph parameters of the underlying graph $G$: the \emph{largest eigenvalue} $\lambda(G)$, the \emph{clique number} $\omega(G)$, and a new graph parameter $\sigma(G)$ which we define and relate to the $\SPN_n$ cone; see \cref{thm:graph-map-properties}. This general constructions serves as a tool for generating new positive indecomposable maps. We prove in Proposition \ref{prop:triangle-free-bipartite} that any triangle-free non-bipartite graph yields such maps; this is the case for odd cycles, see Proposition \ref{prop:cycle-graph}. More significantly, by exploiting the rich algebraic structure of rank 3 \emph{strongly regular graphs}, we derive an exact formula for the parameter $\sigma(G)$ in \cref{thm:sigma-rank3-SRG}. This allows us to prove that infinite families of strongly regular graphs, including the classical \emph{Paley graphs} (see Corollary \ref{cor:Paley}), generate positive indecomposable maps, providing a wealth of new, highly structured examples and thus answering \boxed{{\rm Q}3}.

    \item \underline{Connections to quantum state extendibility}: In the final part of the paper, we connect our results to hierarchies of quantum state extendibility. We show in Proposition \ref{prop:Ext-to-SOS}, \cref{thm:PPTExtDual,thm:Ext-from-SOS} that the well-known \emph{sum-of-squares} (SOS) hierarchies for copositive matrices ($\K_n^{(r)}$, see \cref{sec:graphs} and also \cite{deklerk2002approximation,schrijver2003comparison}) correspond directly to witnesses for different levels of PPT-bosonic extendibility of quantum states, such as (mixture of) Dicke states. In particular, such connections allow us to construct a surprisingly wide family of entanglement witnesses which do not have SOS properties of any order (\cref{thm:EWnotSOS}), thus providing counterexamples to the question dating back to \cite{DPS04}. Furthermore, on the dual side, we completely characterize the PPT bosonic extendibility of Dicke states (with respect to the complete graph $\mathcal{K}_{r+2}$) in terms of the dual set $(\mathsf{K}_n^{(r)})^{\circ}$ (\cref{thm:DickeExt}), which allows us to provide explicit examples of bosonic entangled states that are arbitrarily highly extendible whenever the local dimension $n$ is not less than $5$. This answers \boxed{{\rm Q}4}. As a byproduct, we also prove in Corollary \ref{cor:MultiDickeEnt} the existence of $r$-partite entangled bosonic (Dicke) states of which are PPT with respect to all bipartitions.
\end{itemize}

\bigskip

The paper is organized as follows. \cref{sec:preliminaries} establishes our notation and reviews the necessary preliminaries on matrix cones and positive maps. \cref{sec:COPCP} introduces the central object of study, the pairwise copositive cone $\COPCP_n$, and proves its duality with $\PCP_n$ as well as some other basic properties. \cref{sec:positivity}  establishes the main constructive theorem linking the $\COPCP_n$ cone to $\COP_n$, hence providing a characterization of positivity for large classed of (C)DUC linear maps. \cref{sec:decomposability} develops the analogous theory for the pairwise decomposable cone $\PDEC_n$. In \cref{sec:XJ} we discuss Markov-Choi linear maps and their positivity properties. \cref{sec:graphs} applies this framework to define the graph-based maps and determines the exact positivity and decomposability thresholds; we also specialize this analysis to the case of strongly regular graphs. Finally, in \cref{sec:sos-hierarchies} we put forward  the remarkable connection between hierarchies of quantum state extendibility and  hierarchies of sum-of-squares programs for membership in the $\COP_n$ cone.

\section{Preliminaries}\label{sec:preliminaries}
\subsection{Notation}

We start by defining the notation used throughout this paper. A vector $v$ is an element in either \(\mathbb{C}^n\) or \(\mathbb{R}^n\). We sometimes also use Dirac's \emph{bra-ket} notation to write vectors. In this notation, column vectors $v \in \mathbb{C}^n$ are written as kets $\ket{v}$ and their dual row vectors (conjugate transposes) $v^* \in (\mathbb{C}^n)^*$ are written as bras $\bra{v}$. The standard \emph{inner product} $v^*w= \braket{v,w}$ on $\mathbb{C}^n$ is denoted by $\langle v | w \rangle$ and the rank one matrix $vw^*$ is denoted by the \emph{outer product} $\ketbra{v}{w}$. The standard basis in $\C{n}$ is denoted by $\{ \ket{i}\}_{i\in [n]}$, where $[n]:= \{0,1,\ldots ,n-1 \}$.

We define $\M{n}$ as the set of \(n \times n\) complex matrices and $\Msa{n}:= \{A\in \M{n} : A=A^* \}$ as the set of $n\times n$ self-adjoint complex matrices, where the conjugate transpose of $A\in \M{n}$ is denoted by $A^*$. $\Mreal{n}$ and $\Mrealsa{n}$ are defined similarly for real matrices. The cone of positive semi-definite (PSD) matrices in $\M{n}$ is denoted by $\PSDC_n$, while the cone of real positive semidefinite matrices in $\Mreal{n}$ is denoted by $\PSDR_n$. The cone of entry-wise non-negative matrices by $\EWP_{n}$. The set of all linear maps $\Phi : \M{n}\to \M{n}$ is denoted by \(\mathcal{T}_n(\mathbb{C})\). In this paper, we shall overload the $\diag$ notation to denote three different things:
\begin{itemize}
    \item For a matrix $A \in \M{n}$, $\diag(A) \in \M{n}$ denotes the diagonal matrix obtained by setting the off-diagonal entries of $A$ to zero. In other words, $\diag(\cdot)$ is the conditional expectation from $n \times n$ matrices to the set of diagonal matrices.
    \item For a matrix $A \in \M{n}$, $\diag[A] \in \C{n}$ denotes the vector containing the diagonal elements of the matrix $A$.
    \item For a vector $a \in \C{n}$, $\diag\{a\} \in \M{n}$ denotes the diagonal matrix with the entries of the vector $a$ on the diagonal.
\end{itemize}

In general, the meaning of the operator can be inferred from the context, but we shall consistently use the appropriate brackets for the sake of clarity. We use ${I}_n$ to denote the identity matrix, and ${J}_n$ to denote the matrix with all entries equal to $1$. We define the vector of all ones as $\ket{\mathbf{1}_n} := \sum^n_{i=1} \ket{i}$ such that $\ketbra{\mathbf{1}_n}{\mathbf{1}_n} = {J}_n$. Finally, we introduce a specific notation that will be used frequently throughout this paper: \emph{the removal of the diagonal part} of a matrix, denoted by
    $$\mathring{A}:=A-\diag(A)=\sum_{i\neq j}A_{ij}|i\ra\la j|, \quad A\in \M{n}.$$

\subsection{Convex properties of matrices}

In this section, we discuss some important convex conic subsets of matrices. Recall that the cone of \emph{real positive-semidefinite matrices} can be characterized in two equivalent ways,

$$\PSDR_n = \operatorname{cone}\Big\{ \ketbra {v}{v}  \, : \, v \in \mathbb R^n \Big\} = \Big\{ X \in \mathcal M^\mathrm{sa}_n(\mathbb R) \, : \, \braket{v | X | v} \geq 0, \, \forall v \in \mathbb R^n \Big\}.$$ 
We use $\mathbb{R}^n_+$ to mean the set of all vectors with non-negative entries in the computational basis. Similarly, we denote by $\EWP_n$ the set of entrywise non-negative $n \times n$ matrices 
$$\EWP_n := \Big\{ X \in \mathcal M_n(\R{}) \, : \, X_{ij} \geq 0, \, \forall i,j \in [n] \Big\},$$
as well as its symmetric counterpart $\EWP_n^\sa:=\EWP_n \cap \Mrealsa{n}$.

Next, we define set of \emph{completely positive matrices} \cite{berman2003completely} 

\[\mathsf{CP}_n := \operatorname{cone}\Big\{ \ketbra v  \, : \, v \in \mathbb R^n_+ \Big\}.\]
The extremal rays of $\mathsf{CP}_n$ are easily characterized: \[\operatorname{ext}\mathsf{CP}_n = \{\mathbb{R}_+ \ketbra v \, : \, v \in \R{n}_+\}.\]
It is clear from the definition above that completely positive matrices are positive-semidefinite and entrywise positive; we call the latter set \emph{doubly non-negative matrices}. Thus we have the inclusion,
\begin{equation}\label{eq:CP-subseteq-DNN}
    \CP_n \subseteq \DNN_n:=\PSDR_n \cap \EWP_n.
\end{equation}
It is a remarkable fact that the above inclusion is an equality for $n \leq 4$ \cite[Theorem 2.4]{berman2003completely}, but it is strict for $n \geq 5$ . Below is an example of a doubly non-negative matrix that is not completely positive, see \cite[Example 2.4]{berman2003completely}:
$$X = 
\begin{bmatrix}
1 & 1 & 0 & 0 & 1 \\
1 & 2 & 1 & 0 & 0 \\
0 & 1 & 2 & 1 & 0 \\
0 & 0 & 1 & 2 & 1 \\
1 & 0 & 0 & 1 & 6
\end{bmatrix}.$$

We also introduce the set of \emph{copositive matrices} \cite{shaked2021copositive}  
$$\mathsf{COP}_n := \Big\{ X \in \mathcal M^\mathrm{sa}_n(\mathbb R) \, : \, \braket{v | X | v} \geq 0, \, \forall v \in \mathbb R^n_+ \Big\}$$ 
which is dual to the cone of completely positive matrices $\mathsf{CP}_n^\circ = \mathsf{COP}_n$. We recall here that the dual of the cone $C \subseteq \R{N}$ is the cone
$$C^\circ := \{x \in \R{N} \, : \, \braket{x,y} \geq 0, \, \forall y \in C\}.$$
The complete characterization of the extremal rays for copositive matrices is still an open question, and the complete solution has only been found up to dimension $6$ \cite{hildebrand2012extreme,afonin2021extreme}.

In particular, the cone of positive-semidefinite matrices is self-dual $\mathsf{PSD}_n^{\circ} = \mathsf{PSD}_n$, as are the cones on entrywise non-negative matrices $\EWP_n^\circ = \EWP_n$, $(\EWP_n^\sa)^\circ = \EWP_n^\sa$. The dual of doubly non-negative matrices is the set
\begin{equation}\label{eq:def-SPN}
    \SPN_n := \PSDR_n + \EWP^\sa_n.
\end{equation}
Applying duality to \cref{eq:CP-subseteq-DNN} we obtain the reversed inclusion $\COP_n \supseteq \SPN_n$, which is an equality for $n=4$ and strict for $n \geq 5$, see \cite[Chapters 2.9, 2.10]{shaked2021copositive}. The paradigmatic example of a copositive matrix that cannot be decomposed as the sum of a real positive semidefinite matrix and an entrywise (symmetric) non-negative matrix is the \emph{Horn matrix}
\begin{equation}\label{eq:Horn}
    H := 
    \begin{bmatrix}
        1 & -1 & 1 & 1 & -1 \\
        -1 & 1 & -1 & 1 & 1 \\
        1 & -1 & 1 & -1 & 1 \\
        1 & 1 & -1 & 1 & -1 \\
        -1 & 1 & 1 & -1 & 1
    \end{bmatrix}.
\end{equation}

Combining all the following observations and dualities, we obtain the following chain of inclusions of cones, where the blue arrows indicate a duality relation:
\vspace{-15pt}
$$\begin{tikzpicture}[
node distance = 5pt,
  start chain = A going right,
    inner sep = 0pt,
    outer sep = 0pt,
every node/.style = {on chain=A}                       
                        ]
% equation
\node{$\CP_n$}; % A-1
\node{$\subseteq$}; % A-2
\node{$\DNN_n$}; % A-3
\node{$\subseteq$}; % A-4
\node{$\PSDR_n,\, $}; % A-5
\node{$\EWP^\sa_n$}; % A-6
\node{$\subseteq$}; % A-7
\node{$\SPN_n$}; % A-8
\node{$\subseteq$}; % A-9
\node{$ \COP_n\,.$}; % A-10

% lines
    \begin{scope}[
        every path/.append style = {<->, blue},
                  ]
                  
\draw ([yshift=5pt]A-1.north) to [out=90,in=90]([yshift=5pt, xshift=-2pt]A-10.north);
\draw ([yshift=5pt]A-3.north) to [out=90,in=90] ([yshift=5pt]A-8.north);
\draw ([yshift=5pt, xshift=-5pt]A-5.north) to [out=135,in=45, loop, looseness=10] ([yshift=5pt, xshift=2pt]A-5.north);
\draw ([yshift=5pt, xshift=-5pt]A-6.north) to [out=135,in=45, loop, looseness=10] ([yshift=5pt, xshift=2pt]A-6.north);
    \end{scope}
\end{tikzpicture}$$

Note that all the cones above are subsets of $\Mrealsa{n}$ and are \emph{proper cones}: they are closed, have non-empty interior, and are pointed (they contain no line). 
\bigskip

We provide a generalization of the completely positive cone $\CP_n$ to matrix pairs recently introduced and studied in \cite{johnston2019pairwise,singh2021diagonal}. Consider the convex cone of \emph{pairwise completely positive matrices} \cite[Definition 3.1]{johnston2019pairwise}:

$$\mathsf{PCP}_n := \operatorname{cone}\Big\{ \big(\ketbra{v \odot \bar v}{w \odot \bar w}, \ketbra{v \odot w}{v \odot w} \big) \, : \, v,w \in \mathbb C^n \Big\}$$
where $\odot$ denotes the Hadamard (or entrywise) product of vectors: $(v \odot w)_i = v_i w_i$. This cone lives in the vector space of pairs of matrices having the same (real) diagonal (see \cite{singh2021diagonal})
$$\Mreal{n} \underset{\R{n}}{\times} \Msa{n} := \Big\{ (A,B) \in \Mreal{n} \times \Msa{n} \, : \, \diag[A] = \diag[B]\Big\}.$$
as a \emph{proper cone}. Indeed, we will see later that $(J_n+A,I_n+B)\in \PCP_n$ whenever 
    $$\|(A,B)\|_2^2:=\|A\|_2^2+\|\mathring{B}\|_2^2= \sum_{i,j}|A_{ij}|^2+\sum_{i\neq j}|B_{ij}|^2\leq 1,$$
which implies that $(J_n,I_n)\in {\rm int}(\PCP_n)$, from the correspondence with the separability of diagonal unitary invariant states (Proposition~\ref{prop:PCP-EB}) and the separability criterion in~\cite[Theorem~1]{gurvits2002largest}. Moreover, it is obvious that $\PCP_n$ cannot contain any line, proving that it is a proper cone. 

The extremal rays of the convex cone $\mathsf{PCP}_n$ have been characterized in \cite[Theorem 5.13]{singh2021diagonal}:
\begin{equation}\label{eq:ext-PCP}
    \operatorname{ext}\mathsf{PCP}_n =  \Big\{ \mathbb{R}_+\big(\ketbra{v \odot \bar v}{w \odot \bar w}, \ketbra{v \odot w}{v \odot w} \big) \, : \, v,w \in \mathbb C^n \Big\}.
\end{equation}

The connection between the sets $\CP_n$ and $\PCP_n$ is given by the following result, proven in \cite[Theorem 3.4]{johnston2019pairwise}.

\begin{proposition}\label{prop:PCP-equal}
Given a matrix $A \in \Mrealsa{n}$, the following equivalence holds: 
    \[(A,A) \in \PCP_n \iff A \in \CP_n.\]
\end{proposition}

The convex dual of the cone $\PCP_n$ will be introduced and studied in detail in \cref{sec:COPCP}, with the goal of constructing positive maps.

\subsection{Convex properties of linear maps}
In this section, we provide a background on properties of linear maps that play a significant role in entanglement theory.  

\begin{definition}

We define the following notions of linear maps,  
\begin{enumerate}
\item A linear map \(\mathcal{E} \in \mathcal{T}_n(\mathbb{C})\) is called \textit{positive} if \(\mathcal{E}(X) \in \PSDC_{n}\) for all \(X \in \PSDC_{n}\).

\item A linear map \(\mathcal{E} \in \mathcal{T}_n(\mathbb{C})\) is \textit{\(k\)-positive} if the map \(\operatorname{id}_{k} \otimes \, \mathcal{E} : \mathcal{M}_{k} \otimes \mathcal{M}_n \rightarrow \mathcal{M}_{k}\otimes \mathcal{M}_n\) is positive, where $\operatorname{id}_{k}:\M{k}\to \M{k}$ is the identity map.
\end{enumerate}
\end{definition}

A positive map is trivially $1$-positive. 

\begin{definition}
Let $\top$ be the transposition map. We define the following properties of linear maps,
\begin{enumerate}
    \item A linear map \(\mathcal{E} \in \mathcal{T}_n(\mathbb{C})\) is completely positive if it is $k$-positive for all $k \in \mathbb{N}$.
    \item A linear map \(\mathcal{E} \in \mathcal{T}_n(\mathbb{C})\) is called \textit{completely copositive} if \( \mathcal{E}\circ \top\) is completely positive, or equivalently, $\mathcal{E} = \mathcal{X}\circ \top$ where $\mathcal{X}$ is completely positive.
    \item A linear map \(\mathcal{E} \in \mathcal{T}_n(\mathbb{C})\) is called \textit{PPT}  if it is completely positive and completely copositive.
    \item  A completely positive linear map $\mathcal{E} \in \mathcal{T}_n(\mathbb{C})$ is entanglement breaking if $\mathcal{E} \otimes \operatorname{id}_n (X)$ is separable for all $X \in {\mathsf{PSD}(\Comp^n\otimes \Comp^n)}$
\end{enumerate}
\end{definition}

A linear map \(\mathcal{E} \in \mathcal{T}_n(\mathbb{C})\) that is $n$-positive is completely positive. Positive, Completely positive maps, and PPT maps are closed under composition, while the completely copositive maps are not. 

\begin{definition}
    A linear map \(\mathcal{E} \in \mathcal{T}_n(\mathbb{C})\) is said to be decomposable if it a sum of completely positive and a copositive map, i.e 
    $\mathcal{E} = \mathcal{E}_1 + \mathcal{E}_2\circ \top$ where $\mathcal{E}_1, \mathcal{E}_2$ are completely positive
\end{definition}

The paradigmatic examples of decomposable map is the identity map $\operatorname{id}$, the transposition map $\top$, and the depolarizing map $\mathcal{E}(\rho) = \operatorname{Tr}(\rho)I$, and their convex sums. The convex properties of linear maps have dual connection to the convex properties of bipartite matrices. We provide these well-known duality conditions in the next theorem \cite{horodecki1996separability,kye2023compositions}. 

\begin{theorem}
\label{thm:duality-of-bipartite-states}
    The following duality conditions are true for bipartite quantum states $\rho \in \M{n} \otimes \M{n}$ and linear maps in $\mathcal{T}_n(\mathbb{C})$.

    \begin{itemize}
        \item $\rho$ is positive-semidefinite if and only if 
        $\mathcal{E} \otimes \operatorname{id}_n{(\rho)} \geq 0$ for all completely positive maps $\mathcal{E}\in \mathcal{T}_n(\Comp)$.
        \item $\rho$ is PPT if and only if $\mathcal{E} \otimes \operatorname{id}_n{(\rho)} \geq 0$ for all decomposable maps $\mathcal{E}\in \mathcal{T}_n(\Comp)$.
        \item $\rho$ is separable if and only if $\mathcal{E} \otimes \operatorname{id}_n{(\rho)} \geq 0$ for all positive maps $\mathcal{E}\in \mathcal{T}_n(\Comp)$.
    \end{itemize}
\end{theorem}

\begin{remark}
    A linear map is positive if it is decomposable, and in $\mathcal{T}_2(\mathbb{C})$ all positive linear maps are decomposable. On the dual side, by \cref{thm:duality-of-bipartite-states}, it implies that all PPT states are separable. 
\end{remark} The convex properties of linear maps can also be understood in terms of the matrix properties with the Choi-Jamio{\l}kowski Isomorphism that provides a bijection between linear maps in matrix algebras to bipartite matrices. 

\begin{lemma}[Choi-Jamio{\l}kowski Isomorphism \cite{choi1975completely}]  \label{lemma:CJiso}
Define the linear bijection $J:\mathcal{T}_d(\mathbb{C}) \rightarrow \mathcal{M}_d(\mathbb{C}) \otimes \mathcal{M}_d(\mathbb{C})$ as $J(\Phi) = \sum_{i,j=1}^d \Phi(\vert i \rangle \langle j \vert) \otimes \vert i \rangle \langle j \vert$. Then, $\Phi\in \mathcal{T}_{d} (\mathbb{C})$ is
\begin{enumerate}
    \item hermiticity preserving if and only if $J(\Phi)$ is self-adjoint,
    \item positive if and only if $J(\Phi)$ is block positive, i.e., $\langle x \otimes y \vert J(\Phi) \vert x \otimes y \rangle \geq 0 \, \, \forall \ket{x}, \ket{y} \in \mathbb{C}^d$,
    \item completely positive if and only if $J(\Phi)$ is positive semi-definite,
    \item completely copositive if and only if $J(\Phi)^\Gamma$ is positive semi-definite,
    \item decomposable if and only if $\exists X_1, X_2  \geq 0 \quad \text{such that} \quad J(\Phi) = X_1 + X_2^\Gamma$
    \item entanglement breaking if and only if $J(\Phi)$ is separable.
\end{enumerate}
\end{lemma}
In Lemma~\ref{lemma:CJiso}, part~$(5)$ appears in \cite{stormer1982}, while part~$(6)$ is due to \cite{horodecki2003entanglement}. 

\subsection{Linear maps and matrices with diagonal unitary symmetry} \label{sec:LDUI}

We provide some background about the class of quantum maps and matrices that play an important role in this paper, because of their nice symmetry properties. Let us denote the group of diagonal unitary matrices with  $D\mathcal{U}_n\subseteq \mathcal{U}_n$, i.e the group of matrices with complex phases on the diagonal.

\[
D\mathcal{U}_n \;=\; 
\left\{ \diag\{e^{i\theta_1}, \ldots, e^{i\theta_n}\}
\;\middle|\;
\theta_1, \dots, \theta_n \in \mathbb{R}
\right\}.
\]

Now, we introduce two notions of maps covariant with the diagonal unitary symmetries,  
\begin{definition}
A linear map \(\mathcal{E} \in \mathcal{T}_n(\mathbb{C})\) is called 
\begin{enumerate}
    \item \emph{diagonal unitary covariant (DUC)} if 
    \[\forall X \in \M{n} \text{ and } \forall U \in D\mathcal{U}_n, \quad \mathcal{E}(U X U^*) = U \mathcal{E}(X)U^*\]

    \item \emph{conjugate diagonal unitary covariant (CDUC)} if 
    \[\forall X \in \M{n} \text{ and } \forall U \in D\mathcal{U}_n, \quad \mathcal{E}(U X U^*) = \overline{U} \mathcal{E}(X)U^{\top} \]
\end{enumerate}
\end{definition}

These classes of linear maps were first introduced in \cite{liu2015unitary,lopes2015generic}. In \cite{johnston2019pairwise, singh2021diagonal} , the structure of these channels was characterized completely in terms of matrix pairs, with common diagonals. 
Let $\MLDUI{n} := \{(A,B) \in \M{d}^{\times 2} \mid \operatorname{diag}(A) = \operatorname{diag}(B)\}$ denote the set of such matrix pairs.  

\begin{proposition}
The linear space of (conjugate) diagonal unitary covariant channels is isomorphic to matrix pairs, $\MLDUI{n} \cong \mathsf{DUC}_n \cong \mathsf{CDUC}_n$, with the correspondences $(A,B) \mapsto \PhiDUC{A}{B}$ and $(A,B) \mapsto \PhiCDUC{A}{B}$:
\begin{align*}
    \PhiDUC{A}{B} (X) &= \delta_A (X) + \mathring{B} \odot X,\\
    \PhiCDUC{A}{B} (X) &= \delta_A (X) + \mathring{B} \odot X^{\top}
\end{align*}
for $X\in \M{n}$, where $\delta_A (X):=\diag\{A\diag[X]\} = \sum_{i,j=1}^n A_{ij} X_{jj}|i\ra\la i|$.
\end{proposition}

\begin{remark}

Note that our naming of DUC/CDUC map is \textbf{opposite} to that of \cite[Definition 6.3]{singh2021diagonal}. Also note that, 

\begin{itemize}
    \item the DUC/CDUC maps $\PhiDUCCDUC{A}{B}$ are hermiticity preserving if and only if the matrix pair 
    $(A,B) \in \Mreal{n} \underset{\R{n}}{\times} \Msa{n}$.
    \item the notions of DUC and CDUC maps are connected by the composition with the transpose map, i.e,  
    $$\PhiDUC{A}{B} \circ \top = \PhiCDUC{A}{B}$$
\end{itemize}
\end{remark}

In relation with (C)DUC linear maps, we further introduce the following invariance properties of matrices.
\begin{definition}
A bipartite matrix $X \in \M{n} \otimes \M{n}$ is called
\begin{enumerate}
    \item  \emph{local diagonal unitary invariant (LDUI)} if 
        $$\forall U \in D\mathcal{U}_n, \quad (U \otimes U) X (U \otimes U)^* = X.$$
    \item \emph{conjugate local diagonal unitary invariant (CLDUI)} if 
        $$\forall U \in D\mathcal{U}_n, \quad (U \otimes \overline{U}) X (U \otimes \overline{U})^* = X.$$
\end{enumerate}
We denote these linear spaces by $\LDUI_n$ and $\CLDUI_n$, following the terminology of \cite{singh2021diagonal}. We then further identify $\M{n}^{\times 2}_{\Comp^n} \cong \LDUI_n \cong \CLDUI_n$ via the Choi-Jamio{\l}kowski isomorphism (Lemma \ref{lemma:CJiso}) described below (see also Proposition~\ref{prop:(C)DUC-CP-CLDUI+} for related adapted notation).

\begin{proposition} [\cite{singh2021diagonal}]
For $\Phi\in \mathcal{T}_{n} (\mathbb{C})$,
    \begin{itemize}
        \item $\Phi$ is DUC if and only if $J(\Phi)$ is CLDUI. In particular, we have 
        \begin{equation} \label{eq:CLDUI-DUC}
            J(\PhiDUC{A}{B}) = \sum_{ij}A_{ij} \ketbra{ij}{ij} + \sum_{i \neq j} B_{ij} \ketbra{ii}{jj} =: \XCLDUI{A}{B}, \quad (A,B)\in \MLDUI{n};
        \end{equation}
        
        \item $\Phi$ is CDUC if and only if $J(\Phi)$ is LDUI. In particular, we have 
        \begin{equation} \label{eq:LDUI-CDUC}
            J(\PhiCDUC{A}{B}) = \sum_{ij}A_{ij} \ketbra{ij}{ij} + \sum_{i \neq j} B_{ij} \ketbra{ij}{ji} =: \XLDUI{A}{B}, \quad (A,B)\in \MLDUI{n}.
        \end{equation}
    \end{itemize}
\end{proposition}

\end{definition}

Under the correspondence above, pairs of matrices in $\PCP_n$ have been shown in \cite[Lemmas 6.10, 6.11]{singh2021diagonal} to correspond to \emph{entanglement breaking} (conjugate) diagonal unitary covariant maps. 
\begin{proposition}\label{prop:PCP-EB}
    Given a pair of matrices $(A,B) \in \Mreal{n} \underset{\R{n}}{\times} \Msa{n}$, we have 
    $$(A,B) \in \PCP_n \iff \PhiDUCCDUC{A}{B} \text{ is entanglement breaking} \iff \XLDUICLDUI{A}{B} \text{ is separable}.$$
\end{proposition}

We end this section by recalling the following definition and result from \cite{singh2021diagonal}, characterizing the stronger notion of \emph{complete positivity} for $\mathsf{(C)DUC}$ maps. 

\begin{proposition}\label{prop:(C)DUC-CP-CLDUI+}
    Consider the following cones:
    \begin{align*}
        \LDUI^+_n &:= \Big\{ (A,B) \in \Mreal{n} \underset{\R{n}}{\times} \Msa{n} \, : \, A \in \EWP_n, \, A_{ij}A_{ji}\geq |B_{ij}|^2 \,\forall\, i\neq j \,  \Big\}.\\
        \CLDUI^+_n &:= \Big\{ (A,B) \in \Mreal{n} \underset{\R{n}}{\times} \Msa{n} \, : \, A \in \EWP_n, \, B \in \PSDC_n, \,  \Big\}.
    \end{align*}
Then we have
\begin{align*}
    (A,B) \in \LDUI_n^+ &\iff \PhiCDUC{A}{B} \text{ is completely positive} \\
    &\iff \PhiDUC{A}{B} \text{ is completely copositive} \iff \XLDUI{A}{B}\geq 0, \\
    (A,B) \in \CLDUI_n^+ &\iff \PhiDUC{A}{B} \text{ is completely positive} \\
    &\iff \PhiCDUC{A}{B} \text{ is completely copositive} \iff \XCLDUI{A}{B}\geq 0.
\end{align*}
\end{proposition}
\begin{proof}
    The inclusions follows from easily from Definition \ref{def:COPCP}, while the equivalence is proven in \cite[Lemma 2.12]{singh2021diagonal}, see also \cite[ Lemma 7.6]{nechita2021graphical} or \cite[Theorem 5.2]{johnston2019pairwise} 
\end{proof}

\begin{proposition}\label{prop:(C)DUC-PPT-PDNN}
    Consider the following cone of pairwise doubly non-negative matrices:
    $$\PDNN_n := \Big\{ (A,B) \in \Mreal{n} \underset{\R{n}}{\times} \Msa{n} \, : \, A \in \EWP_n, \, B \in \PSDC_n, \, \forall i,j \,\,  A_{ij} A_{ji} \geq |B_{ij}|^2 \Big\}.$$
Then we have
    $$(A,B) \in \PDNN_n \iff \PhiDUCCDUC{A}{B} \text{ is PPT} \iff \XLDUICLDUI{A}{B} \text{ is PPT}.$$
\end{proposition}
\begin{proof}
    The equivalence follows from \cite[Lemma 6.11,6.12]{singh2021diagonal}
\end{proof}

Finally, we introduce a subclass of LDUI states called the (mixture of) Dicke states \cite{yu2016separability,tura2018separability,singh2021diagonal}, which will be of our main interest in \cref{sec:sos-hierarchies}. Recall that the symmetric space $\Comp^n\vee \Comp^n\subset (\Comp^n)^{\otimes 2}$ has an orthonormal basis $\{|\psi_{ij}\ra\}_{1\leq i\leq j\leq n}$ (called \emph{Dicke basis}) where
    $$|\psi_{ij}\ra:=\begin{cases}
        (|ij\ra+|ji\ra)/\sqrt{2} & \text{if $i<j$,}\\
        |ii\ra & \text{if $i=j$.}
    \end{cases}$$
Bipartite positive semidefinite matrices which are diagonal in the Dicke basis 
    $$X=\sum_{1\leq i\leq j\leq n} Y_{ij}|\psi_{ij}\ra\la \psi_{ij}| \in \mathcal{B}(\Comp^n\vee \Comp^n), \quad Y_{ij}\in \Real_+,$$
are called \emph{Dicke states} (or \emph{diagonal symmetric matrices}). By taking $P\in \Mrealsa{n}$ as 
    $$P_{ij}=\begin{cases}
    Y_{ij}/2 & \text{if $i<j$},\\
    Y_{ji}/2 & \text{if $j<i$},\\
    Y_{ii} & \text{if $i=j$}, 
\end{cases}$$
we have $X=\XLDUI{P}{P}$. From \cite[Example 3.4] {singh2021diagonal}, we conclude that
\begin{align*}
    \XLDUI{P}{P} \text{ is positive semidefinite} &\iff P\in \EWP_n^{\sa},\\
    \XLDUI{P}{P} \text{ is PPT} &\iff P\in \DNN_n,\\
    \XLDUI{P}{P} \text{ is separable} &\iff P\in \CP_n.
\end{align*}
In particular, when $n\leq 4$, every bipartite Dicke state is separable if and only if PPT while when $n\geq 5$, every $P\in \DNN_n\setminus \CP_n$ gives rise to a PPT entangled Dicke state $\XLDUI{P}{P}$.

\section{Pairwise copositivity}\label{sec:COPCP}

We introduce and study in this section the set of \emph{pairwise copositive matrices}, which is the generalization of the cone of copositive matrices $\COP_n$ to the setting of pairs of matrices, in the same way as pairwise completely positive matrices $\PCP_n$ generalize completely positive matrices $\CP_n$. Recall that to a matrix $X$ we associate its diagonal-less version $\mathring X$
$$\mathring{X} = X - \diag(X) \qquad \text{ or } \qquad \mathring{X}_{ij} = \mathds{1}_{i \neq j} X_{ij}.$$

\begin{definition}\label{def:COPCP}
    Define the set of \emph{pairwise copositive matrices} as 
    $$\COPCP_n := \Big\{ (A,B) \in \Mreal{n} \underset{\R{n}}{\times} \Msa{n} \, : \, \braket{v \odot \bar v | A | w \odot \bar w} + \braket{v \odot w | \mathring{B} | v \odot w}  \geq 0, \, \forall v,w \in \mathbb C^n\Big\}.$$
\end{definition}

\begin{proposition}
    The set $\COPCP_n$ is a proper convex cone, dual to $\PCP_n$: 
    $$\COPCP_n = \PCP_n^\circ.$$
\end{proposition}
\begin{proof}
   The convexity property follows immediately from the definition. For the duality claim, let us first write down explicitly the (natural) scalar product on the vector space $\Mreal{n} \underset{\R{n}}{\times} \Msa{n}$: 

   \begin{equation}
   \label{eqn:inner-product}
    \braket{(A_1, B_1), (A_2, B_2)} = \braket{A_1, A_2} + \braket{\mathring{B_1}, \mathring{B_2}} = \braket{\mathring{A_1}, \mathring{A_2}} + \braket{B_1, B_2},
   \end{equation}
   using the Euclidean (or Hilbert-Schmidt) scalar product on $n \times n$ matrices. The duality follows now directly from the from of the extremal rays of the $\PCP_n$ cone \eqref{eq:ext-PCP} and from the fact that $\braket{\mathring{X},Y} = \braket{\mathring{X},\mathring{Y}}$ for arbitrary matrices $X,Y$. Finally, the fact that $\COPCP_n$ is a proper cone follows from the established duality and the fact that $\PCP_n$ is proper. 
\end{proof}

\begin{proposition}
\label{prop:twirling-formulas}
Given a pair of matrices $(A,B) \in \Mreal{n} \underset{\R{n}}{\times} \Msa{n}$, we have 
\begin{align}
    \quad \la vw|\XLDUI{A}{B}|vw\ra &= \la v\odot \overline{v}|A|w\odot \overline{w}\ra+\la v\odot \overline{w}|\mathring{B}|v\odot \overline{w}\ra, \label{eq:LDUI-duality}\\
    \la vw|\XCLDUI{A}{B}|vw\ra &= \la v\odot \overline{v}|A|w\odot \overline{w}\ra+\la v\odot w|\mathring{B}|v\odot w\ra, \label{eq:CLDUI-duality}
\end{align} 
for $v,w\in \Comp^n$ and $|vw\ra:=|v\otimes w\ra$.
\end{proposition}

\begin{proof}
We first consider the \textit{LDUI-twirling} 
\begin{equation} \label{eq:LDUITwirl}
    \mathcal{T}_{\LDUI}(X)=\mathcal{T}_{D\mathcal{U}_n^{\otimes 2}}(X):=\int_{D\mathcal{U}_n}(U\otimes U)X(U\otimes U)^*dU,\quad X\in \M{n}^{\otimes 2},
\end{equation}
where $dU$ denotes the normalized Haar measure on $D\mathcal{U}_n$. Then $\mathcal{T}_{\LDUI}$ can be considered as a orthogonal projection of $\M{n}^{\otimes 2}$ onto the linear space $\LDUI_n$ (we refer to \cite{singh2021diagonal,park2024universal} for general properties of twirling maps). Furthermore, we have
\begin{align*}
    \forall\,X\in \M{n}^{\otimes 2},\quad  &X\in \LDUI_n \iff \mathcal{T}_{\LDUI}(X)=X,\\
    \forall\,X,Y\in \M{n}^{\otimes 2},\quad &\Tr(X \mathcal{T}_{\LDUI}(Y))=\Tr(\mathcal{T}_{\LDUI}(X) Y)=\Tr(\mathcal{T}_{\LDUI}(X)\mathcal{T}_{\LDUI}(Y)),\\
    \forall\, v,w\in \Comp^n, \quad &\mathcal{T}_{\LDUI}(|vw\ra\la vw|)=\XLDUI{|v\odot \bar{v}\ra\la w\odot \bar{w}|}{|v\odot \bar{w}\ra\la v\odot \bar{w}|},    
\end{align*}
where the first two facts follows from \cite[Section 2.3]{park2024universal} and the last fact is from \cite[Remark 5.16]{singh2021diagonal}. Now we show the proof of the first equation by the following computations.
\begin{align*}
\langle vw \mid \XLDUI{A}{B} \mid vw \rangle 
&= \operatorname{Tr}\!\big( \, |vw\rangle \langle vw| \, \XLDUI{A}{B} \big) \\[4pt]
&= \operatorname{Tr}\!\big( \, |vw\rangle \langle vw| \, \mathcal{T}_{\mathsf{LDUI}}(\XLDUI{A}{B}) \big) \\[4pt]
&= \operatorname{Tr}\!\big( \, \mathcal{T}_{\mathsf{LDUI}}(|vw\rangle \langle vw|) \, \XLDUI{A}{B} \big) \\[4pt]
&= \operatorname{Tr}\!\Big( 
    \XLDUI{\ketbra{v \odot \overline{v}}{w \odot \overline{w}}}{|v \odot \overline{w} \rangle \langle v \odot \overline{w}|} \, \XLDUI{A}{B} 
   \Big) \\
&= \la v\odot \overline{v}|A|w\odot \overline{w}\ra+\la v\odot \overline{w}|\mathring{B}|v\odot \overline{w}\ra.
\end{align*}

The second formula can be shown similarly via the \textit{CLDUI-twirling}
\begin{equation} \label{eq:CLDUITwirl}
    \mathcal{T}_{\mathsf{CLDUI}}(X):=\int_{D\mathcal{U}_n} (U\otimes \overline{U})X(U\otimes \overline{U})^*dU, \quad X\in \M{n}^{\otimes 2},
\end{equation}
which is left to the reader.
\end{proof}

In \cite[Theorem 6.6]{singh2021diagonal} it was shown that elements in $\COPCP$ corresponds to \emph{positive} conjugate diagonal unitary covariant maps, a dual result to Proposition \ref{prop:PCP-EB}.
\begin{proposition}\label{prop:COPCP-postivie}
    Given a pair of matrices $(A,B) \in \Mreal{n} \underset{\R{n}}{\times} \Msa{n}$, we have 
    $$(A,B) \in \COPCP_n \iff \PhiDUCCDUC{A}{B} \text{ is positive} \iff \XLDUICLDUI{A}{B} \text{ is block-positive}.$$
\end{proposition}

\begin{proof}
Assume that $(A,B) \in \COPCP_n$. Then, by Proposition \ref{prop:twirling-formulas}, we get $$\forall v,w \in \C{n} \quad \la vw|\XLDUI{A}{B}|vw\ra = \la v\odot \overline{v}|A|w\odot \overline{w}\ra+\la v\odot \overline{w}|\mathring{B}|v\odot \overline{w}\ra \geq 0$$
This is equivalent to the condition that that $\XLDUI{A}{B}$ is block-positive $\iff \PhiCDUC{A}{B}$ is positive \cite{eom2000duality}. 
\end{proof}

We further collect in the following proposition some simple results about the convex cone $\COPCP$. 

\begin{proposition}\label{prop:properties-COPCP}
    For every pair $(A,B) \in \COPCP_n$, the following implications hold:
    \begin{enumerate}
        \item $A \in \mathsf{EWP}_n$.\label{itm:A-EWP}

        \item  $\forall i \neq j \in [n] \quad \sqrt{A_{ii}A_{jj}} + \sqrt{A_{ij}A_{ji}} - |B_{ij}| \geq 0$.\label{itm:A-B-ineq}
        
        \item $(A + \mathring B) + (A + \mathring B)^\top = A + A^\top + 2\operatorname{Re}(\mathring B) \in \mathsf{COP}_n$.\label{itm:A+B-COP}

        \item $(D_1^*D_1AD_2^*D_2,D_1^*D_2^*BD_1D_2)\in \COPCP_n$ for any complex diagonal matrices $D_1$ and $D_2$.\label{itm:diagPCOP}
    \end{enumerate}
    Regarding diagonal elements of the pair and the connection to the $\PCP$ cone, we have:  
    \begin{enumerate}[resume]
        \item $(A, \diag(A)) \in \COPCP_n \iff A\in \EWP_n \iff (A, \diag(A)) \in \PCP_n$.\label{itm:B-diagonal}
        
        \item $(\diag(B),B) \in \COPCP_n \iff B \in \PSDC_n$.\label{itm:A-diaogonal-COPCP}
        
        \item $(\diag(B),B) \in \PCP_n \iff B = \diag B$.\label{itm:A-diaogonal-PCP}
    \end{enumerate}
    Regarding identical elements in the pair, we have: 
    \begin{enumerate}[resume]
        \item $(A,A) \in \COPCP_n \iff A \in \EWP_n^{\sa} \iff {\XLDUI{A}{A}\geq 0}$.\label{itm:prop-COPCP-equal}
    \end{enumerate}
    
\end{proposition}

\begin{proof}
From the definition of $\COPCP_n$ we have
\begin{equation}
\label{eq:pairwise-cop}
    \braket{v \odot \bar v | A | w \odot \bar w} + \braket{v \odot w | \mathring{B} | v \odot w}  \geq 0, \, \forall v,w \in \mathbb C^n.
\end{equation} 

For \cref{itm:A-EWP,itm:A-B-ineq}, we refer to \cite[Proposition 6.7]{singh2021diagonal}. 

To show \cref{itm:A+B-COP}, we replace in \cref{eq:pairwise-cop}, $v = z, w = \bar z \in \C{n}$.We get 
\begin{align}
    &\la z\odot \bar{z}|A+\mathring{B}|z\odot \bar{z}\ra\geq 0, \quad \forall z\in \mathbb{C}^n \iff \la p|(A+\mathring{B})+(A+\mathring{B})^{\top}|p\ra\geq 0, \quad \forall p\in \mathbb{R}_+^n.
\end{align}
This implies $(A+\mathring{B})+(A+\mathring{B})^{\top}=A+A^{\top}+2{\rm Re}(\mathring{B})\in \mathsf{COP}_n$. 

\cref{itm:diagPCOP} follows from the substitution $v\mapsto D_1v= \diag[D_1]\odot v$ and $w\mapsto D_2w= \diag[D_2]\odot w$ in \cref{eq:pairwise-cop}, in which case $D_1v\odot \overline{D_1v}=D_1^*D_1(v\odot \overline{v})$ and $D_1v\odot D_2w=D_1D_2(v\odot w)$.

\cref{itm:B-diagonal} follows immediately since $\mathring{B} = 0$, while \cref{itm:A-diaogonal-COPCP} and \cref{itm:A-diaogonal-PCP} follow from $\mathring{A} = 0$.

Finally, \cref{itm:prop-COPCP-equal} follows from \cref{itm:A-EWP} and Proposition \ref{prop:(C)DUC-CP-CLDUI+}.
\end{proof}

\begin{remark}
The results regarding pairs of matrices $(A,B)$ where either $A$ or $B$ are diagonal can be understood in terms of the corresponding linear maps as follows. 
\begin{itemize}
    \item  \cref{itm:B-diagonal}: A ``classical'' map $\PhiDUCCDUC{A}{\diag(A)}(X) = \diag\{A \diag[X]\}$
is positive $\iff$ it is completely positive $\iff$ it is entanglement breaking.
    \item \cref{itm:A-diaogonal-COPCP}: A (transpose-)Schur product map $\PhiDUC{\diag(B)}{B}(X) = B \odot X$ (resp.~$\PhiCDUC{\diag(B)}{B}(X) = B \odot X^\top$) is positive $\iff$ it is completely positive (resp.~completely co-positive).
    \item \cref{itm:A-diaogonal-PCP}: A (transpose-)Schur product map $\PhiCDUC{\diag(B)}{B}(X) = B \odot X^{\top}$ is entanglement breaking $\iff$ it is ``classical'' (i.e.~only takes into account the diagonal of the input).
\end{itemize}
\end{remark}

\begin{remark}
    We have seen in Proposition \ref{prop:PCP-equal} that $(A,A)\in \PCP_n \iff A\in \CP_n$. From \cref{itm:prop-COPCP-equal} above, it follows that a similar property does not hold for $\COPCP_n$ and $\COP_n$. In the language of linear maps, the result from \cref{itm:prop-COPCP-equal} translates to the fact that the linear $\mathsf{DUC}$ map
    \[\PhiDUC{A}{A}(X) = \operatorname{diag}\{A \operatorname{diag}[X]\}  + \mathring{A} \odot X\] 
    is positive if and only if it is completely positive.
\end{remark}

Recall that $\PhiDUC{A}{B}$ (resp. $\PhiDUC{A}{B}$) is completely positive (resp.~completely copositive), if $(A,B) \in \CLDUI^+_n$. An important example of element in the $\COPCP$ cone, different that the ones one might construct using Proposition \ref{prop:(C)DUC-CP-CLDUI+}, is given below. 

\begin{example}\label{ex:J-J}
We have $(\mathring{N},-\mathring{N})\in \COPCP_n$ for any $N\in \EWP_n^{\sa}$. Indeed, for $v,w\in \Comp^n$, one has
\begin{align*}
    \braket{v \odot \bar v | \mathring N | w \odot \bar w} - \braket{v \odot w | \mathring N | v \odot w} &= \sum_{i \neq j} N_{ij} |v_i|^2|w_j|^2 - \sum_{i \neq j} N_{ij} \bar v_i \bar w_i v_jw_j\\
    &=\sum_{i < j} N_{ij}(|v_i|^2|w_j|^2+|v_j|^2|w_i|^2-\bar v_i \bar w_i v_jw_j -\bar{v}_j\bar{w}_j v_iw_i)\\
    &= \sum_{i<j}N_{ij}|v_i\bar{w}_j-v_j\bar{w}_i|^2 \geq 0.
\end{align*}
In particular, $(\mathring J_n, - \mathring J_n) \in \COPCP_n$. Note that since $-\mathring N \notin \PSDC_n$ unless $\mathring{N}=0$, we have 
$$(\mathring N, - \mathring N) \in \COPCP_n \setminus \CLDUI^+_n.$$
\end{example}

This example is the main motivation for the rest of the paper: we shall consider perturbations of the $\mathring N$ matrix appearing the first and in the second slot of the pair, yielding interesting examples of $\COPCP$ pairs, and in effect, positive maps.

\medskip

We complete this section by providing the characterizations of pairwise copositivity in terms of copositivity, which generalize \cref{itm:A+B-COP} of Proposition \ref{prop:properties-COPCP}. First, for $A\in \EWP_n$, let us denote by
    $$\widetilde{A}:=(A\odot A^{\top})^{\odot 1/2}=(\sqrt{A_{ij}A_{ji}})_{1\leq i,j\leq n}.$$
\begin{theorem} \label{thm:PCOP-from-COP1}
If a pair $(A,B) \in \Mreal{n} \underset{\R{n}}{\times} \Msa{n}$ satisfies $A\in \EWP_n$ and $\widetilde{A}+{\rm Re}(U^* \mathring{B}U)\in \COP_n$ for every $U\in D\mathcal{U}_n$, then $(A,B)\in \COPCP_n$. In particular, for $A\in \Mrealsa{n}$ \emph{real symmetric},
    $$(A,B)\in \COPCP_n \iff A\in \EWP_n \text{ and }  A+ {\rm Re}(U^* \mathring{B}U)\in \COP_n\;\; \forall\, U\in D\mathcal{U}_n$$
\end{theorem}
\begin{proof}
For $v,w\in \Comp^n$ , let us write $v\odot w = Up$ where $p\in \Real_+^n$ and $U\in D\mathcal{U}_n$. Then we have
\begin{align*}
    \la v \odot \bar v| A\big|w \odot \bar w\ra+\la v\odot w|\mathring{B}|v\odot w\ra 
    &= \sum_{i,j}A_{ij} |v_i|^2|w_j|^2+\la p|U^*\mathring{B}U|p\ra\\
    &= \sum_{i<j}(A_{ij}p_j^2\frac{|v_i|^2}{|v_j|^2}+A_{ji}p_i^2\frac{|v_j|^2}{|v_i|^2}) + \la p|U^*BU|p\ra\\
    & \geq \sum_{i<j}2\sqrt{A_{ij}A_{ji}}p_ip_j + \la p|U^*BU|p\ra \\
    &= \la p \big| \widetilde{A} + {\rm Re}(U^* \mathring{B}U) \big|p\ra \geq 0,
\end{align*}
where in the second inequality, we use that $\diag(A)=\diag(B)=U^*\diag(B)U$ and in the next step the AM-GM inequality. This shows the first assertion. 

The direction ($\Longleftarrow$) of the second assertion now follows from the first by noting that $\widetilde{A} = A$ since $A$ is a real symmetric matrix. For the other direction ($\Longrightarrow$), note that $(A,U^*BU)\in \COPCP_n$ for all $U\in D\mathcal{U}_n$ from Proposition \ref{prop:properties-COPCP}, \cref{itm:diagPCOP}, with $D_1=U$ and $D_2=I_n$. Therefore, we conclude by applying Proposition \ref{prop:properties-COPCP}, \cref{itm:A+B-COP}. 
\end{proof}

The corollary below would play a key role in the next section. Let us first denote by $D\mathcal{O}_n:=\big\{O\in D\mathcal{U}_n: \diag[O]\in \{\pm 1\}^n\big\}$ the finite group of \emph{diagonal orthogonal matrices}.

\begin{corollary} \label{cor:PCOP-from-COP2}
If $(A,B)$ satisfies $A\in \EWP_n$, $B\in \Mrealsa{n}$, and $\widetilde{A}+ O \mathring{B}O\in \COP_n$ for every $O\in D\mathcal{O}_n$, then $(A,B)\in \COPCP_n$. In particular, for $A,B\in \Mrealsa{n}$ \emph{real symmetric},
    $$(A,B)\in \COPCP_n \iff A\in \EWP_n \text{ and } A+ O \mathring{B} O\in \COP_n\;\; \forall\, O\in D\mathcal{O}_n.$$
\end{corollary}
\begin{proof}
We only need to show the first assertion. First of all, for all $D$ real diagonal, we can write $D=OP=PO$ for some $P$ a nonnegative diagonal matrix and $O\in D\mathcal{O}_n$. Then $|D|:=\sqrt{D^*D}=P$ and hence the assumption implies that

\begin{equation}
\label{eqn:DAD-COP}
       |D|\widetilde{A}|D|+D\mathring{B}D=P(\widetilde{A}+O\mathring{B}O)P\in \COP_n.
\end{equation}
 
Now for $U\in D\mathcal{U}_n$, let us write $U=D_1+iD_2\in D\mathcal{U}_n$ with $D_1, D_2$ real diagonal (hence $D_1^2+D_2^2=I_n$). We claim that for all $A \in \mathsf{EWP}_n$, we have $\widetilde{A}-|D_1|\widetilde{A}|D_1|-|D_2|\widetilde{A}|D_2|\in \EWP_n^{\sa}$. 
To show this, we use the AM-GM inequality and the fact that $D_1^2+D_2^2 = I_n$:
\begin{equation}\label{eq:DiiDjj}
    \forall i,j \in [n], \qquad |D_1|_{ii} |D_1|_{jj} + |D_2|_{ii} |D_2|_{jj} \leq \frac{|D_1|_{ii}^2 + |D_1|_{jj}^2}{2} + \frac{|D_2|_{ii}^2 + |D_2|_{jj}^2}{2} = 1.
\end{equation}
Then, since $B$ is real, we have
\begin{align*}
    \widetilde{A}+{\rm Re}(U^*\mathring{B}U) &= \widetilde{A} +D_1\mathring{B}D_1+D_2\mathring{B}D_2 \\ 
    &= \underbrace{\big(|D_1|\widetilde{A}|D_1|+D_1 \mathring{B} D_1\big)}_{\mathsf{COP}_n \text{ by \cref{eqn:DAD-COP}}} + \underbrace{\big(|D_2|\widetilde{A}|D_2|+D_2 \mathring{B}D_2\big)}_{\mathsf{COP}_n \text{ by \cref{eqn:DAD-COP}}} + \underbrace{\big(\widetilde{A}-|D_1|\widetilde{A}|D_1|-|D_2|\widetilde{A}|D_2|)}_{\EWP_n^{\sa} \text{ by \cref{eq:DiiDjj}}}
\end{align*}
This shows that $\widetilde{A}+{\rm Re}(U^*\mathring{B}U)$ is copositive, and then the first claim follows now from \cref{thm:PCOP-from-COP1}.  
\end{proof}

\section{Lifting copositivity to pairwise copositivity}\label{sec:positivity}

In this section, we study the membership problem for the cone $\COPCP$ in the case of several classes of pairs of matrices. Many of the results in this section are focused on establishing an equivalence with the membership problem for the cone of copositive matrices $\COP$. Another motivation is the generalization of Example \ref{ex:J-J} to pairs where we perturb the second component by small amount of matrix, which turns out to be copositive.

More precisely, we apply Corollary \ref{cor:PCOP-from-COP2} to obtain the following main result of this section, an efficient method to generate a wide family of pairs in $\COPCP_n$ lifted from a \emph{single} copositive matrix.

\begin{theorem} \label{thm:COPCP-from-COP}
Suppose $M,N\in \Mrealsa{n}$ satisfy $\diag(N)=\diag(M)$, $N\in \EWP_n^{\sa}$, and $N_{ij}\geq \frac{1}{2}M_{ij}$ for all $i\neq j\in [n]$. Then 
$$M\in \COP_n \iff (N,M-\mathring{N})\in \COPCP_n.$$
In particular, denoting $(M_+)_{ij}:=\max(M_{ij},0)$ and $(M_-)_{ij}:=-\min(M_{ij},0)$, we have 
\begin{align*}
    M\in \COP_n &\iff \frac{1}{2} (M_+ +\diag(M), M_+-2M_-+\diag(M))\in \COPCP_n \\
    &\iff (M_+,\diag(M)-M_-) \in \COPCP_n.
\end{align*}
In a similar vein, for all $\lambda \geq \frac{1}{2}\max(0,\max(\mathring M))$, we have 
$$M \in \COP_n \iff (\diag(M) + \lambda\mathring J_n, M - \lambda\mathring J_n) \in \COPCP_n.$$
\end{theorem}
\begin{proof}
The direction ($\Longleftarrow$) of the first statement follows from \cref{itm:A+B-COP} from Proposition \ref{prop:properties-COPCP} as 
$$N + N^{\top} + 2 \operatorname{Re}(\mathring M- \mathring N) = 2 M \in \mathsf{COP}_n.$$
For the direction ($\Longrightarrow$), let $O\in D\mathcal{O}_n$ and let us decompose $O=O_+-O_-$ such that both $O_{\pm}$ are diagonal matrices with entries either $0$ or $1$ and $O_+ O_-=0$. Also, note that $O_+ + O_{-} = I_n$. Then since $2N-M\in \EWP_n^{\sa}$, we have the following decomposition into sum of copositive matrices
\begin{align*}
    N+O(\mathring{M}-\mathring{N})O&=(N-ONO)+OMO = 2(O_+NO_-+O_-NO_+)+OMO \\
    &= \underbrace{\big(O_+(2N-M)O_-+O_-(2N-M)O_+\big)}_{\EWP_n^{\sa}\subseteq \COP_n} + \underbrace{O_+MO_+}_{\COP_n} + \underbrace{O_-MO_-}_{\COP_n} \in \COP_n.
\end{align*}
Since this is true for all $O \in D\mathcal{O}_n$, we have by Corollary \ref{cor:PCOP-from-COP2} that $(N,M-\mathring{N})\in \COPCP_n$.

Now the second and third assertions simply follows from the first statement with $N=\frac{1}{2}(M_++\diag(M))$, $N=M_+$ (note that $\diag(M_+)=\diag(M)$, $\diag(M_-)=0$, and $M=M_+-M_-$), and $N=\diag(M) + \lambda\mathring J$, respectively.
\end{proof}

\begin{remark}
The first statement of \cref{thm:COPCP-from-COP} can be interpreted as follows: for any fixed copositive matrix $M$, one can generate the family
    $$\mathcal{F}_M=\big\{(N,M-\mathring{N}): N\in \EWP_n^{\sa},\;\diag(N)=\diag(M), \text{ and } N_{ij}\geq \frac{1}{2}M_{ij}\; \forall\, i\neq j\big\}$$
of pairs in $\COPCP_n$. In other words, any pair $(N,M-\mathring{N})\in \mathcal{F}_M$ defines positive linear maps $\PhiDUCCDUC{N}{M-\mathring{N}}$. Note that this family is both simple and rich, in the sense that it allows for \emph{$(n^2-n)/2$ degrees of freedom} in the choice of the upper off-diagonal entries $(N_{ij})_{1\leq i<j\leq n}$. In the remainder of the paper, we will show that this construction indeed generates a broad class of \emph{positive indecomposable} linear maps and \emph{entanglement witnesses} capable of detecting entanglement in PPT states that are arbitrarily close to separable ones.
\end{remark}

\begin{example} \label{ex:HornPCOP1}
For the Horn matrix $H$ defined in \cref{eq:Horn}, the last assertion in \cref{thm:COPCP-from-COP} implies that 
\begin{center}
    $(I_5+\lambda \mathring{J}_5,H-\lambda \mathring{J}_5)\in \COPCP_5$ whenever $\lambda\geq \frac{1}{2}$.
\end{center}
We show that the condition $\lambda\geq \frac{1}{2}$ is also necessary (see also Example \ref{ex:HornPCOP2}), which implies that the lower bound $\lambda\geq \frac{1}{2}\max(0,\max(\mathring M))$ in the previous theorem is \emph{tight} as a multiple of $\max(0,\max(\mathring M))$. Indeed, for $t\geq 0$ we may choose $v=(1,0,t,t,0)^{\top}$ and $w=(-1,0,t,t,0)^{\top}$ to compute that
    $$\la v^{\odot 2}\big| I_5+\lambda \mathring{J}_5\big| w^{\odot 2}\ra+\la v\odot w|\mathring{H}-\lambda\mathring{J}_5|v\odot w\ra=1+4t^2(2\lambda-1).$$
Thus, the condition $(I_5+\lambda \mathring{J}_5,H-\lambda \mathring{J}_5)\in \COPCP_5$ implies that $2\lambda-1\geq -\frac{1}{4t^2}$ for all $t>0$. By taking $t\to \infty$, we obtain $\lambda\geq 1/2$.
\end{example}

On the other hand, the bound $\lambda \geq \frac{1}{2}\max(0,\max(\mathring M))$ from \cref{thm:COPCP-from-COP} is \emph{not tight} in another perspective: if $M\in \PSDR_n$, then $(\operatorname{diag}(M) + \lambda\mathring J, M - \lambda \mathring J) \in \COPCP_n$ for all $\lambda\geq 0$ (see \cref{itm:A-diaogonal-COPCP} in Proposition \ref{prop:properties-COPCP}) while $\max{(\mathring M)}>0$ whenever $M$ has positive off-diagonal elements.

\medskip

The simple corollary below will be used later in \cref{sec:graphs} to establish the positivity of linear maps defined by graphs.  
\begin{corollary}\label{cor:diagA-A}
Let $A \in \Mrealsa{n}$ such that $I-A\in \EWP_n$. Then 
    $$(\diag(A)+\mathring J_n , A)\in \COPCP_n \quad \iff \quad A+\mathring J_n \in \COP_n.$$
\end{corollary}

\section{Pairwise decomposability} \label{sec:decomposability}
Between the class of positive linear maps and the one of completely positive maps lies the important intermediate class of decomposable maps. We study in this section the corresponding cone of pairs of matrices and connect it to decomposable maps in operator algebra. Recall that a linear map \(\mathcal{E} \in \mathcal{T}_n(\mathbb{C})\) is called decomposable if we can write $\mathcal{E} = \mathcal{E}_1 + \mathcal{E}_2\circ {\top}$, for some CP maps $\mathcal{E}_1$ and $\mathcal{E}_2$. To study decomposable diagonal unitary covariant linear maps, we introduce the following set:

\begin{definition}\label{def:PDEC}
    Define the set of \emph{pairwise decomposable matrices} $\PDEC_n$ as the set of pairs $(A,B) \in \Mreal{n} \underset{\R{n}}{\times} \Msa{n}$ having the following properties: 
    \begin{itemize}
        \item $A \in \EWP_n$
        \item $B = B^{(1)} + B^{(2)}$ with $B^{(1)} \in \PSDC_n$ and $B^{(2)} \in \Msa{n}$ such that $B^{(2)}_{ii} \geq 0$ for all $i \in [n]$ and $|B^{(2)}_{ij}|^2 \leq A_{ij}A_{ji}$ for all $i \neq j \in [n]$.\label{item:conditions-dec}
    \end{itemize}
\end{definition}

Note from the definition that
\begin{equation} \label{eq:SymmPDEC}
    (A,B)\in \PDEC_n \iff (\widetilde{A},B)=\big((A\odot A^{\top})^{\odot 1/2},B\big)\in \PDEC_n,
\end{equation}
which is not true in general for $(A,B)\in \COPCP_n$ (we refer to Section \ref{sec:XJ}).

We first characterize the decomposability property of (conjugate) diagonal unitary covariant maps. 

\begin{proposition} \label{prop:PairDEC}
The following statements are equivalent:
\begin{enumerate}
    \item $\PhiDUCCDUC{A}{B}$ is decomposable.
    \item $(A,B) \in \PDEC_n$.
\end{enumerate}

\end{proposition}
\begin{proof}
Since $\PhiDUC{A}{B}\circ \top=\PhiCDUC{A}{B}$ and since decomposability is invariant under the composition with transposition, we may only consider $\PhiDUC{A}{B}$.

\textbf{((1)$\Rightarrow$(2))} Assume $\PhiDUC{A}{B}$ is a decomposable linear map, 
$\PhiDUC{A}{B} = \mathcal{E}_1  + \mathcal{E}_2$ where $\mathcal{E}_1, \mathcal{E}_2$ are (resp.) completely positive and completely copositive maps. We use the DUC-twirling map 
    $$\mathcal{T}_{\mathsf{DUC}}(\mathcal{E})(\cdot) := \int_{\mathcal{DU}_n} U^* \mathcal{E}(U \cdot U^*) U\, dU
    $$ 
on both sides, and we have by \cite[Proposition 2.1]{park2024universal},
\begin{align}
    \mathcal{T}_{\mathsf{DUC}}(\PhiDUC{A}{B}) = \mathcal{T}_{\mathsf{DUC}}(\mathcal{E}_1)  + \mathcal{T}_{\mathsf{DUC}}(\mathcal{E}_2) \implies 
    \PhiDUC{A}{B} =\PhiDUC{A_1}{B_1}+\PhiDUC{A_2}{B_2}
\end{align}
for CP map $\PhiDUC{A_1}{B_1}$ and a co-CP map $\PhiDUC{A_2}{B_2}$. Now it is straightforward to verify that $B^{(i)}=B_i$, $i=1,2$, satisfy the conditions above.
For the reverse direction \textbf{((2)$\Rightarrow$(1))}, we construct a decomposition $$(A,B)=(\diag{B^{(1)}},B^{(1)})+(A-\diag{B^{(1)}},B^{(2)})$$ which gives us a decomposition into completely positive and a completely copositive map.
\end{proof}

Below, the implications in the first row can be either verified by hand from the definitions, or obtained from Proposition \ref{prop:COPCP-postivie}, \ref{prop:(C)DUC-CP-CLDUI+}, and \ref{prop:PairDEC} from the obvious implications in the second row. 

\begin{align*}
    (A,B) \in \PCP_n & \implies & \begin{cases} (A,B) \in \LDUI^+_n \text{ or} \\ (A,B) \in \CLDUI^+_n \end{cases} & \implies & (A,B) \in \PDEC_n & \implies & (A,B) \in \COPCP_n \\
    \Updownarrow\qquad &  & \Updownarrow\qquad &  & \Updownarrow\qquad &  & \Updownarrow\qquad \\
    \PhiDUCCDUC{A}{B} \text{ is EB} & \implies & \begin{cases} \PhiCDUC{A}{B} \text{ is CP or} \\ \PhiDUC{A}{B} \text{ is CP} \end{cases} & \implies & \PhiDUCCDUC{A}{B} \text{ is dec.} & \implies & \PhiDUCCDUC{A}{B} \text{ is positive.}
\end{align*}
\medskip

We start with an example directly strengthening Example \ref{ex:J-J}, obtained by taking $A = \mathring N$, $B^{(1)} = 0$, and $B^{(2)} = \pm \mathring N$ in Definition \ref{def:PDEC}.

\begin{example}\label{ex:N-N}
    For any matrix $N \in \EWP^\sa_n$, we have
    $(\mathring N, \pm \mathring N) \in \PDEC_n$.
\end{example}

We rephrase next \cite[Proposition 6.8]{singh2021diagonal} to provide a useful necessary condition for membership in the cone $\PDEC_n$.  

\begin{proposition} \label{prop-SN21COPCP} Let $A \in \Mreal{n}$ and $B \in \Msa{n}.$ Then, we have
    $$\forall i \neq j \in [n] \qquad \frac{1}{n-1}\sqrt{A_{ii}A_{jj}} + \sqrt{A_{ij}A_{ji}} - |B_{ij}| \geq 0 \implies (A,B) \in \PDEC_n$$
\end{proposition}

Recall now the definition of the $\SPN$ cone: these are real matrices that can be written as a sum of a real positive semidefinite matrix and a symmetric entrywise non-negative matrix: 
\begin{equation*}
    \SPN_n= \PSDR_n + \EWP_n^\sa,
\end{equation*}

We have the following analogue of \cref{itm:A+B-COP} of Proposition \ref{prop:properties-COPCP}.

\begin{proposition} \label{prop:SPN-from-PairDEC}
If $(A,B)\in \PDEC_n$, then $\widetilde{A} + {\rm Re}(\mathring{B}) \in \SPN_n$.
\end{proposition}
\begin{proof}
We can write $B=B^{(1)}+B^{(2)}$ satisfying the conditions listed in Definition \ref{def:PDEC}. Then, since $|\Re(B_{ij}^{(2)})|\leq \sqrt{A_{ij}A_{ji}}=\widetilde{A}_{ij}$ for all $i\neq j\in [n]$, we can show the following inclusion: 

    $$\widetilde{A} + {\rm Re}(\mathring{B})= \underbrace{(\widetilde{A}+\Re(\mathring{B}^{(2)}))}_{\EWP_n^\sa} + \underbrace{\Re(\mathring{B}^{(1)})}_{\PSDR_n}\in \SPN_n.$$
\end{proof}
Note that the condition $\widetilde{A} + {\rm Re}(\mathring{B}) \in \SPN_n$ above also implies the direct analogue of Proposition \ref{prop:properties-COPCP}, \cref{itm:A+B-COP}:
    $$(A+\mathring{B})+(A+\mathring{B})^{\top}=A+A^{\top}+2\Re(\mathring{B})\in \SPN_n,$$
since $(A+A^{\top})/2-\widetilde{A}\in \EWP_n^{\sa}$.

We now establish the main result of this section, a theorem that should be compared with the ``lifting" theorem for copositivity, \cref{thm:COPCP-from-COP} from the previous section. 

\begin{theorem}\label{thm:PDEC-from-SPN}
    Let $M\in \Mrealsa{n}$ and let $N\in \EWP_n^\sa$ with $\diag(M)=\diag(N)$, and suppose $N_{ij}\geq \frac{1}{2}M_{ij}+\frac{1}{4}(M_{ii}+M_{jj})$ for all $i\neq j \in [n]$. Then, we have 
    $$M\in \SPN_n \quad \iff \quad (N,M-\mathring{N})\in \PDEC_n.$$
    Furthermore, if $M$ is positive semidefinite (resp.~entrywise non-negative), then the condition $N\in \EWP_n^\sa$ (resp.~$N\in \EWP_n^\sa$ such that $N-M/2\in \EWP_n$) is enough to obtain $(N,M-\mathring{N})\in \PDEC_n$.
\end{theorem}
\begin{proof}
The condition $(N,M-\mathring{N})\in \PDEC_n$ implies $M\in \SPN_n$ from Proposition \ref{prop:SPN-from-PairDEC}. For the converse, let us assume $M\in \SPN_n$ and write $M=A^{(1)}+A^{(2)}$ where $A^{(1)}\in \PSDR_n$ and $A^{(2)}\in \EWP_n^\sa$. Then
\begin{itemize}
    \item $(\diag (A^{(1)}),A^{(1)})\in \PDEC_n$ and $(\mathring{N},-\mathring{N}) \in \PDEC_n$ whenever $N\in \EWP_n^\sa$;
    \item $(\diag(A^{(2)})+\mathring{N},A^{(2)}-\mathring{N}) \in \PDEC_n$ whenever $N-A^{(2)}/2\in \EWP_n$. Indeed, $N_{ij}^2\geq (A_{ij}^{(2)}-N_{ij})^2$ holds whenever $N_{ij}\geq A_{ij}^{(2)}/2$, and hence the decomposition $A^{(2)}-\mathring{N}=B^{(1)}+B^{(2)}$ with $B^{(1)}=0$ satisfies the conditions in Definition \ref{def:PDEC}.
\end{itemize}
Consequently, we have 
    $$(N,M-\mathring{N})=(\diag(A^{(1)}),A^{(1)})+(\diag(A^{(2)})+\mathring{N},A^{(2)}-\mathring{N}),\in \PDEC_n$$
if $N- A^{(2)}/2\in \EWP_n$. On the other hand, if $i\neq j\in [n]$, then
\begin{align*}
    A_{ij}^{(2)} \leq M_{ij}+|A_{ij}^{(1)}|\leq M_{ij}+\sqrt{A_{ii}^{(1)}A_{jj}^{(1)}} \leq M_{ij}+\frac{A_{ii}^{(1)}+A_{jj}^{(1)}}{2}\leq M_{ij}+\frac{M_{ii}+M_{jj}}{2}.
\end{align*}
Therefore, $N- A^{(2)}/2\in \EWP_n$ is satisfied whenever $N_{ij}\geq \frac{1}{2}M_{ij}+\frac{1}{4}(M_{ii}+M_{jj})$ for all $i\neq j$.
\end{proof}

\begin{remark} \label{rmk:MATS21}
The proof above actually shows that, if $M=A^{(1)}+A^{(2)}\in \SPN_n$ with $A^{(1)}\in \PSDR_n$ and $A^{(2)}\in \EWP_n^{\sa}$, then
    $$\big(\diag(A^{(1)}),A^{(1)}\big)+\frac{1}{2}\big(A^{(2)}+\diag(A^{(2)}),A^{(2)}+\diag(A^{(2)})\big)\in \PDEC_n,$$
by choosing $N=\diag(M)+\frac{1}{2}\mathring{A}^{(2)}$. This precisely recovers \cite[Theorem 2.2]{marconi2021entangled}:
    $$W=\XLDUI{\diag(A^{(1)})}{A^{(1)}}+\Pi_S\XLDUI{\diag(A^{(2)})}{A^{(2)}}\Pi_S \in \PPT(\Comp^n\otimes \Comp^n)^{\circ}$$
(i.e. $W$ is \emph{decomposable}), where $\Pi_S:=\frac{1}{2}(I_{d^2}+F_d)=\frac{1}{2}(I_{d^2}+\sum_{i,j=1}^d |ij\ra\la ji|)$ is the projection onto the symmetric subspace $\Comp^n\vee \Comp^n\subset (\Comp^n)^{\otimes 2}$. Indeed, the relations from \cite[Proposition 4.3]{singh2021diagonal} imply that
    $$\Pi_S \XLDUI{A}{B} \Pi_S = \XLDUI{(A+B)/2}{(A+B)/2}$$
whenever $A,B \in \Mrealsa{n}$ with $\diag(A)=\diag(B)$.
\end{remark}

\begin{example} \label{ex:HornPCOP2}
For the matrix
    $$\widetilde{H}:= \begin{bmatrix}
    1 & -1 & 1 & 1 \\
    -1 & 1 & -1 & 1 \\
    1 & -1 & 1 & -1 \\
    1 & 1 & -1 & 1 \\
    \end{bmatrix}= |v\ra \la v|+\begin{bmatrix}
    0 & 0 & 0 & 2 \\
    0 & 0 & 0 & 0 \\
    0 & 0 & 0 & 0 \\
    2 & 0 & 0 & 0 \\
    \end{bmatrix}\in \SPN_4, \quad v=(1,-1,1,-1)^{\top},$$
\cref{thm:PDEC-from-SPN} implies that $(I_4+\lambda \mathring{J}_4,\widetilde{H}-\lambda\mathring{J}_4)\in \PDEC_4$ if $\lambda\geq 1$. We claim that this condition is also necessary. Note that $\widetilde{H}$ is a principal $4\times 4$ submatrix of the Horn matrix $H$, and thus the argument in Example \ref{ex:HornPCOP1} implies
\begin{align*}
    (I_5+\lambda \mathring{J}_5,{H}-\lambda\mathring{J}_5)\in \COPCP_5\;\; \forall\,\lambda\geq \frac{1}{2} \implies  (I_4+\lambda \mathring{J}_4,\widetilde{H}-\lambda\mathring{J}_4) \in \COPCP_4  \;\; \forall\,\lambda\geq \frac{1}{2}.
\end{align*}
In particular, we have $(I_4+\lambda \mathring{J}_4,\widetilde{H}-\lambda\mathring{J}_4) \in \COPCP_4 \setminus \PDEC_4$ in the interval $\lambda\in [1/2,1)$, which leads to the \emph{positive indecomposable linear maps} $\PhiDUCCDUC{I_4+\lambda\mathring{J}_4}{\widetilde{H}-\lambda \mathring{J}_4}$. 

To show the claim, suppose $\widetilde{H}-\lambda\mathring{J}_4 = B^{(1)}+B^{(2)}$ where $B^{(1,2)}$ satisfy the conditions in Definition \ref{def:PDEC}, i.e., $B^{(2)}\in \Msa{4}$, $B_{ii}^{(2)}\geq 0$, $|B^{(2)}_{ij}|\leq \lambda$, for $i\neq j$, and 
    $$B^{(1)}=\begin{bmatrix}
    1 & -1-\lambda & 1-\lambda & 1-\lambda \\
    -1-\lambda & 1 & -1-\lambda & 1-\lambda \\
    1-\lambda & -1-\lambda & 1 & -1-\lambda \\
    1-\lambda & 1-\lambda & -1-\lambda & 1 \\
    \end{bmatrix} - B^{(2)}\in \PSDC_4.$$
First, the psd condition for each of $2\times 2$ principal submatrices of $B^{(1)}$ implies that $\diag(B^{(2)})=0$ and $B^{(2)}_{ij}=-\lambda$ for all $i,j\in [4]$ with $|i-j|=1$. Next, the condition for $3\times 3$ principal submatrix 
    $$\begin{bmatrix}
    1 & -1 & 1-\lambda-B^{(2)}_{13} \\
    -1 & 1 & -1 \\
    1-\lambda-\overline{B^{(2)}_{13}} & -1 & 1
    \end{bmatrix} \geq 0 \implies \begin{vmatrix}
    1 & -1 & 1-\lambda-B^{(2)}_{13} \\
    -1 & 1 & -1 \\
    1-\lambda-\overline{B^{(2)}_{13}} & -1 & 1
    \end{vmatrix}=-|\lambda+ B^{(2)}_{13}|^2\geq 0$$
forces that $B^{(2)}_{13}=B^{(2)}_{31}=-\lambda$, and similarly $B^{(2)}_{24}=B^{(2)}_{42}=-\lambda$. Finally, the last condition $B^{(1)}\geq 0$ would force that $1-\lambda-B^{(2)}_{14}=-1$, i.e., $B^{(2)}_{14}=2-\lambda$. Since $|B^{(2)}_{14}|\leq \lambda$, this implies $\lambda \geq 1$.
\end{example}

Finally, the following corollary allows us to provide a \emph{$(n^2-n)/2$-dimensional region} of positive indecomposable linear maps within the classes of DUC and CDUC linear maps.

\begin{corollary} \label{cor:PosIndecomp}
If $M\in \COP_n\setminus \SPN_n$, then both the maps $\PhiDUCCDUC{N}{M - \mathring{N}}$ are \emph{positive and non-decomposable} for all $N\in \EWP_n^\sa$ such that $\diag(N)=\diag(M)$ and $N_{ij}\geq \frac{1}{2}M_{ij}$ for all $i > j$. Moreover, if $P\in \DNN_n$ with $\la P,M\ra < 0$, then the Dicke state $\XLDUI{P}{P}$ is PPT entangled with 
    $$\PhiCDUC{N}{M - \mathring{N}}(\XLDUI{P}{P}) \not\geq 0.$$
In other words, the entanglement in the state $\XLDUI{P}{P}$ can be detected by the witness map $\PhiDUC{N}{M - \mathring{N}}$.
\end{corollary}

In \cref{sec:DPS}, we strengthen Corollary \ref{cor:PosIndecomp} by providing a broad family of \emph{entanglement witnesses} capable of detecting entanglement in quantum states with arbitrarily high levels of \emph{PPT extendibility}.

As a final remark, it is natural to ask which is the most efficient among these large classes of entanglement witnesses. In this context, \cite{lewenstein2000optimization} introduced the notion of \emph{optimality}. A nonzero entanglement witness $W\in {\rm Sep}(\Comp^n\otimes \Comp^n)^{\circ}$ is called \emph{optimal} if it detects the maximal set of entangled states, i.e., if another nonzero entanglement witness $W'\in \mathsf{Sep}(\Comp^n\otimes \Comp^n)^{\circ}$ satisfies
    $$\Tr(W \rho)<0 \implies \Tr(W'\rho)<0$$
for every (entangled) state $\rho$, then $W'=cW$ for some $c>0$ (we refer to \cite[Lemma 1 and Corollary 1]{lewenstein2000optimization}). It is shown \cite[Theorem 1]{lewenstein2000optimization} that the optimality of $W$ is equivalent to, for arbitrary nonzero positive semidefinite matrix $X\in (M_n\otimes M_n)^+$, we have $W- X\notin \mathsf{Sep}(\Comp^n\otimes \Comp^n)^{\circ}$. In other words, any decomposition $W=W'+X$ where $W'\in \mathsf{Sep}(\Comp^n\otimes \Comp^n)^{\circ}$ and $X\geq 0$ implies $X=0$. Therefore, optimal witnesses can be regarded as \emph{edge objects} as discussed in \cite{bruss2002reflections}. A necessary condition for optimality of $W$ is that, there exists a separable state $\rho_{\rm sep}$ such that $\Tr(W\rho_{\rm sep})=0$, meaning that $W$ is on the \textit{boundary} of $\mathsf{Sep}(\Comp^n\otimes \Comp^n)^{\circ}$ (every such $W$ is on an \emph{exposed face} of $\mathsf{Sep}(\Comp^n\otimes \Comp^n)^{\circ}$ according to \cite{kye2013facial}). However, the converse is not true: using the flip operator $F$, define the witness $W:=F+|11\ra\la 11|$. Clearly, since $F$ is also an entanglement witness, $W$ is not optimal, but it satisfies
    $$\Tr(W|12\ra\la 12|)=\la 12|F_d|12\ra=0,$$
proving that it is on the boundary of $\mathsf{Sep}(\Comp^2\otimes \Comp^2)^{\circ}$.

In the following statement, we show that the previous construction again provides a wide ($(n^2-n)/2$-dimensional) class of optimal entanglement witnesses.

\begin{corollary} \label{cor:OptimalEW}
Let $M$ be an extremal copositive matrix which is also exceptional, i.e., $M\notin \SPN_n$. Then for any $N\in \EWP_n^{\sa}$ with $\diag(N)=\diag(M)$ and $N_{ij}\geq \frac{1}{2} M_{ij}$ for all $i\neq j$, the CLDUI matrix $\XCLDUI{N}{M-\mathring{N}}$ is optimal and indecomposable.
\begin{proof}
Suppose we have a decomposition $\XCLDUI{N}{M-\mathring{N}}=X_1+X_2$ where $X_1$ is block-positive and $X_2$ is a positive semidefinite matrix. By taking the \textit{CLDUI-twirling} as in \cref{eq:CLDUITwirl},
we have $\XCLDUI{N}{M-\mathring{N}}=X_1'+X_2'$ where $X_i'=\mathcal{T}_{\mathsf{CLDUI}}X_i=:\XCLDUI{A_i}{B_i}$. Then \cite[Proposition 2.1]{park2024universal} implies that $X_1'$ is block-positive and $X_2\geq 0$, and hence $(A_1,B_1)\in \COPCP_n$, $A_2\in \EWP_n$, and $B_2\in \PSDC_n$. Now let us focus on the decomposition
    $$(N,M-\mathring{N})=(A_1,B_1)+(A_2,B_2)=(A_1,B_1)+(\mathring{A}_2,0)+(\diag(B_2),B_2),$$
which yields the following decomposition
    $$M=((A_1+A_1^{\top})/2+{\rm Re}(\mathring{B}_1))+ ((\mathring{A}_2+\mathring{A}_2^{\top})/2 + {\rm Re}(B_2))$$
of the copositive matrix $M$. Since $(\mathring{A}_2+\mathring{A}_2^{\top})/2 + {\rm Re}(B_2)\in \SPN_n$, the extremality and exceptionality of $M$ forces $(\mathring{A}_2+\mathring{A}_2^{\top})/2 + {\rm Re}(B_2)=0$. By comparing both sides, we further have $(\mathring{A}_2+\mathring{A}_2^{\top})/2={\rm Re}(B_2)=0$, and hence $A_2=B_2=0$. Finally, since the CLDUI-twirling $\mathcal{T}_{\mathsf{CLDUI}}$ is trace-preserving, we have $X_2=0$, which concludes the optimality of $\XLDUI{N}{M-\mathring{N}}$.
\end{proof}
\end{corollary}

\begin{remark}
\begin{enumerate}
    \item The condition $M\notin \SPN_n$ is essential above. Indeed, $M=3(|i\ra\la j|+|j\ra\la i|)$ is an extremal copositive matrix \cite[Theorem 2.29]{shaked2021copositive}, and the choice $N=2M/3$ gives a copositive pair $\frac{1}{3}(2M,M)\in \COPCP_n\setminus \CLDUI_n^+$. However, we have a decomposition
        $$(2M,M)=(M,M)+(M,0),$$
    where $(M,M)\in \COPCP_n$ and $(M,0)\in \CLDUI_n^+$.

    \item If $N_{ij}\geq \frac{1}{2}\max(0,M_{ij})$ for some $i\neq j$, then the LDUI matrix $\XLDUI{N}{M-\mathring{N}}$ is not optimal. Indeed, We can decompose the pair into
    \begin{align*}
        (N,M-\mathring{N})&=(N-N',M-\mathring{N}+N') + (N',-N'),
    \end{align*}
    where $N'=(N_{ij}-\frac{1}{2}\max(M_{ij},0))(|i\ra\la j|+|j\ra\la i|)$ and $(N',-N')\in \LDUI^+_n$.
\end{enumerate}
\end{remark}

\bigskip

\section{Markov-Choi maps}\label{sec:XJ}

In this section, we consider another variant of Example \ref{ex:J-J} obtained by perturbing the first element of the pair $(\mathring{J}, -\mathring{J})$ with an arbitrary real matrix $A-\mathring{J}$ (with $A$ nonnegative entries), while keeping the second element of the form $\diag(A) - \mathring{J}$. This class of examples yields $\DUC$ maps that are very important in the literature, such as the Choi map \cite{choi1975} and many of its generalizations \cite{choi1977,tanahashi1988,kye1992,cho1992,ha2003class,chruscinski2018generalizing,bera2025generalizing,scala2024optimality}. These maps are of the form 
$$\PhiDUC{A}{\diag(A) - \mathring J}(Z) = \diag\{A \diag[Z]\} - \mathring Z=\sum_{i,j=1}^n (A_{ij}+\mathds{1}_{i=j})Z_{jj}|i\ra\la i|-Z.$$
We will call such maps \emph{Markov-Choi} because of their precise form: on the diagonal, the map acts like a classical Markov matrix $A$, while on the off-diagonal elements, the map flips the sign of the entry, exactly like the celebrated Choi map \cite{choi1975}.

Let us first gather several positivity properties of such maps in the following proposition. 

\begin{proposition}\label{prop:equivalent-MC-positive}
Let $A \in \EWP_n$. Then:
\begin{enumerate}
    \item $(A,\diag(A)- \mathring J) \in \COPCP_n$ if and only if $\displaystyle{\sum_{i=1}^n \frac{p_i}{p_i + (Ap)_i} \leq 1}\, \quad \forall\, p \in \mathbb R^n_{>0}$.

    \item $(A,\diag(A)-\mathring{J})\in \PDEC_n$ if and only if $\diag(A)-\mathring{J}+S\in \PSDR_n$ for some $S\in \Mrealsa{n}$ with $\diag(S)=0$ and $|S_{ij}|^2\leq A_{ij}A_{ji}$ $\forall\, i\neq j\in [n]$.
    
    \item $(A,\diag(A)- \mathring J) \in \CLDUI^+_n$ if and only if $\displaystyle{\sum_{i=1}^n \frac{1}{1+A_{ii}} \leq 1}$.

    \item $(A,\diag(A)- \mathring J) \in \LDUI^+_n$ if and only if $A_{ij}A_{ji}\geq 1$ $\forall\, i\neq j\in [n]$.
    
    \item $(A,\diag(A)- \mathring J) \in \PCP_n$ if and only if $(A,\diag(A)- \mathring J) \in \PDNN_n = \CLDUI^+_n \cap \LDUI^+_n$ if and only if $\displaystyle{\sum_{i=1}^n \frac{1}{1+A_{ii}} \leq 1}$ and $A_{ij}A_{ji}\geq 1$ $\forall\, i\neq j\in [n]$.
\end{enumerate}
\end{proposition}
\begin{proof}
For the first point, $(A,\diag(A)- \mathring J) \in \COPCP_n$ if and only if the Markov-Choi map 
$$\PhiDUC{A}{\diag A - \mathring J}(Z) =  \diag\{A \diag[Z]\} - \mathring{Z} = \diag\{(I+A) \diag[Z]\} - Z$$
is positive. For an arbitrary vector $v \in \C{n}$, where $v_i\neq 0$ for all $i\in [n]$, we have $$\Phi(\ketbra{v}{v}) \geq 0 \iff \diag\{(I+A) \diag[\ketbra{v}{v}]\} \geq \ketbra{v}{v}$$ We define $D_v := \diag\{(I+A) \diag[\ketbra{v}{v}]$, ane hence it follows that, 
\begin{align*}
    \Phi(\ketbra{v}{v}) \geq 0 &\iff D_v^{-1/2}\ketbra{v}{v} D_v^{-1/2} \leq I \\ &\iff \left\| D_v^{-1/2} \ket v \right\| \leq 1 \\ &\iff \left\| \left(\frac{v_i}{\sqrt{\sum_{j=1}^n A_{ij} + \mathds{1}_{i=j})|v_j|^2}}\right)_{i \in [n]}\right\|^2 \leq 1\\ & \iff \sum_{i=1}^n \frac{|v_i|^2}{|v_i|^2+(A|v_\cdot|^2)_i}  \leq 1,
\end{align*}
which is the condition in the statement for $p_i := |v_i|^2>0$. Conversely, the condition implies that $\PhiDUC{A}{\diag A - \mathring J}(|v\ra\la v|)\geq 0$ for all $v\in \C{n}$ with non-zero entries, and thus the positivity of $\PhiDUC{A}{\diag A - \mathring J}$ by a density argument.

For the second point, $(A,\diag(A)-\mathring{J})\in \PDEC_n$ if and only if $\diag(A)-\mathring{J}=B^{(1)}+B^{(2)}$ where $B^{(1,2)}$ satisfies the conditions in Definition \ref{def:PDEC}. Note that we may assume that $B^{(1,2)}$ are real symmetric, by taking (entrywise) real part, and $\diag(B^{(2)})=0$ by choosing $\widetilde{B}^{(1)}=B^{(1)}+\diag(B^{(2)})$ and $\widetilde{B}^{(2)}=\mathring{B}^{(2)}$ if necessary. Therefore, we have the advertised condition by setting $S=-\widetilde{B}^{(2)}$.

For the third point, by Proposition \ref{prop:(C)DUC-CP-CLDUI+}, the condition on the LHS is equivalent to $\diag(A)- \mathring J \in \PSDC_n$. Therefore, similarly as in the proof of the first point,
\begin{align*}
    \diag(A) \geq \mathring J &\iff \diag(I+A) \geq J = \ketbra{\mathbf 1}{\mathbf 1}\iff \left\|(\diag(I+A))^{-1/2} \ket{\mathbf 1} \right\| \leq 1\\
    &\iff \left\| \left(\frac{1}{\sqrt{1+A_{ii}}}\right)_{i \in [n]}\right\|^2= \sum_{i=1}^n \frac{1}{1+A_{ii}}\leq 1.
\end{align*}
The fourth point is straightforward from the condition in Proposition \ref{prop:(C)DUC-CP-CLDUI+}. Finally, last fifth point follows from the \emph{comparison matrix} argument from \cite[Lemma 2.11]{singh2021diagonal} since $M(\diag(A)-\mathring{J})=\diag(A)-\mathring{J}$.
\end{proof}

Note that the last statement of Proposition \ref{prop:equivalent-MC-positive} implies that
\begin{center}
    $\XLDUICLDUI{A}{\diag(A)-\mathring{J}}$ is separable if and only if PPT,
\end{center}
by Propositions \ref{prop:PCP-EB} and \ref{prop:(C)DUC-PPT-PDNN}. In other words, it is impossible to find PPT entanglement directly from pairs of the form $(A, \diag(A)-\mathring{J})$.

It is in general challenging to solve analytically the conditions in Proposition~\ref{prop:equivalent-MC-positive}~(1) and~(2) (we refer to \cite{cho1992,yamagami1993,ha2003class,bera2025generalizing} for exact solutions in several special cases). Since our main interest lies in finding \emph{exceptional} pairs $(A,\diag(A)-\mathring{J})$ in $\COPCP_n\setminus \PDEC_n$ (equivalently, those for which the corresponding Markov–Choi map $\PhiDUC{A}{\diag A - \mathring J}$ is positive and indecomposable), we shall further restrict our attention to several special classes of matrices $A$ for which membership in the cones $\COPCP_n$ or $\PDEC_n$ can be explicitly obtained, and accordingly, examples of positive indecomposable maps can be readily constructed.

\subsection{Markov-Choi maps from exceptional copositive matrices} \label{sec:MC1}

In relation with copositivity, we first record some simple necessary and, respectively, sufficient conditions for $(A,\diag(A)-\mathring{J})\in \COPCP_n$.

\begin{proposition}
Let $A\in \EWP_n$ and let $\widetilde{A}:=(A\odot A^{\top})^{\odot 1/2}=(\sqrt{A_{ij}A_{ji}})_{i,j\in [n]}$. Then we have the following implications:
    $$\widetilde{A}-\mathring{J}\in \COP_n \implies (A,\diag(A)-\mathring{J})\in \COPCP_n \implies (A+A^{\top})/2-\mathring{J}\in \COP_n.$$
\end{proposition}
\begin{proof}
{
If $\widetilde{A}-\mathring{J} \in \COP_n$, then the first statement of \cref{thm:COPCP-from-COP} with $N=\widetilde{A}$ and $M=\widetilde{A} - \mathring{J}$ implies that $(\widetilde{A},\diag(\widetilde{A})-\mathring{J})\in \COPCP_n$. By Proposition \ref{prop:properties-COPCP} (4) , this implies that $\widetilde{A}-O\mathring{J}O\in \COP_n$ for all $O\in D\mathcal{O}_n$ and hence, by Corollary \ref{cor:PCOP-from-COP2}, we have that $(A,\diag(A)-\mathring{J})\in \COPCP_n$. This shows the first implication.}
The second implication is straightforward from Proposition \ref{prop:properties-COPCP} (3).
\end{proof} Note that if $A$ is further assumed to be \emph{symmetric}, i.e., $A\in \EWP_n^{\sa}$, then we have
\begin{equation} \label{eq:Symm-Markov-Choi}
    (A,\diag(A)-\mathring{J})\in \COPCP_n \iff A-\mathring{J}\in \COP_n.
\end{equation}
which is a special case of \cref{thm:COPCP-from-COP} by choosing $N=A, M = A - \mathring J$.

On the other hand, a necessary condition for $(A,\diag(A)-\mathring{J})\in \PDEC_n$ in Proposition \ref{prop:SPN-from-PairDEC} is $\widetilde{A}-\mathring{J}\in \SPN_n$. Therefore, one can construct the matrix pair as follows: for any \emph{exceptional} copositive matrix $M\in \COP_n\setminus \SPN_n$, we have $(A,\diag(A)-\mathring{J})\in \COPCP_n\setminus \PDEC_n$ whenever
    $$\widetilde{A}=(A\odot A^{\top})^{\odot 1/2}=\lambda M+\mathring{J}$$
for $0<\lambda<(\max(M_-))^{-1}$ (note that $M_-\neq 0$ since $M\notin \EWP_n$). Note that when $A=\lambda M+\mathring{J} = \tilde A$, the corresponding positive indecomposable Markov-Choi maps
    $$\PhiDUC{A}{\diag(A)-\mathring{J}}(Z)=\lambda \sum_{i,j}M_{ij}Z_{jj}|i\ra\la i|+\Tr(Z)I_n-Z, \quad Z\in \M{n},$$
already provides us with a new class of constructions. {To our knowledge, all the previous results in the literature assumed $A$ to be \emph{asymmetric} to construct examples of positive-indecomposable maps. Note that in the cases they consider, $\widetilde{A}-\mathring{J}$ is also not necessarily copositive. This shows that, this is not a necessary condition for $\COPCP_n$, and shift our focus on this case in the remainder of the section.}

\subsection{Markov-Choi maps from non-copositive matrices}

Here we discuss two cases where the exceptional pair $(A,\diag(A)-\mathring{J})\in \COPCP_n\setminus \PDEC_n$ is obtained from a matrix $A$ such that $\widetilde{A}-\mathring{J}\notin \COP_n$. Both cases would result in generalizing previous examples of positive indecomposable Markov-Choi maps from \cite{kye1992,cho1992,ha2003class}.

The first case is when $\widetilde{A}=(A\odot A^{\top})^{\odot 1/2}$ is diagonal. In this case, it is shown that
    $$\widetilde{A}-\mathring{J}\in \COP_n \iff \widetilde{A}-\mathring{J}\in \SPN_n \iff \widetilde{A}-\mathring{J}\in \PSDR_n \iff \sum_{i=1}^n \frac{1}{1+A_{ii}}\leq 1,$$
from \cite[Theorem 2.102]{shaked2021copositive}. 

We show below that this is also equivalent to the membership in $\PDEC_n$.

\begin{proposition} \label{prop:MC-PDEC1}
Suppose $A\in \EWP_n$ satisfies $A_{ij}A_{ji}=0$ for all $i\neq j\in [n]$. Then we have
    $$(A,\diag(A)-\mathring{J})\in \PDEC_n \iff \sum_{i=1}^n \frac{1}{1+A_{ii}}\leq 1 \iff (A,\diag(A)-\mathring{J}) \in \CLDUI_n^+.$$
In particular, we have $(A,\diag(A)-\mathring{J})\in \COPCP_n\setminus \PDEC_n$ whenever $A$ satisfies $\sum_{i=1}^n \frac{1}{1+A_{ii}} > 1$ and $\sum_{i=1}^n \frac{p_i}{p_i+(Ap)_i}\leq 1$ for all $p\in \Real_{>0}^n$.
\end{proposition}
\begin{proof}
Since $A_{ij}A_{ji}\equiv 0$, the condition in Proposition \ref{prop:equivalent-MC-positive} (2) forces that $S=0$, and hence $\diag(A)-\mathring{J}\in \PSDR_n$, which is equivalent to $(A,\diag(A)-\mathring{J})\in \CLDUI_n^+$ as in Proposition \ref{prop:equivalent-MC-positive} (3).
\end{proof}

\begin{example} \label{ex:MC-PosIndecomp1}
For a $3\times 3$ matrix $A=\begin{bmatrix}
    a_1 & 0 & c_1 \\ c_2 & a_2 & 0 \\ 0 & c_3 & a_3
\end{bmatrix}\in \EWP_3$, it is shown \cite{kye1992,chruscinski2018generalizing} that $(A,\diag(A)-\mathring{J})\in \COPCP_3$ whenever $a_1,a_2,a_3\geq a$ and $(c_1c_2c_3)^{1/3} \geq 2-a$ for some $a
\geq 1$. More generally, for $A\in \EWP_n$, the result of \cite[Theorem 2.4]{ha2003class} implies that $(A,\diag(A)-\mathring{J})\in \COPCP_n$ if, for some $k=1,\ldots, n-1$ and $a\geq n-2$, we have $A_{ii}\geq a$ for all $i\in [n]$ and $\big(\prod_{i=1}^nA_{i,i+k}\big)^{1/n}\geq n-a-1$ where the sum in the index is in modulo $n$. {Indeed, we may add an entrywise positive matrix to setting in \cite[Theorem 2.4]{ha2003class}, namely, $A=N+A^{\rm Ha}$ with $N\in \EWP_n$ and 
    $$A^{\rm Ha}_{ii}\equiv a \text{ and } A_{ij}^{\rm Ha}=c_i\delta_{j,i+k}, \quad a\geq n-2, \;\; (c_1\cdots c_n)^{1/n}\geq n-a-1,$$
and using the cone property of $\COPCP_n$}. Therefore, one can construct the exceptional pair $(A,\diag(A)-\mathring{J})\in \COPCP_n\setminus \PDEC_n$ under the conditions
    $$A_{ij}A_{ji}=0,\;\; a\leq A_{ii}<n-1 \text{ for all $i\neq j\in [n]$, and } \Big(\prod_{i=1}^nA_{i,i+k}\Big)^{1/n}\geq n-a-1$$
for some $k=1,\ldots, n-1$ and $n-2\leq a<n-1$. 

\end{example}

\begin{remark}
In \cite{ha2003class}, the author showed a stronger form of indecomposability in the previous map: when $A_{ii}\equiv a$ and $A_{ij}=c_i\, \delta_{j,i+k}$ with $(c_1\cdots c_n)^{1/n}\geq n-a-1$, the corresponding map $\PhiDUC{A}{\diag(A)-\mathring{J}}$ is \emph{atomic}, i.e., it cannot be a sum of a 2-positive and a 2-copositive linear map.
\end{remark}

\medskip

Next, let us consider the case where $\widetilde{A} = (A \odot A^T)^{\odot 1/2}=aI_n+\lambda\mathring{J}_n$ for some $a,\lambda \geq 0$.
As in the previous case, we can show $\widetilde{A}-\mathring{J}\in \COP_n$ if and only if $\widetilde{A}-\mathring{J}\in \SPN_n$, which also characterizes pairwise decomposability of $(A,\diag(A)-\mathring{J})$.

\begin{proposition} \label{prop:MC-PDEC2}
Suppose $A\in \EWP_n$ satisfies $A_{ii}\equiv a$ and  $\sqrt{A_{ij}A_{ji}}\equiv \lambda$ for all $i\neq j\in [n]$. Then we have
    $$(A,\diag(A)-\mathring{J})\in \PDEC_n \iff \widetilde{A}-\mathring{J}\in \COP_n \iff \widetilde{A}-\mathring{J}\in \SPN_n \iff \lambda\geq 1-\frac{a}{n-1}.$$
In particular, we have $(A,\diag(A)-\mathring{J})\in \COPCP_n\setminus \PDEC_n$ whenever $0\leq a<n-1$, $0\leq \lambda<1-\frac{a}{n-1}$, and  $A$ satisfies $\sum_{i=1}^n \frac{p_i}{p_i+(Ap)_i}\leq 1$ for all $p\in \Real_{>0}^n$.
\end{proposition}
\begin{proof}
The condition $(A,\diag(A)-\mathring{J})\in \PDEC_n$ is equivalent to $(aI_n+\lambda\mathring{J},aI_n-\mathring{J}_n)\in \PDEC_n$, thanks to \cref{eq:SymmPDEC}. The decomposability of the corresponding Markov-Choi map
    $$\PhiDUC{aI_n+\lambda\mathring{J}}{aI_n-\mathring{J}_n}(Z) = \lambda \Tr(Z)I_n+(a+1-\lambda)\diag(Z)-Z, \quad Z\in \M{n},$$
has been characterized in \cite[Section 4]{park2024universal}, where it is shown that it is equivalent to its positivity. This shows that
    $$(A,\diag(A)-\mathring{J})\in \PDEC_n \iff (\widetilde{A},\diag(A)-\mathring{J})\in \COPCP_n \iff \widetilde{A}-\mathring{J}\in \COP_n,$$
where the last equivalence follows from \cref{thm:COPCP-from-COP}.
Now if $\lambda\geq 1$, then clearly $\widetilde{A}-\mathring{J}_n=aI_n+(\lambda-1)\mathring{J}_n\in \EWP_n^{\sa}$. If $\lambda<1$, then by \cite[Theorem 2.102]{berman2023completely},
\begin{align*}
    \widetilde{A}-\mathring{J}=aI_n+(\lambda-1)\mathring{J}_n\in \COP_n & \iff aI_n+(\lambda-1)\mathring{J}_n\in \SPN_n\\
    &\iff aI_n+(\lambda-1)\mathring{J}_n\in \PSDR_n\\
    &\iff a+(n-1) 
    (\lambda-1)\geq 0,
\end{align*}
which shows the remaining equivalences in the statement.
\end{proof}

\begin{example} \label{ex:MC-PosIndecomp2}
Consider a $3\times 3$ matrix $A=\begin{bmatrix}
    a & b_1 & c_1 \\ c_2 & a & b_2 \\ b_3 & c_3 & a
\end{bmatrix}\in \EWP_3$ satisfying
    $$\sqrt{b_1c_2}=\sqrt{b_2c_3}=\sqrt{c_1b_3}=\lambda, \quad a+(b_1b_2b_3)^{1/3}+(c_1c_2c_3)^{1/3}\geq 2, \quad \text{and} \quad a+\lambda\geq 1.$$
If $b_*:=(b_1b_2b_3)^{1/3}$ and $c_*:=(c_1c_2c_3)^{1/3}$, the last two conditions are equivalent to $a+b_*+c_*\geq 2$ and $a+\sqrt{b_*c_*}=a+\lambda\geq 1$. Therefore, we may apply the result from \cite[Corollary 2.4]{chruscinski2018generalizing}, which is mainly based on \cite[Theorem 2.1]{cho1992}, to obtain that $(A,\diag(A)-\mathring{J})\in \COPCP_3$. 
On the other hand, Proposition \ref{prop:MC-PDEC2} implies that $(A,\diag(A)-\mathring{J})\in \PDEC_n$ if and only if $a+2\lambda\geq 2$, which depends only on $a$ and $\lambda$. Consequently, we can obtain positive indecomposable Markov-Choi maps $\PhiDUC{A}{\diag(A)-\mathring{J}}$ under the conditions $0\leq a <2$ and $1-a\leq \lambda<1-\frac{a}{2}$, thereby recovering and generalizing the construction in \cite{cho1992}.
\end{example}

{Finally, we provide another family of positive indecomposable maps in all dimensions.}

\begin{example} \label{ex:MC-PosIndecomp3}
Let us consider a \emph{circulant matrix} $A={\rm circ}(a,c_1,\ldots, c_{n-1})\in \EWP_n$, i.e.,
\begin{equation} \label{eq:circulant}
    A_{ii}\equiv a \text{ and } A_{i,i+k}\equiv c_k; \quad i=1,\ldots, n,\;\;k=1,\ldots, n-1,
\end{equation}
with $a,c_j\geq 0$ and $\sqrt{c_jc_{n-j}}=\lambda$ for all $j=1,\ldots, \lfloor\frac{n}{2}\rfloor$. Then by taking $c:=\max_{1\leq j\leq n-1} c_j$, the argument in Example \ref{ex:MC-PosIndecomp1} based on \cite{ha2003class} implies that $(A,\diag(A)-\mathring{J})\in \COPCP_n$ whenever $a\geq n-2$ and $c\geq n-a-1$. Therefore, Proposition \ref{prop:MC-PDEC2} implies that $(A,\diag(A)-\mathring{J})\in \COPCP_n\setminus \PDEC_n$ under the conditions
    $$n-2\leq a < n-1, \quad c\geq n-a-1, \quad \text{and} \quad 0\leq \lambda < 1-\frac{a}{n-1}.$$
\end{example}

\section{Maps associated to graphs}\label{sec:graphs}

In this section, we define a one-parameter family of maps using adjacency matrix of a graph. Recall that for every graph $G$ we can define its \emph{adjacency matrix} 
$$A_G(i,j) = \begin{cases}
    1 &\quad \text{ if $(i,j)$ is an edge in $G$}\\
    0 &\quad \text{ otherwise}.
\end{cases}$$

Throughout this and the next section, we assume that $G$ is a simple, undirected, and loop-less graph with $n$ vertices, in which case we always have $A_G\in \Mrealsa{n}$ and $A_G=\mathring{A}_G$. We introduce the following quantities:
\begin{itemize}
    \item $\lambda(G)$ the \emph{maximal eigenvalue} of its adjacency matrix $A_G$. For a non-empty graph, $\lambda(G) \geq 1$;
    \item $\omega(G)$ its \emph{clique number}, i.e.~the size of the largest clique in $G$;
    \item $\alpha(G)$ the \emph{stable or independence number} of $G$, i.e.~the size of its largest stable set; clearly $\alpha(G) = \omega(\bar G)$, where $\bar G$ is the complementary graph of $G$.
\end{itemize}
We refer the reader to \cite{godsil2013algebraic} for more details about algebraic graph theory and graph parameters. 

For a graph $G$ on $n$ vertices, we also introduce the quantity 
\begin{equation}\label{eq:def-sigma}
    \sigma(G) := \max\{t \in \R \, : \,  J - t A_G \in \SPN_n\},
\end{equation}
where we recall that $\SPN_n$ is the set of real symmetric matrices that can be written as a sum of a positive semidefinite real matrix and an entrywise non-negative matrix. In particular, $\sigma(G) \geq 1+\frac{1}{n-1}$ for all graphs $G$ since
    $$J_n-\Big(1+\frac{1}{n-1}\Big)A_G = \frac{1}{n-1}(nI_n-J_n)+\frac{n}{n-1}(J_n-I_n-A_G)\in \PSDR_n+\EWP_n^{\sa}=\SPN_n.$$
Note that the value $\sigma(G)$ as defined is a property of the isomorphism class of the graph $G$ because the set $\SPN_n$ is invariant under permutations. Moreover, given a graph $G$, the value $\sigma(G)$ can be computed using semidefinite programming \cite{boyd2004convex}. Moreover, it turns out that $\sigma(G)$ is related to a \emph{hierarchy of semidefinite programs} that approximates the independence number of $G$ \cite{deklerk2002approximation, schrijver2003comparison} that we review next; this hierarchy will also play an important role in \cref{sec:sos-hierarchies}. First, we recall the following result casting the independence number as a copositive program. 

\begin{theorem}[{{\cite[Corollary 2.4]{deklerk2002approximation}}}]\label{thm:stable-COP}
    For any graph $G$ on $n$ vertices,
    $$\alpha(G) = \min\{ t \, : \, t(I+A_G) - J \in \COP_n\}.$$
\end{theorem}

Since membership in the copositive cone is a computationally hard problem \cite{murty1987some}, the following hierarchy of semidefinite program has been introduced. We recall that a polynomial $P$ in $n$ variables $x_1, \ldots, x_n$ is a sum of squares (SOS) if it can be decomposed as a
$$P(x) = \sum_{i=1}^d Q_i(x)^2$$
for some polynomials $Q_1, \ldots, Q_d$; SOS polynomials are clearly non-negative, but the converse is false in general \cite{blekherman2012nonnegative, blekherman2012semidefinite}.

\begin{definition}[{{\cite[Definition 3.1 and Eq.~(20)]{deklerk2002approximation}}}]\label{def:graph-parameters-K}
    For any graph $G$ on $n$ vertices, define
    $$\vartheta^{(r)}(G) := \min\{ t \, : \, t(I+A_G) - J \in \K_n^{(r)}\},$$
    where
    $$\K_n^{(r)}:=\{ A \in \Mrealsa{n} \, : \, x\in \Real^n \mapsto \la x^{\odot 2}| A | x^{\odot 2} \ra \|x\|_2^{2r} \text{ is SOS}\}.$$
\end{definition}
We gather next some important and relevant results regarding the parameters $\vartheta^{(r)}(G)$ and the corresponding sets $\K_n^{(r)}$ \cite{deklerk2002approximation, gvozdenovic2007semidefinite, laurent2023exactness, schweighofer2024sum}. First, for $r=0$ the sum of squares decomposition for the polynomial $\la x^{\odot 2}| A | x^{\odot 2} \ra$ corresponds precisely to decomposing $A = P + E$ with a real positive semidefinite matrix $P$ and a matrix $E$ having non-negative entries \cite{parrilo2000structured}; we have thus
$$\K_n^{(0)} = \SPN_n.$$
The fundamental result here is that the numbers $\vartheta^{(r)}(G)$ provide increasingly good approximations for the independence number \cite{deklerk2002approximation}: 
$$\alpha(G) \xleftarrow{r \to \infty} \vartheta^{(r)}(G)\leq \cdots \leq \vartheta^{(1)} \leq  \vartheta^{(0)}(G).$$
It has been recently shown \cite{schweighofer2024sum} that the sequence $(\vartheta^{(r)}(G))_{r\geq 0}$ is actually finite: for every graph $G$ there exists a finite $r_0$ such that $\alpha(G)= \vartheta^{(r_0)}(G)$, while the conjecture \cite{deklerk2002approximation} that $\vartheta^{(\alpha(G)-1)}(G)=\alpha(G)$, proven for $\alpha(G) \leq 8$, remains open in the general case. 

Since $J-tA_G = J-t(J-I-A_{\bar{G}})= (t-1)(\frac{t}{t-1}(I+A_{\bar{G}})-J)$ and since $\sigma(G)>1$, it is now clear that the quantity $\sigma$ from \cref{eq:def-sigma} is related to the value $\vartheta^{(0)}$ from Definition \ref{def:graph-parameters-K}: 
$$\sigma(G) = \frac{\vartheta^{(0)}(\bar G)}{\vartheta^{(0)}(\bar G)-1}.$$

Since $\vartheta^{(0)}(\bar G) \geq \alpha(\bar G) = \omega(G)$, we have the following upper bound:
\begin{equation}\label{eq:UB-sigma}
    \sigma(G) \leq \frac{\omega(G)}{\omega(G) - 1} = 1+\frac{1}{\omega(G)-1}.  
\end{equation}

\begin{lemma}
    For every graph $G$, we have the following lower bound:
    $$\sigma(G) \geq 1 + \frac{1}{\lambda(G)}.$$
\end{lemma}
\begin{proof}
    The claim follows from the decomposition
    $$J-\Big(1 +\frac{1}{\lambda(G)}\Big) A_G=
    \underbrace{I - \frac{1}{\lambda(G)}A_G}_{\in \PSDR}+
    \underbrace{J-I-A_G}_{ = A_{\bar G} \in \EWP}\in \SPN_n.$$
\end{proof}

\medskip

We come now to the main result of this section, characterizing the different positivity properties of a linear map defined by  a graph. More precisely, given a graph $G$ and $t \geq 0$, let us consider the pair $(J, I-t A_G) \in \Mreal{n} \underset{\R{n}}{\times} \Msa{n}$.

\begin{theorem}\label{thm:graph-map-properties}
    Let $G$ be a graph with at least one edge, and consider, for a non-negative parameter $t \geq 0$, the linear map 
    \begin{equation}\label{eq:def-Phi-G-t}
        \Phi^G_t(Z) := (\Tr Z) I - t A_G \odot Z, \quad Z\in \M{n}
    \end{equation}
    Then, the map $\Phi^G_t(Z)$ is
    \begin{enumerate}
        \item completely positive $\iff$ PPT $\iff$ EB $\iff$ $t \leq 1/\lambda(G)$;
        \item completely copositive $\iff$ $t \leq 1$;
        \item decomposable $\iff$ $t \leq \sigma(G)$;
        \item positive $\iff$ $t \leq 1+1/(\omega(G)-1)$;
    \end{enumerate}
\end{theorem}

\begin{proof}
The map $\Phi^G_t$ is a $\DUC$ map corresponding to the pair $(J, I - t A_G)$. Hence, many of the properties in the statement can be addressed using the theory developed in this work and in \cite{singh2021diagonal}. 

\medskip

\noindent \underline{\textit{Complete positivity}}. We have 
$$(J, I - t A_G) \in \CLDUI_n^+ \iff \begin{cases}
    J \in \EWP_n \\
    I-t A_G \in \PSDR_n \iff t \lambda(G) \leq 1.
\end{cases}$$
Note that $I-t A_G\in \mathsf{PSD}_n$ if and only if $I-t A_G$ is a \textit{correlation matrix}. By \cite[Proposition 3.7]{singh2021diagonal}, we have 
\begin{center}
    $\Phi_t^{G}$ is EB $\iff (J,I-t A_G)\in \PCP_n \iff$ $\Phi_t^{G}$ is PPT $\iff$ $I-t A_G\in \PSDR_n \iff $ $\Phi_t^{G}$ is CP.
\end{center}

\medskip

\noindent\underline{\textit{Complete copositivity}}. This is very similar to the previous case: 
$$(J, I - t A_G) \in \LDUI_n^+ \iff \begin{cases}
    J \in \EWP_n &\\
    \forall i \neq j, \  (I-t A_G)_{ij}^2 \leq 1 \iff t \leq 1.
\end{cases}$$

\medskip

\noindent\underline{\textit{Decomposability}}. 
Since $\max(J-tA_G)=1$, the map $\Phi_t^G$ is decomposable, i.e., $(J, I-t A_G) \in \PDEC_n$, if and only if $A=J - t A_G \in \SPN_n$ by \cref{thm:PDEC-from-SPN} with $N=J$, and thus $t \leq \sigma(G)$ by the definition of the $\sigma$.

\medskip

\noindent\underline{\textit{Positivity}}. With $Y := I - tA_G$, we are in the setting of Corollary \ref{cor:diagA-A}: 
$$\Phi^G_t \text{ is positive} \iff Y + \mathring J \in \COP_n \iff J - t A_G \in \COP_n.$$

We are are now going to switch to the complement $\bar G$ of the graph G and use the fact that $A_{\bar G} = \mathring J - A_G$. The condition becomes now 
$$J - t A_G \in \COP_n \iff t(A_{\bar G} + I) + (1-t)J\in \COP_n.$$
If $t \leq 1$, then the matrix above is doubly non-negative, as the sum of a entrywise positive matrix ($t(A_{\bar G} + I)$) and a positive semidefinite matrix ($(1-t)J$), and thus copositive. This settles the case $t \leq 1$. If $t>1$, let us rewrite 
$$t(A_{\bar G} + I) + (1-t)J = (t-1)\left[\frac{t}{t-1}(A_{\bar G} + I) - J\right] $$
and then
use \cref{thm:stable-COP}:
$$ \frac{t}{t-1}(A_{\bar G} + I) - J \in \COP_n \iff \frac{t}{t-1} \geq \alpha(\bar G) = \omega(G) \iff t \leq \frac{\omega(G)}{\omega(G)-1} = 1 + \frac{1}{\omega(G)-1},$$
proving the claim in the statement.

\end{proof}

We summarize all the bounds in the theorem above in \cref{fig:parameter-t-range}. Importantly, the construction above yields:
\begin{enumerate}
    \item \emph{positive but not completely positive maps}, for all values of the parameter 
    $$t \in \Big( \frac{1}{\lambda(G)}, 1 + \frac{1}{\omega(G)-1} \Big] \neq \emptyset.$$
    For example, the value $t=1$ yields such a map, for all non-empty graphs $G$.
    \item \emph{positive indecomposable maps}, when $G \neq \mathcal K_n$,  for all values of the parameter 
    $$t \in \Big( \sigma(G), 1 + \frac{1}{\omega(G)-1} \Big], \quad \text{whenever } \sigma(G) < 1 + \frac{1}{\omega(G)-1}.$$
    We shall discuss several examples of graphs for which the interval above is non-empty in the following subsections.
\end{enumerate}

\begin{figure}[htb]
    \centering
    \includegraphics[scale=1.25]{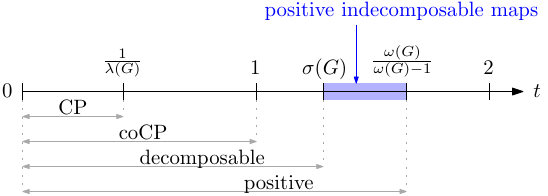}
    \caption{The positivity properties of the linear maps $\Phi^G_t$ from \eqref{eq:def-Phi-G-t} as the parameter $t$ varies. The blue region corresponds to the important case of positive indecomposable maps.}
    \label{fig:parameter-t-range}
\end{figure}

\begin{remark} 
    Linear maps having a similar structure as the ones in \eqref{eq:def-Phi-G-t} been studied in \cite[Section 4]{kennedy2018composition}: 
    $$\Psi_t^G(Z) = \frac t n ( \Tr Z)I_n + A_G \odot Z, \quad Z\in \M{n}.$$
    The main difference between the maps $\Phi_t^G$ and $\Psi_t^G$ is the sign in front of the Hadamard product term $A_G \odot Z$. This leads to very different structural properties of the maps, the two regimes being complementary. Note that the fact that the PPT$^2$ conjecture holds for the maps $\Psi_t^G$, see  \cite[Corollary 4.4]{kennedy2018composition}, is a special case of the more general proof for $\mathsf{DUC}$ maps from \cite[Theorem 4.5]{singh2022ppt}. Note that the PPT$^2$ conjecture for the maps $\Phi_t^G$ trivially holds, since these maps are automatically PPT and EB as soon as they are completely positive, by the first point of \cref{thm:graph-map-properties}.
\end{remark}

\begin{remark}
    While $\sigma(G)$ can be computed with a SDP, computing the parameter $\omega(G)$ is NP-complete. This makes it challenging to characterize whether the map $\Phi^G_t$ is indecomposable for an arbitrary graph.    
\end{remark}
In the next sections we focus on the parameter $\sigma(G)$ and $\omega(G)$ with the aim of exhibiting large families of graph-induced linear maps $\Phi^G_t$ that are positive indecomposable. 
To this end, we first gather some basic general properties of the quantity $\sigma(\cdot)$.
{
\begin{proposition}\label{prop:sigma-properties}
    The quantity $\sigma(\cdot)$ has the following properties:
    \begin{enumerate}
        \item We have the following dual formulation for the quantity $\sigma$:
        \begin{equation}\label{eq:sigma-dual}
            \sigma(G) = \min \{ \Tr(JX) \, : \, X \in \SPN_n^{\circ} \, \text{ s.t. } \Tr(A_GX) = 1\}.
        \end{equation}
        \item $\sigma(G \sqcup H) = \min(\sigma(G), \sigma(H))$.
    \end{enumerate}
\end{proposition}

\begin{proof}
\noindent\underline{\textit{Dual formulation}}. 
The computation of the dual semidefinite program of \cref{eq:def-sigma} is straightforward. Note that the dual set of $\SPN_n$ can be readily computed using \cref{eq:def-SPN}:
$$\SPN_n^{\circ} = \PSDR_n \cap \EWP^{\sa}_n = \DNN_n$$
The reader should also refer to \cite[Lemma 5.2]{deklerk2002approximation}, where several formulations for the dual of the $\vartheta^{(0)}$ quantity are discussed and a connection to Schrijver's $\vartheta'$ function \cite{schrijver2003comparison} is put forward. 

\medskip

\noindent\underline{\textit{Disjoint union of graphs}}. 
First, note that if $G' \subseteq G$ is an induced subgraph of $G$, by taking submatrices we obviously have 
$$\sigma(G') \geq \sigma(G).$$
This fact establishes the inequality $\sigma(G \sqcup H) \leq \min(\sigma(G), \sigma(H))$. For the reverse inequality, let $t:= \min(\sigma(G), \sigma(H))$. There exist positive semidefinite matrices $R_{1,2}$ and entrywise positive matrices $E_{1,2}$ such that
$$J - tA_G = R_1 + E_1 \quad \text{ and } \quad J - t A_H = R_2 + E_2.$$
We write then 
$$J - t(A_G \oplus A_H) = (R_1 \oplus R_2) + (E_1 \oplus E_2 + \underbrace{J-J \oplus J}_{\in \EWP}),$$
finishing the proof. 
\end{proof}
}

We now study a way to compute the value $\sigma(G)$ for graphs with symmetries. To simplify this problem, we consider the set of $n \times n$ matrices commuting with the automorphisms of a graph $G$: 
\begin{equation}
\label{eq:CG-def}\mathsf C_G := \{X \in \Mrealsa{n} \, : \, [X, P_\pi] = 0 \quad \forall \pi \in \Aut(G)\}.
\end{equation}
Above, $P_\pi$ denotes the permutation matrix associated to a permutation $\pi\in S_n$, and elements in $\Aut(G):=\{\pi\in S_n:[A_G,P_\pi]=0\}$ are called \emph{automorphisms} of $G$. 

Now consider the \emph{twirling map} $\mathcal T_G : \M{n} \to \mathsf C_G$ defined by
$$ \mathcal T_G(X) := \frac{1}{|\Aut(G)|} \sum_{\pi \in \Aut(G)} P_\pi X P_\pi^{-1}.$$
The fact that the range of this map lies inside $\mathsf C_G$ is standard. Introduce 
    $$\SPN_G := \{R + E \, : \, R \in \PSDR_n \cap \mathsf C_G  \text{ and } E \in \EWP_n^\sa \cap \mathsf C_G\}$$
and 
    $$\hat \sigma(G):= \max\{ t \in \R{} \, : \, J - t A_G \in \SPN_G\}.$$

\begin{proposition}\label{prop:hat-sigma}
    For any graph $G$, we have $\sigma(G) = \hat \sigma(G)$.
\end{proposition}
\begin{proof}
    Since $\SPN_G \subseteq \SPN_n$, we have $\hat \sigma(G) \leq \sigma(G)$. For the reverse inequality, let $\sigma(G) = t$ and $X = J-tA_G$. Since $X \in \SPN_n$, we have $X = R+E$ with $R \in \PSDR_n$ and $E \in \EWP_n^\sa$. Write 
    $$X = \mathcal T_G(X) = \underbrace{\mathcal T_G(R)}_{=:\hat R} + \underbrace{\mathcal T_G(E)}_{=:\hat E}.$$
    We have $\hat R, \hat E \in \mathsf C_G$. Moreover, it is clear that the twirling operation preserves positive semidefiniteness and entrywise positivity, hence $\hat R \in \PSDR_n$ and $\hat E \in \EWP_n^\sa$, showing that $X \in \SPN_G$ and concluding the proof.
\end{proof}

\subsection{Triangle-free and cyclic graphs}

A graph $G$ is called \emph{triangle-free} if it does not contain any three-cycles, i.e., there exists no sequence of 3 edges $(a,b), (b,c), (c,a)$ of $G$ with $a\neq b\neq c\in [n]$. The triangle-free property of a graph $G$ with at least one edge is alternatively characterized by $\omega(G)=2$.
In the first result of this section,  we completely characterize the set of triangle-free graphs $G$ such that  $\sigma(G) = \omega(G)/(\omega(G) - 1)$, i.e.~the upper bound in \eqref{eq:UB-sigma} is saturated. 
\begin{proposition}\label{prop:triangle-free-bipartite}
    If a graph $G$ is triangle free, then $\sigma(G) = 2 = \omega (G) /(\omega(G)-1)$ if and only if $G$ is bipartite.  
\end{proposition}

\begin{proof} 
    If the graph $G$ is bipartite on $n+m$ vertices, then its adjacency matrix is of the form $$A_G = \begin{bmatrix}
        0 & B \\ B^\top &0
    \end{bmatrix}$$
    for some $n \times m$ matrix $B$. We have now
    $${J} - 2A_G = \begin{bmatrix}
        {J}_n & {J}_{n \times m} - 2B \\ {J}_{m \times n} - 2B^\top& {J}_m \end{bmatrix}
        = 
        \underbrace{\begin{bmatrix}
        {J} & -{J}  \\ -{J} & {J}
    \end{bmatrix}}_{\in \PSDR_{n+m}}
    + 
    2\underbrace{\begin{bmatrix}
        0 & + {J} - B \\
       {J} - B^\top &0 
        \end{bmatrix}}_{\in \EWP_{n+m}} \in \SPN_{n+m}.
    $$ 
    
    For the reverse implication, we first consider the case of \emph{connected} triangle-free graphs. Now assume that for a connected triangle-free graph $G$ on $n$ vertices, we have $\sigma(G) = 2$, i.e., ${J} - 2 A_G \in \SPN_n$. We conclude using \cite[Theorem 2.140]{shaked2021copositive}, {stating that if $A=(A_{ij})\in \Mrealsa{n}$ satisfies ${\rm diag}(A)=I_n$, $A_{ij}\geq -1$ for all $i,j$, and if $G_{-1}(A)$ is connected, then $A\in \SPN_n$ if and only if $G_{-1}(A)$ is bipartite. (Here $G_{-1}(A)$ denotes the graph with $n$ vertices such that $A_{G_{-1}(A)}(i,j)=1$ if and only if $A_{ij}=-1$).} Indeed, the entries of the matrix $A=J - 2 A_G$ are all larger than or equal to $-1$ and the graph associated to the $-1$ values ${G}_{-1}(J - 2 A_G)=G$ is connected. Hence, the graph $G$ is bipartite. 
    If the graph has more than $m$ connected components $G_i$, by Proposition \ref{prop:sigma-properties}, we have $\operatorname{min}(\sigma(G_i)) = 2$. Also note that $\sigma(G_i) \leq \omega(G_i)/((\omega(G_i)-1) \leq 2$. This implies that $\forall i \, \, (\sigma(G_i)) = 2$ implying that the connected graph $G_i$ is bipartite. Hence $G = \bigsqcup G_i$ is bipartite.
\end{proof}

As a corollary, every triangle-free non-bipartite graph $G$ yields positive indecomposable maps. 
\begin{corollary}
    Let $G$ be a triangle-free non-bipartite graph. Then, for every value of the parameter $t$ in the non-empty interval $(\sigma(G), 2]$, the linear map $\Phi^G_t$ is positive and indecomposable. 
\end{corollary}

The smallest triangle-free non-bipartite graph is the cycle graph on 5 vertices, denoted as $C_5$ (also known as a pentagon).

Recall that the \emph{cycle graph} on $n$ vertices is the graph $C_n$ having edges 
$$\{(i,j) \, : \,  |i-j| = 1 \, (\text{mod }n) \}.$$
Cycle graphs with $n \geq 4$ are clearly triangle-free, and $C_n$ is bipartite if and only if $n$ is even; in this case, by Proposition \ref{prop:triangle-free-bipartite}, we have $\sigma(C_n) = 2 = \omega(C_n)$, hence there are no positive indecomposable maps associated to even cycles. For odd $n \geq 5$, we compute in what follows $\sigma(C_n)$.

\begin{proposition}\label{prop:cycle-graph}
    For odd $n \geq 5$, $\sigma(C_n) = 1 + \cos(\pi / n)$. Hence, for all 
    $$t \in (1+\cos(\pi/n), 2],$$
    the linear map $\Phi^{C_n}_t$ is positive indecomposable.
\end{proposition}

\begin{proof}
For the cyclic graph $C_n$, the group of automorphisms is the dihedral group of permutations, $D_{n}$ \cite{rodriguez2014cycle}. The set of $n \times n$ matrices that commute with the automorphisms, $\mathsf C_G$ are real, symmetric and circulant matrices. Since the circulant matrices are parametrized by the vector $a\in \Real^n$ with $a_i=A_{0i}$, we use the map $\operatorname{circ}(a) = A$ for any circulant matrix $A$ associated to the vector $a$ (we refer to \cref{eq:circulant} for the definition of circulant matrices).
A circulant matrix has eigenvalues given by a multiple of the \emph{discrete Fourier transform} of this vector \cite{davis1979circulant}
    $${\forall k \in [n] \, ; \, \lambda_k  =\sqrt{n}(\mathcal{F}a)_k = \sum_l a_l \exp\Big(\frac{2\pi i}{n}kl\Big)}$$ 
hence the positive-semidefinite of matrix $A$ is equivalent to $\forall i \, :\, (\mathcal{F}a)_i \geq 0$. Then, by using Proposition \ref{prop:hat-sigma}, we have 
$$\sigma(C_n):= \max\{ t \in \R{} \, : \, J - t A_{C_n} \in \SPN_{C_n}\}.$$ 
such that 
$$\SPN_{C_n}= \{R + E \, : \, R \in \PSDR_n \cap \mathsf C_{C_n}  \text{ and } E \in \EWP_n^{sa} \cap \mathsf C_{C_n}\}$$ Using $R = \operatorname{circ}(r)$ and $E = \operatorname{circ}(e)$ and noting that $J-tA_{C_n}=\operatorname{circ}(1,1-t,1,\ldots, 1,1-t)$, we can recast this as a linear program in the entries, 
\begin{align*}
        \sigma(C_n) = \max \quad &t\\
        \text{s.t.}\quad 
        &r_{\pm1} + e_{\pm1} = 1-t \\
        \forall i \notin \{0,\pm 1\}\quad  &r_i + e_i = 1\\
        &\mathcal{F}r \geq 0 \\
        &e \geq 0 \\
        & e_i = e_{-i} \\
        & e,r \in \mathbb{R}^n 
    \end{align*}
    
Assume the maximum is at $t = t^*$ with the values $r = r^*$ and $e=e^*$.

\begin{enumerate}
    \item Define, $r' = r + e_0  \ket{0}$ and $e' = e - e_0 \ket{0}$. Clearly $\mathcal{F}r' =  \mathcal{F}r + e_0\mathcal{F}  \ket{0} \geq 0$ and
    $e = e - e_0 \ket{0} \geq 0$. Also, $r' + e' = r + e$, therefore this change of variables gives the same value for $t$ while having $e'_0 = 0$ and satisfying all the constraints. Therefore, we can assume $e^*_0 = 0$.
    \item Assume that $e^*_{\pm 1} > 0$, then, since we have, $\ket{\mathbf{1}_n} - t^* a_G = r^* + e^*$, it implies that $\ket{\mathbf{1}_n} - (t^*+ e^*_{1}) a_g = r^* + (e^* - e_1\ket{1} - e_{-1}\ket{-1})$, contradicting the fact that $t^*$ is the optimum value. Hence, we also assume $e^*_{\pm 1} = 0$.
\end{enumerate}

The arguments above simplifies the optimization problem to the following one, 
\begin{align*}
        \sigma(C_n) = \max \quad &t\\
        \text{s.t.}\quad &r_0 = 1 \\
        &r_{\pm1} = 1-t \\
        \forall i \notin \{0,\pm 1\}\quad & r_i + e_i = 1 \\
        &\mathcal{F}r \geq 0 \\ 
        \forall i \notin \{0,\pm 1\}\quad & e_i \geq 0\\
        & e_i = e_{-i} \\
        & e,r \in \mathbb{R}^n
\end{align*} 
We show that $e_i \geq 0$ follows from the other constraints. The constraints, $\mathcal{F}r \geq 0$ and $r_0 = 1$ imply that $\forall i,  \,|r_i| \leq 1$. Also, have $\forall i \notin \{0, \pm 1\},  \, \,  r_i + e_i = 1$. Therefore, it holds that $e_i = 1- r_i \geq 0$. Moreover, we can  these $e_i$ freely such that $r_i + e_i = 1$. A quick calculation shows that,
    $$\sigma(C_n) = 1 - \operatorname{min}\{r_1 \in \mathbb{R} :r_0 = 1, \, \mathcal{F}r \geq 0, r \in \mathbb{R}^n\}$$ 
Since the set $\{r\in \Real^n:r_0=1,\,\mathcal{F}r\geq 0\}$ is convex, we may examine its extreme points, namely $\mathcal{F}^{-1}|0\ra=\sqrt n \sum_{j=0}^{n-1}|j\ra$ and $\frac{1}{2}\mathcal{F}^{-1}(|k\ra+|-k\ra)=\sqrt n \sum_{j=0}^{n-1}\cos\big(\frac{2\pi jk}{n}\big)|j\ra$ ($k\neq 0$), as candidates for optimal solutions to the above problem. Consequently, we can compute the value
\[\sigma(C_n) = 1-\min_k \cos(2\pi k/n) =
\begin{cases}
    1 + \cos (\pi/n) &\text{ for $n$ odd},\\
    2 &\text{ for $n$ even}.
\end{cases}
\]
\end{proof} Finally, we remark that as $n \to \infty$, $\sigma(C_n) \to 2$ which is equal to $1+\frac{1}{\omega(C_n)-1}$ for all $n$. 

\subsection{Small graphs}

In this section we list graphs with a small number of vertices $n$ which lead to positive indecomposable maps via the strict inequality 

$$\sigma(G) < 1 + \frac{1}{\omega(G)-1},$$
where $\sigma(G)$ is the quantity defined in \cref{eq:def-sigma} and $\omega(G)$ is the clique number of $G$. 

For $n \leq 4$ there are no graphs for which the inequality above holds since $\SPN_n=\COP_n$. For $n=5$, the $5$-cycle is the only such graph, see Proposition \ref{prop:cycle-graph}. For $n=6$, there are three such graphs, that we list below, see also \cref{fig:gap-6}.

\begin{itemize}
    \item  the \emph{square + path} graph, with $\sigma(G) = (5+\sqrt 5)/4$ and $\omega(G)=2$
    \item  the \emph{tadpole graph} $T_{5,1}$, with $\sigma(G) = (5+\sqrt 5)/4$ and $\omega(G)=2$
    \item  the \emph{wheel graph} $W_6$, with $\sigma(G) = 1+1/\sqrt{5}$ and $\omega(G)=3$.
\end{itemize}

\begin{figure}[htb]
    \centering
    \includegraphics[width=0.25\linewidth]{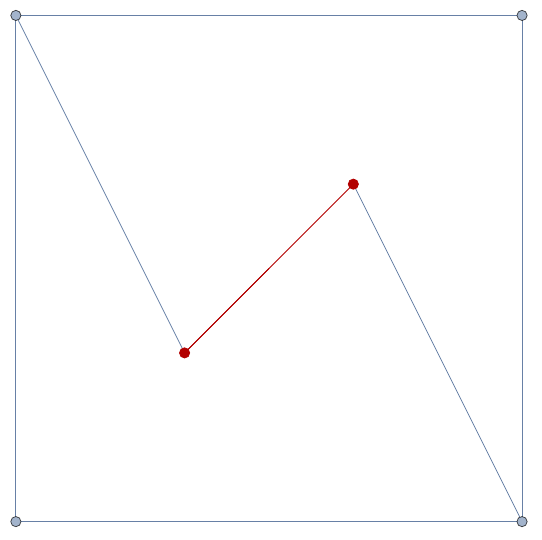}\qquad\qquad
    \includegraphics[width=0.25\linewidth]{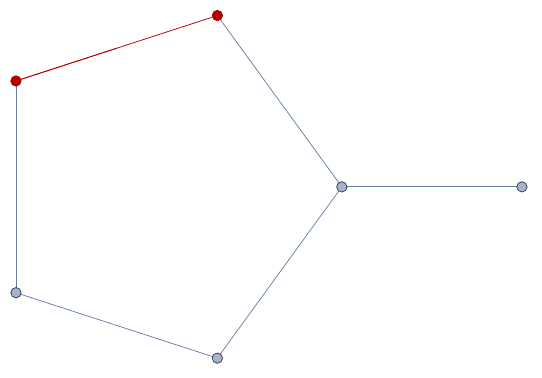}\qquad\qquad
    \includegraphics[width=0.25\linewidth]{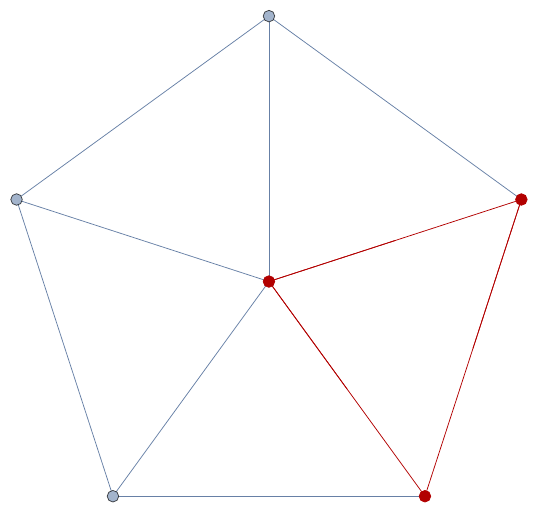}
    \caption{Graphs on 6 vertices leading to positive indecomposable maps. From left to right, the square + path graph, the tadpole graph, and the wheel graph. A maximal clique of each graph is highlighted in red. Note that the wheel graph is the first graph not covered by the Proposition \ref{prop:triangle-free-bipartite} as it is not triangle-free. }
    \label{fig:gap-6}
\end{figure}
Let us prove that $\sigma(G) = 1+1/\sqrt 5$ for the wheel graph $G=W_6$, leaving the other two cases for the reader. We shall exhibit primal and dual feasible points for the semidefinite programs from \cref{eq:def-sigma} and \cref{eq:sigma-dual} that correspond to lower, resp.~upper bounds for $\sigma(G)$. Assigning the last index (6) to the center of the wheel, the adjacency matrix of $G$ is given by 
$$A = 
    \begin{bmatrix}
     0 & 1 & 0 & 0 & 1 & 1 \\
     1 & 0 & 1 & 0 & 0 & 1 \\
     0 & 1 & 0 & 1 & 0 & 1 \\
     0 & 0 & 1 & 0 & 1 & 1 \\
     1 & 0 & 0 & 1 & 0 & 1 \\
     1 & 1 & 1 & 1 & 1 & 0 \\
    \end{bmatrix}.
$$
For the primal feasible point, consider the decomposition $J - (1+1/\sqrt 5) A = P + E$, where
$$P = \frac{1}{\sqrt 5} J + \Big(1 - \frac{1}{\sqrt 5} \Big) I - \frac{2}{\sqrt 5} A \quad \text{ and } E = \Big(1-\frac{1}{\sqrt 5}\Big)(J-I-A).$$
Clearly $E \in \EWP_6$, while a simple computation shows that the eigenvalues of $P$ are $(2,2,2,0,0,0)$, thence $P \in \PSDR_6$, proving that $\sigma(G) \geq 1+1/\sqrt 5$ by using \cref{eq:def-sigma}.

For the upper bound, consider the matrix 
$$X = 
\begin{bmatrix}
 b & a & 0 & 0 & a & b \\
 a & b & a & 0 & 0 & b \\
 0 & a & b & a & 0 & b \\
 0 & 0 & a & b & a & b \\
 a & 0 & 0 & a & b & b \\
 b & b & b & b & b & c \\
\end{bmatrix}, \qquad \text{ with } 
\begin{cases}
    a = \displaystyle{\frac{3-\sqrt 5}{20}}\\
    b = \displaystyle{\frac{\sqrt 5-1}{20}}\\
    c = \displaystyle{\frac{5-\sqrt 5}{20}}.
\end{cases}
$$
One can show that the eigenvalues of $X$ are $\big((5-\sqrt{5})/10,(3 \sqrt{5}-5)/20,(3 \sqrt{5}-5)/20,0,0,0\big)$, proving that $X \in \DNN_6$. We also have $\Tr(AX) = 10(a+b) = 1$ and $\Tr(JX) = 10a+15b+c = 1+1/\sqrt 5$, hence $\sigma(G) \leq 1+1/\sqrt 5$ by \cref{eq:sigma-dual}, finishing the proof. 

\medskip

For $n=7$, we have numerically found 33 such graphs using an SDP solver. The code is available at \cite{PosMapsCOP2025}.

\subsection{Strongly regular graphs}

A \emph{strongly regular graph} (SRG) is a graph that exhibits a specific kind of regularity \cite{brouwer2022strongly},  \cite[Chapter 10]{godsil2013algebraic}. Beyond being $k$-regular, the number of common neighbors between any two distinct vertices depends only on whether they are adjacent or not. In this section we shall use the symmetry properties of SRG to compute the decomposability parameter $\sigma$ from \eqref{eq:def-sigma} in order to put forward large families of positive indecomposable maps $\Phi^G_t$ associated to strongly regular graphs. 

The structural properties of an SRG $G$ are precisely captured by four parameters $(n,k,\lambda,\mu)$:
\begin{itemize}
    \item $n$ is the number of vertices of $G$
    \item every vertex of $G$ is of degree $k$ ($G$ is $k$-regular)
    \item every pair of \emph{adjacent} vertices in $G$ have $\lambda$ common neighbors
    \item every pair of \emph{non-adjacent} vertices in $G$ have $\mu$ common neighbors.
\end{itemize}

Strongly regular graphs are fascinating objects that appear in various branches of mathematics, including combinatorics (e.g.~in the study of combinatorial designs and error-correcting codes), group theory (often arising as orbital graphs of permutation groups), and algebraic graph theory, where their highly symmetric structure leads to interesting spectral properties. The existence and the number of non-isomorphic SRG for given parameters is subject to ongoing research \cite{brouwer2022strongly, brouwerSRGlist}. See \cref{fig:SRG-example} for two famous strongly regular graphs.

\begin{figure}[htb]
    \centering
    \includegraphics[width=0.3\linewidth]{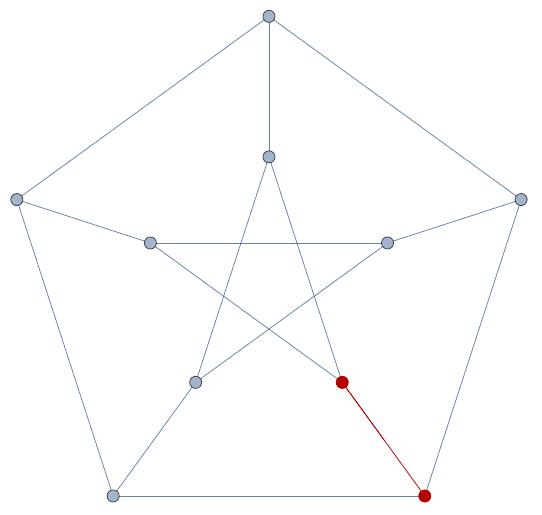}
    \qquad\qquad\qquad\qquad
    \includegraphics[width=0.3\linewidth]{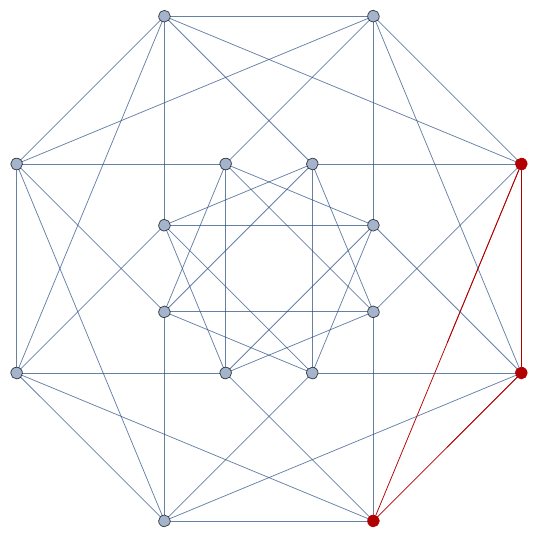}
    \caption{The Petersen (left) and the Shrikhande (right) are SRG with respective parameters $(10, 3, 0, 1)$ and $(16,6,2,2)$. Maximal cliques are highlighted in red.}
    \label{fig:SRG-example}
\end{figure}

By definition, the adjacency matrix $A$ of a SRG satisfies the following properties:
\begin{itemize}
    \item $AJ = JA = kJ$, since each vertex has degree $k$;
    \item $A^2 = kI + \lambda A + \mu (\mathring J - A)$. This equation follows from considering the number of paths of length 2 between any two vertices.
\end{itemize}

If $G$ is connected, its adjacency matrix $A$ typically has three distinct eigenvalues \cite[Chapter 10.2]{godsil2013algebraic}.
The largest eigenvalue is $k$, the degree of each vertex. This eigenvalue has multiplicity $1$ (if $G$ is connected and $n>1$) and its corresponding eigenvector is the all-ones vector $\mathbf{1}$.
The other two eigenvalues, commonly denoted $r$ and $s$, are given by the formula:
$$r,s = \frac{\lambda-\mu \pm \sqrt{(\lambda-\mu)^2+4(k-\mu)}}{2}.$$
By convention, $r$ is taken with the '$+$' sign, so $r \geq 0 \geq  s$. For example, the two graphs in \cref{fig:SRG-example} have the following spectra:
\begin{itemize}
    \item the Petersen graph has eigenvalues $k=3$ with multiplicity 1, $r=1$ with multiplicity 5, and $s=-2$ with multiplicity 4;
    \item the Shrikhande graph has eigenvalues $k=6$ with multiplicity 1, $r=2$ with multiplicity 6, and $s=-2$ with multiplicity 9.
\end{itemize}

Let us also note that if $G$ is a SRG with parameters $(n,k,\lambda,\mu)$, then the complementary graph $\bar G$ is also strongly regular, with parameters 
$$\bar n = n, \quad \bar k = n-k-1, \quad \bar \lambda = n+\mu - 2(k+1), \quad \bar \mu = n+\lambda - 2k.$$
The matrices $A_G$ and $A_{\bar G}$ commute, with $A_{\bar G}$ having non-trivial eigenvalues corresponding to those of $A_G$: 
\begin{equation}\label{eq:ev-bar-G}
    \bar r = - r - 1 \quad \text{ and } \quad \bar s = -s-1;
\end{equation}
we refer to \cite[Chapter 10]{godsil2013algebraic} for the spectral properties of SRGs.
\medskip

We now show how the symmetry properties of certain SRGs enable us to compute exactly the decomposability parameter $\sigma$ defined in \cref{eq:def-sigma}. Recall the matrix space $C_G$ considered previously in \cref{eq:CG-def}

$$\mathsf C_G := \{X \in \Mrealsa{n} \, : \, [X, P_\pi] = 0 \quad \forall \pi \in \Aut(G)\}.$$

Clearly, $I, J, A_G \in \mathsf C_G$ and $A_{\bar{G}}=J-I-A_G\in \mathsf{C}_G$. We are interested in SRGs for which these 3 matrices form a basis of the vector space $\mathsf C_G$.  

\begin{definition}
    A strongly regular graph $G$ is said to be \emph{rank 3} if the automorphism group $\Aut(G)$ acts on the vertex set of $G$ as a rank 3 permutation group. In other words, the action of $\Aut(G)$ on $V(G) \times V(G)$ partitions it into exactly 3 orbits:
    \begin{itemize}
        \item the diagonal $\{(v,v) \, : \, v \in V(G)\}$
        \item adjacent vertices $\{(u,v) \, : \, u\sim v, \, u \neq v\}$
        \item non-adjacent vertices $\{(u,v) \, : \, u\nsim v, \, u \neq v\}$.        
    \end{itemize}
\end{definition}

There has been considerable effort in characterizing rank 3 SRGs, see \cite[Chapter 11]{brouwer2022strongly}. For example, the Petersen graph from \cref{fig:SRG-example} is rank 3, while the Shrikhande graph has rank 4: non-adjacent vertices have 2 common neighbors that can be either adjacent or non-adjacent (see \cite[Chapter 10.6]{brouwer2022strongly}).

For rank 3 SRGs, the vector space $\mathsf C_G$ is isomorphic to $\R{3}$:
\begin{equation}\label{eq:CG-rank-3}
    \mathsf C_G = \{X_{\alpha \beta \gamma}:=\alpha I + \beta A_G + \gamma A_{\bar G} \, : \, \alpha, \beta, \gamma \in \R{}\}.    
\end{equation}

We shall now use the symmetry properties of (rank 3) SRGs to provide a simpler characterization of the parameter $\sigma$. By Proposition \ref{prop:hat-sigma}, we have 
$\sigma(G) = \hat \sigma (G)$ where $$\hat \sigma(G):= \max\{ t \in \R{} \, : \, J - t A_G \in \SPN_G\}.$$
and  $$\SPN_G := \{R + E \, : \, R \in \PSDR_n \cap \mathsf C_G  \text{ and } E \in \EWP_n^\sa \cap \mathsf C_G\}$$

We now have all the ingredients for the main result of this section, an explicit formula for $\sigma(G)$ in the case of rank 3 SRGs. 

\begin{theorem}\label{thm:sigma-rank3-SRG}
    For a rank 3 SRG $G$ with parameters $(n,k,\lambda, \mu)$, we have
    $$\sigma(G) = \frac{n(r+1)}{r(n-1)+k}.$$
\end{theorem}
\begin{proof}
    Using Proposition \ref{prop:hat-sigma} and \cref{eq:CG-rank-3}, we have 
    \begin{align*}
        \sigma(G) = \max \{t \in \R{} \, : \, &J-tA_G = X_{\alpha' \beta' \gamma'} + X_{\alpha'' \beta'' \gamma''} \text{ with } \\
        &X_{\alpha' \beta' \gamma'} \in \PSDR_n \text{ and } X_{\alpha'' \beta'' \gamma''} \in \EWP_n^\sa \}.
    \end{align*}
    Note that $J-tA_G = X_{1, 1-t, 1}$ and, using the eigenvalue formulas for the complementary graph \eqref{eq:ev-bar-G} 
    \begin{align*}
        X_{\alpha' \beta' \gamma'} \in \PSDR_n &\iff
        \begin{cases}
            \alpha' + \beta' k + \gamma'(n-k-1) \geq 0\\
            \alpha' + \beta' r - \gamma'(r+1) \geq 0\\
            \alpha' + \beta' s - \gamma'(s+1) \geq 0
        \end{cases}
        \\
        X_{\alpha'' \beta'' \gamma''} \in \EWP_n^{sa} & \iff \alpha'', \beta'', \gamma'' \geq 0.
    \end{align*}
    We can now recast the semidefinite program definition of $\sigma(G)$ as a much simpler linear program:
    \begin{align*}
        \sigma(G) = \max \quad &t\\
        \text{s.t.}\quad  & \alpha' + \alpha'' = 1\\
        & \beta' + \beta''  = 1-t\\
        & \gamma' + \gamma''  = 1 \\
        & \alpha' + \beta' k + \gamma'(n-k-1) \geq 0\\
        & \alpha' + \beta' r - \gamma'(r+1) \geq 0\\
        & \alpha' + \beta' s - \gamma'(s+1) \geq 0\\
        & \alpha'', \beta'', \gamma'' \geq 0.
    \end{align*}
    A simple analysis shows that one can take $\alpha'' = \beta'' = 0$ and $\alpha'=1$ and rewrite the optimization problem as follows:
    \begin{align}
        \nonumber\sigma(G) = 1-\min \quad &\beta'\\
        \label{eq:LP-1}\text{s.t.}\quad  & 1 + \beta' k + \gamma'(n-k-1) \geq 0\\
        \label{eq:LP-2}& 1 + \beta' r - \gamma'(r+1) \geq 0\\
        \label{eq:LP-3}& 1 + \beta' s - \gamma'(s+1) \geq 0\\
        \label{eq:LP-4}& \gamma' \leq 1.
    \end{align}
    Using $s \leq 0 \leq r$ one can show that \eqref{eq:LP-2} and \eqref{eq:LP-3} $\implies$ \eqref{eq:LP-4}. Hence the feasible set is a triangle determined by the three vertices obtained by solving pairs of equations from Eqs.~\eqref{eq:LP-1}, \eqref{eq:LP-2}, \eqref{eq:LP-3}. Intersecting the equalities from \eqref{eq:LP-2} and \eqref{eq:LP-3} yields $\beta'=1$, while intersecting \eqref{eq:LP-1} with \eqref{eq:LP-2} and \eqref{eq:LP-3} respectively yields 
    $$\beta' = - \frac{n-k+r}{r(n-1) + k} \quad \text{ and } \quad \beta' = - \frac{n-k+s}{s(n-1) + k}.$$
    The first value above is smaller, and a simple computation gives the announced result.
\end{proof}

\begin{corollary}
    For the complete graph $\mathcal K_n$, we have $$\sigma(\mathcal K_n) = \omega(\mathcal K_n)/(\omega(\mathcal K_n) - 1) = n/(n-1).$$
\end{corollary}
\begin{proof}
    The complete graph $\mathcal K_n$ is a rank $2$ SRG with parameters $(n,n-1,n-2,\mu)$, with an undefined parameter $\mu$. The result above still applies, with the positive eigenvalue $r = \lambda_{\max}(J-I) = n-1$.  
\end{proof}

We present in \cref{tbl:small-rank3-SRG} a table of all non-isomorphic rank 3 strongly regular graphs with less than 20 vertices, following \cite[Chapter 11.6]{brouwer2022strongly}. We refer to \cite{brouwer2022strongly,brouwerSRGlist} for the families of graphs appearing in the table. We highlight in the last column the graphs $G$ for which positive indecomposable maps of the form $\Phi^G_t$ exist; these graphs are identified by the strict inequality 
$$\sigma(G) < \frac{\omega(G)}{\omega(G)-1},$$
where $\sigma(G)$ is computed using \cref{thm:sigma-rank3-SRG}. Note that the Hamming graph of parameters (2,4), which is the direct product of two complete graphs $\mathcal K_4$, is a strongly regular graph of rank 3 having parameters $(16,6,2,2)$; these parameters are identical to those of the Shrikhande graph from \cref{fig:SRG-example}, proving that these two SRGs having the same parameters are non-isomorphic.

\begin{table}[htb]
{\renewcommand{\arraystretch}{2}
\begin{tabular}{|r|c|c|c|c|c|}
\hline
\rowcolor[HTML]{EFEFEF} 
Name & Graph & Parameters & $\omega$ & $\sigma$ & gap \\ \hline
pentagon $C_5$ = Paley(5) & \tabfigure{scale=0.15}{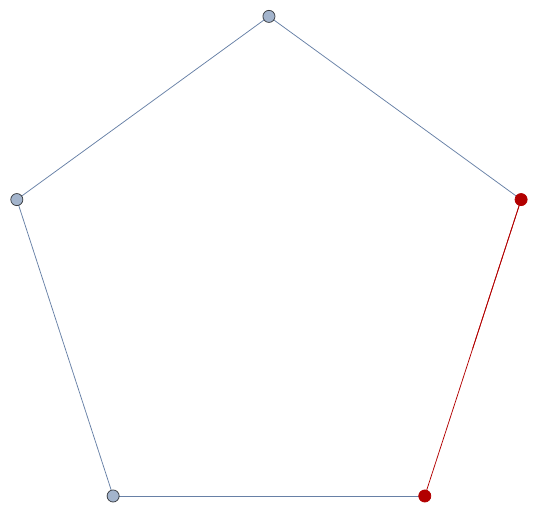} &  $(5,2,0,1)$ & 2 & $\displaystyle{\frac{5+\sqrt 5}{4}}$ & \textcolor{green}{\CheckmarkBold} \\ \hline
Paley(9) & \tabfigure{scale=0.15}{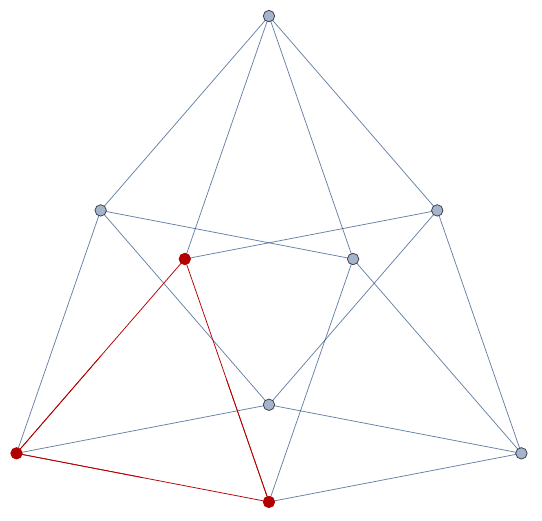} &  $(9,4,1,2)$ & 3 & $\displaystyle{\frac{3}{2}}$ & \textcolor{red}{\XSolidBrush} \\ \hline
Petersen & \tabfigure{scale=0.15}{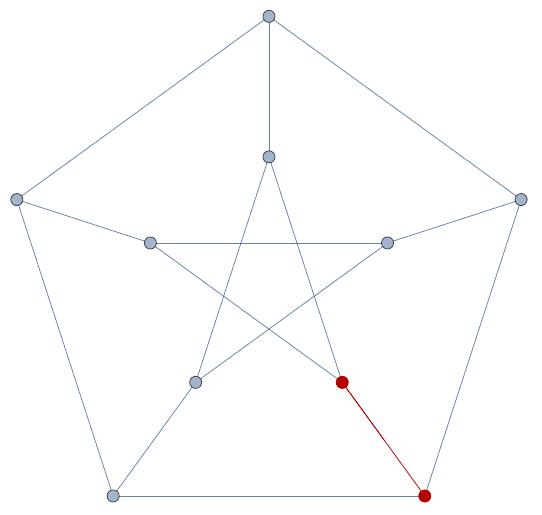} &  $(10,3,0,1)$ & 2 & $\displaystyle{\frac{5}{3}}$ & \textcolor{green}{\CheckmarkBold} \\ \hline
complement of Petersen & \tabfigure{scale=0.15}{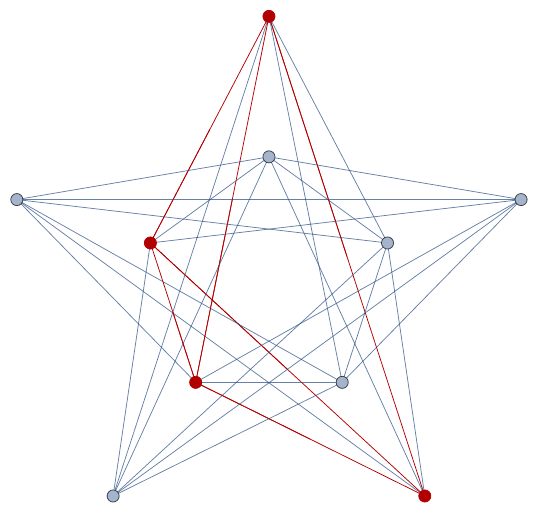} &  $(10,6,3,4)$ & 4 & $\displaystyle{\frac{4}{3}}$ & \textcolor{red}{\XSolidBrush} \\ \hline
Paley(13) & \tabfigure{scale=0.15}{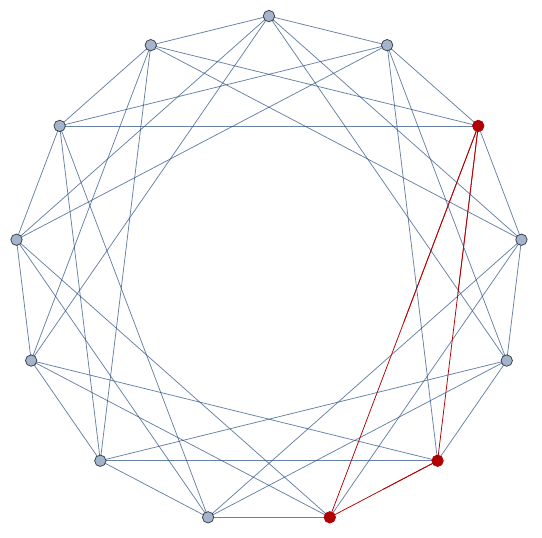} &  $(13,6,2,3)$ & 3 & $\displaystyle{\frac{13+\sqrt{13}}{12}}$ & \textcolor{green}{\CheckmarkBold} \\ \hline
GeneralizedQuadrangle(2,2) & \tabfigure{scale=0.15}{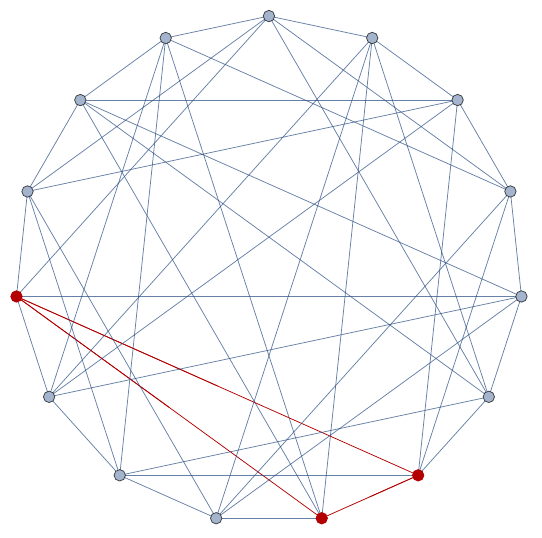} &  $(15,6,1,3)$ & 3 & $\displaystyle{\frac{3}{2}}$ & \textcolor{red}{\XSolidBrush} \\ \hline
Triangular(6) & \tabfigure{scale=0.15}{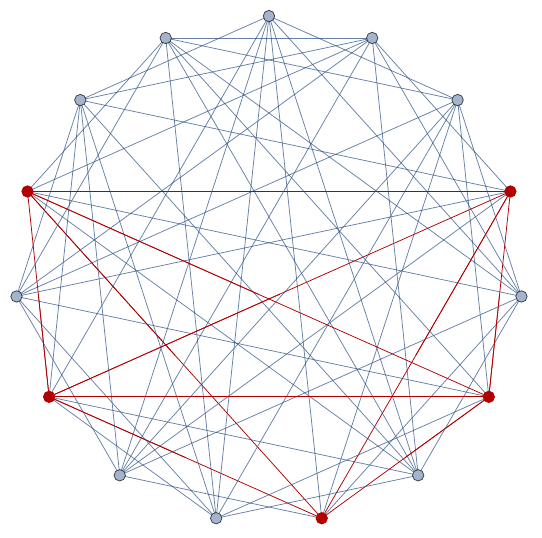} &  $(15,8,4,4)$ & 5 & $\displaystyle{\frac{5}{4}}$ & \textcolor{red}{\XSolidBrush} \\ \hline
Clebsch & \tabfigure{scale=0.15}{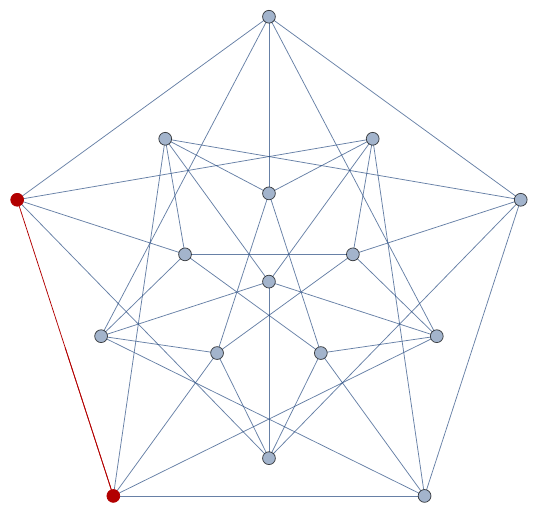} &  $(16,5,0,2)$ & 2 & $\displaystyle{\frac{8}{5}}$ & \textcolor{green}{\CheckmarkBold} \\ \hline
complement of Clebsch & \tabfigure{scale=0.15}{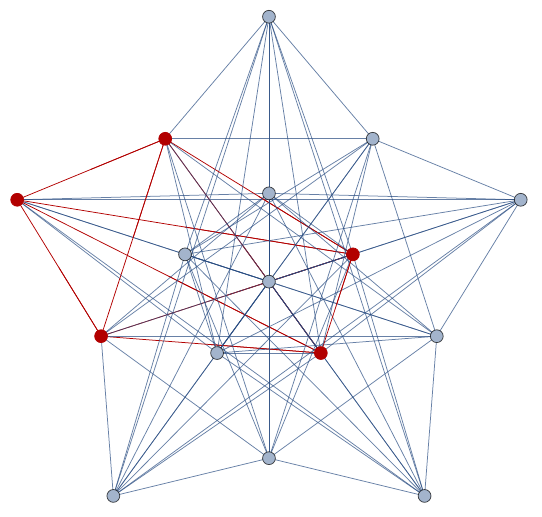} &  $(16,10,6,6)$ & 5 & $\displaystyle{\frac{6}{5}}$ & \textcolor{green}{\CheckmarkBold} \\ \hline
Hamming(2,4) & \tabfigure{scale=0.15}{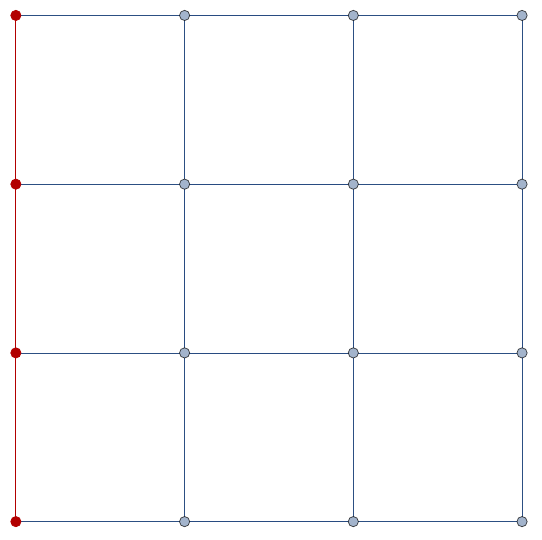} &  $(16,6,2,2)$ & 4 & $\displaystyle{\frac{4}{3}}$ & \textcolor{red}{\XSolidBrush} \\ \hline
complement of Hamming(2,4) & \tabfigure{scale=0.15}{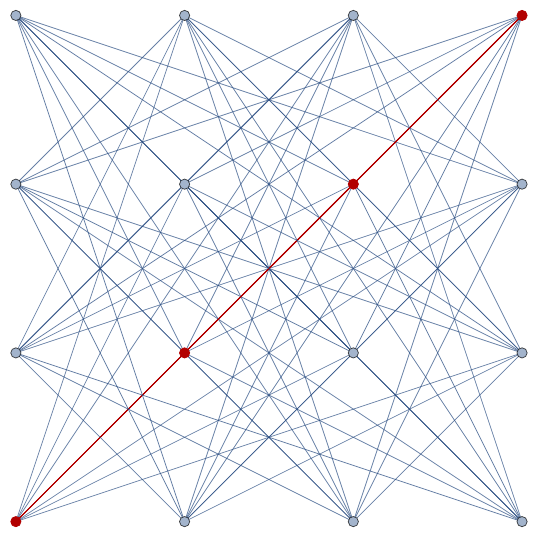} &  $(16,9,4,6)$ & 2 & $\displaystyle{\frac{4}{3}}$ & \textcolor{red}{\XSolidBrush} \\ \hline
Paley(17) & \tabfigure{scale=0.15}{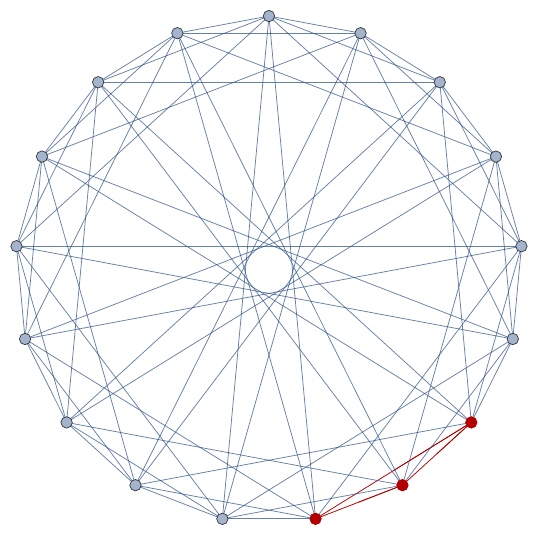} &  $(17,8,3,4)$ & 3 & $\displaystyle{\frac{17+\sqrt{17}}{16}}$ & \textcolor{green}{\CheckmarkBold} \\ \hline
\end{tabular}
}
\caption{Rank 3 strongly regular graphs with $\leq 20$ vertices along with their parameters. The last column indicates whether $\sigma(G) < \omega(G)/(\omega(G)-1)$, i.e.~whether positive indecomposable maps $\Phi^G_t$ exists.}
\label{tbl:small-rank3-SRG}
\end{table}

Let us conclude this section by providing an infinite family of SRGs which yield positive indecomposable maps. Recall that \emph{Paley graphs} of order $q$ are defined for prime powers $q$ such that $q = 1 \, (\textrm{mod }4)$. The vertex set of $\Paley(q)$ is the finite field $\mathbb F_q$, and two vertices are adjacent if their difference is a (non-zero) square in $\mathbb F_q$. $\Paley(q)$ is a strongly regular graphs with parameters 
$$\Big(q, \frac{q-1}{2}, \frac{q-5}{4}, \frac{q-1}{4} \Big)$$
having eigenvalues 
$$r,s = \frac{-1 \pm \sqrt q}{2}$$
with multiplicities $(q-1)/2$. Importantly, Paley graphs are self-complementary, and vertex- and edge-transitive; hence they are rank 3. We refer the reader to \cite[Chapter 7.4]{brouwer2022strongly}. \cref{thm:sigma-rank3-SRG} gives 
$$\sigma(\Paley(q)) = 1 + \frac{1}{\sqrt q - 1} = \frac{q + \sqrt q}{q-1}.$$
Computing the clique (or independence) number of Paley graphs is a very difficult problem \cite[Chapter 7.4.4]{brouwer2022strongly}. The Hoffman spectral bound gives
\begin{equation}\label{eq:omega-Paley-trivial-bound}
    \omega(\Paley(q)) \leq \sqrt q,
\end{equation}
which is known in the literature as the ``trivial bound''. This is relevant for the existence of positive indecomposable maps given by Paley graphs, since 
$$\sigma(\Paley(q)) < 1 + \frac{1}{\omega(\Paley(q))-1} \iff \omega(\Paley(q)) < \sqrt q.$$
In other words, a Paley graph $G$ provides positive indecomposable maps $\Phi^G_t$ for some values of $t$ if and only if it does not saturate the trivial bound for the clique number! It is known that if $q$ is an even prime power, then the trivial bound \eqref{eq:omega-Paley-trivial-bound} is saturated \cite{broere1988clique}. However, if $q=p$ is prime, it has been shown in \cite{hanson2021refined} that 
$$\omega(\Paley(p)) \leq \frac{\sqrt{2p-1}+1}{2}<\sqrt p,$$
which provides us with the following corollary. 
\begin{corollary}\label{cor:Paley}
    For any prime number $p$, and every 
    $$t \in \Big( 1 + \frac{1}{\sqrt p -1}, 1 + \frac{2}{\sqrt{2p-1}-1} \Big],$$
    the linear map $\Phi^{\Paley(p)}_t$ is positive and indecomposable. 
\end{corollary}

Note that the considerations above match the data from \cref{tbl:small-rank3-SRG}: the Paley graphs with parameters $q=5,13,17$ have a gap, while the perfect square $q=9$ does not.    

\section{Bosonic extendibility and SOS hierarchies} \label{sec:sos-hierarchies}
In this section, we will shift our focus to detect the extendibility properties of quantum states. This notion is closely connected to the problem of entanglement detection as it provides useful necessary criteria for separability of quantum states. We introduce various different notions  of quantum extensions, and make interesting connections to the hierarchies that have been well-studied for copositive matrices. This allows us to apply known results from polynomial optimization to the problem of entanglement detection. 

It is known that checking copositivity is a co-NP complete problem \cite{murty1987some,DG14}. Hierarchies of polynomially-sized semidefinite programs (SDP) or sum-of-squares (SOS) programs can be used to approximate the cone $\mathsf{COP}_n$. These hierarchies either provide inner and outer approximations that converge to the actual cone. See, for example, the works of \cite{deklerk2002approximation} and \cite{schweighofer2024sum} for details on such approximation schemes.
First of all, we recall the hierarchy of inner approximations to $\COP_n$ based on sum of squares \cite{parrilo2000structured}: 
    $$\K_n^{(r)}=\{X\in \mathcal{M}_n^{\rm sa}(\mathbb{R}):x\in \Real^n\mapsto \la x^{\odot 2}|X|x^{\odot 2}\ra\|x\|_2^{2r} \text{ is SOS}\}.$$
The following results are known about this hierarchy:
\begin{itemize}
    \item \cite{parrilo2000structured} $\K_n^{(0)}= \SPN_n=\PSDR_n+\EWP_n^\sa $. In particular we have, $\COP_n=\K_n^{(0)}$ for $n\leq 4$.

    \item \cite{parrilo2000structured,DDGH13} If $n\geq 5$, $\K_n^{(0)}\subsetneq \K_n^{(1)}\subseteq\cdots\subseteq \K_n^{(r)}\subsetneq \COP_n$. 

    \item \cite{laurent2023exactness,schweighofer2024sum} $\bigcup_{r\geq 0}\K_5^{(r)}=\COP_5$ and ${\rm int}(\COP_n)\subsetneq \bigcup_{r\geq 0}\K_n^{(r)}\subsetneq \overline{\bigcup_{r\geq 0}\K_n^{(r)}}=\COP_n$ for $n\geq 6$.
\end{itemize}

Let us mention that connection between hierarchies of semidefinite programs for polynomial optimization and hierarchies of separability tests in quantum information has been previously considered in \cite{DPS04,fang2021sum}, where the well-known Doherty-Parrilo-Spedalieri (DPS) hierarchy for detecting entanglement has been shown to be  dual to a sum-of-squares (SOS) representations for certain Hermitian polynomials. This duality has been used to establish quantitative bounds on the convergence rate of the DPS hierarchy. In this section, we establish hierarchies for LDUI/CLDUI witnesses that detect both PPT-bosonic extendibility \cite{DPS04,navascues2009power} and complete graph extendibility \cite{ACG+23+}, and we provide a complete characterization of such witnesses using the copositive hierarchy. In particular, we provide surprisingly wide family of entanglement witnesses which are not certifiable by any level of the DPS hierarchy, thus providing counterexamples to the question dating back to \cite{DPS04}. We also conjecture (see Question \ref{question:ExtLifting}) that the copositive hierarchy can be lifted to provide explicit witnesses for the DPS extendibility hierarchy interpolating the ``lifting" in \cref{thm:COPCP-from-COP} and \cref{thm:PDEC-from-SPN}.

\subsection{DPS hierarchy} \label{sec:DPS}
We use $\bigvee^r \Comp^n$ to denote the $r$-fold symmetric subspace of $(\Comp^n)^{\otimes r}$. {We start by defining the DPS hierarchy, or the hierarchy of bosonic (or bose-symmetric) extendibility}

\begin{definition}
\label{def:star-graph-extendibility}
    A state $\rho \in \mathsf{PSD}(\C{n} \otimes \C{n})$ is called $r$-bosonic extendible if there exists a positive operator $\tilde \rho \in \mathsf{PSD}(\Comp^n \otimes \bigvee^r \Comp^n)$ such that for all $1 \leq i \leq r$, the marginal state of system $\{0,i\}$ is $\rho$:
    $$\forall \,  1 \leq i \leq r \quad \rho_{0i}:=\id_n\otimes \operatorname{Tr}_{[r] \backslash \{i\}}\big(\tilde \rho \big) = \rho,$$
where we label the subsystems of $\Comp^n\otimes \bigvee^r \Comp^n$ by ${0,1,\ldots, r}$. Similarly, we call a state $r$-PPT bosonic extendible, if such an extension is also required to be PPT across all bipartitions. We will denote the set of $r$-bosonic extendible and its PPT version as 
$\mathsf{BExt}^{(r)}_n$ and  
$\mathsf{PPTBExt}^{(r)}_n$ respectively.
\end{definition} 

\begin{figure}[htb]
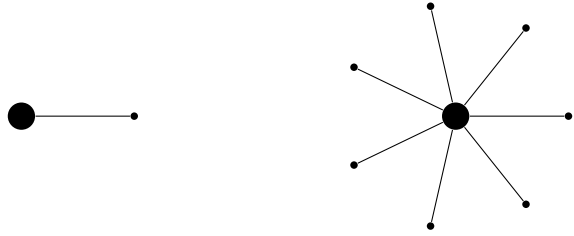

\centering
\stargraph{1}{1.5} \hspace{7em}\stargraph{7}{1.5}
\caption{On the left, we denote the state a bipartite $\rho$, and on the right, the the notion of $r$-extension that we consider in Definition \ref{def:star-graph-extendibility}. The marginal over any edges of the graph should be equal to $\rho$. Moreover, we demand that the systems corresponding to the nodes of the star graph together live in the bosonic subspace $\mathsf{BExt}^{(r)}_n$, and for PPT extendible states $\mathsf{PPTBExt}^{(r)}_n$, we impose the PPT condition across all bipartitions.}
\end{figure}

We have the following inclusion and convergence relations between the sets,

\[
\begin{array}{ccccccccccc}
% First line of the diagram
\mathsf{Sep}(\mathbb{C}^n \otimes \mathbb{C}^n)
& \subseteq & \mathsf{BExt}_n^{(r)}
& \subseteq & \cdots
& \subseteq & \mathsf{BExt}_n^{(2)}
& \subseteq & \mathsf{BExt}_n^{(1)}
& = & \mathsf{PSD}(\mathbb{C}^n \otimes \mathbb{C}^n) \\[8pt]

% Second line with upward symbols
\upwardeq &  & \upwardsubseteq & & \upwardsubseteq & & \upwardsubseteq & & \upwardsubseteq & & \upwardsubseteq \\[8pt]

\vspace{10pt}

% Third line of the diagram
\mathsf{Sep}(\mathbb{C}^n \otimes \mathbb{C}^n)
& \subseteq & \mathsf{PPTBExt}_n^{(r)}
& \subseteq & \cdots
& \subseteq & \mathsf{PPTBExt}_n^{(2)}
& \subseteq & \mathsf{PPTBExt}_n^{(1)}
& = & \mathsf{PPT}(\mathbb{C}^n \otimes \mathbb{C}^n)
\end{array}
\]

$$\mathsf{Sep}(\Comp^n\otimes \Comp^n)=\bigcap_{r\geq 1} \mathsf{BExt}_n^{(r)} =\bigcap_{r\geq 1} \mathsf{PPTBExt}_n^{(r)},$$
where the last fact follows from \cite{DPS04}.

\begin{remark}
    Note that the separable states have a $r$-bosonic extension and an $r$-PPT bosonic extension for all $r \geq 2$. For a separable state $\rho = \sum_k \ketbra{v_k}{v_k} \otimes \ketbra{w_k}{w_k}$ we have the following extension that is also PPT across all bipartitions, 
    $$\tilde \rho = \sum_k \ketbra{v_k}{v_k} \otimes \ketbra{w_k}{w_k}^{\otimes r}$$
\end{remark}

Let us first recall the following remarkable characterization of \emph{bosonic extendibility witnesses} in terms of the SOS hierarchy for a certain Hermitian form. The precise connection was implicit in \cite{DPS04,doherty2014entanglement} and later made explicit in \cite{fang2021sum}.

\begin{theorem}
\label{thm:ppt-extension-bosonic}
Let $W \in \Msa{n} \otimes \Msa{n}$ be a bipartite Hermitian operator. Then for $r\geq 1$,
\begin{enumerate}
    \item $W \in (\mathsf{BExt}_n^{(r)})^{\circ}$ if and only if there exists a positive semidefinite matrix $Z \in \M{n}^{\otimes (r+1)}$ such that, for all vectors $v,w \in \Comp^n$,
    \begin{equation} \label{eq:HermSOS1}
        \langle v w \,|\, W \,|\, v w \rangle \cdot \|w\|_2^{2(r-1)} \;=\;  \langle v w \cdots w \,|\, Z \,|\, v w \cdots w \rangle.
    \end{equation}

    \item $W\in (\mathsf{PPTBExt}_n^{(r)})^{\circ}$ if and only if there exist psd matrices $Z_0,Z_1,\ldots, Z_r\in \M{n}^{\otimes (r+1)}$ such that, for all vectors $v,w\in \Comp^n$,
    \begin{align}
        \la vw|W|vw\ra \|w\|_2^{2(r-1)} &= \la v w\cdots w|Z_0+Z_1^{\Gamma_{[1]}}+\cdots + Z_r^{\Gamma_{[r]}}|v w \cdots w\ra \nonumber\\
        &=\la vw^{\otimes r}|Z_0|vw^{\otimes r}\ra + \sum_{j=1}^r \la v\bar{w}^{\otimes j}w^{\otimes r-j}|Z_j|v\bar{w}^{\otimes j}w^{\otimes r-j}\ra \label{eq:HermSOS2},
    \end{align}
    where $\Gamma_{[j]}$ is the partial transpose operation with respect to the tensor bipartition $\{1,\ldots, j\}:\{0,j+1,\ldots, r\}$.
\end{enumerate}
\end{theorem}

Note that for all $r\geq 0$, we have the following inclusions: 
    $$(\mathsf{BExt}_n^{(r)})^{\circ}\subseteq (\mathsf{PPTBExt}_n^{(r)})^{\circ} \subseteq \mathsf{Sep}(\Comp^n\otimes \Comp^n)^{\circ}=\{W\in \Msa{n}^{\otimes 2}:\la vw|W|vw\ra\geq 0 \;\forall v,w\in \Comp^n\},$$
where the last set is the cone of \emph{block-positive} matrices. We also remark that \cref{eq:HermSOS1,eq:HermSOS2} indeed provide SOS decompositions of the Hermitian form $(v,w)\mapsto \la vw|W|vw\ra\|w\|_2^{2(r-1)}$ since they are conic combinations of Hermitian forms of the type
    $$v,w \in \Comp^n\mapsto \big|\la \psi|v\bar{w}^{\otimes j}w^{\otimes r-j}\ra\big|^2, \quad \psi\in (\Comp^n)^{\otimes r+1}, \;\;j=0,1,\ldots, r.$$

Now following the previous arguments in Propositions \ref{prop:properties-COPCP} and \ref{prop:SPN-from-PairDEC}, we further provide a \emph{necessary} condition for the LDUI/CLDUI operators to witness $r$-extendibility based on the copositive SOS hierarchy $(\K_n^{(r)})_{r\geq 0}$.

\begin{proposition} \label{prop:Ext-to-SOS}
Let $(A,B)\in \Mreal{n} \underset{\R{n}}{\times} \Msa{n}$. If $r\geq 1$ and $\XLDUICLDUI{A}{B}\in (\mathsf{PPTBExt}_n^{(r)})^{\circ} $, then
    $$(A+\mathring{B})+(A+\mathring{B})^{\top}=A+A^{\top}+ 2{\rm Re}(\mathring{B})\in \K_n^{(r-1)}.$$
\end{proposition}

\begin{proof}
From \cref{thm:ppt-extension-bosonic}, there exist PSD matrices $Z_0,\ldots, Z_r\in \M{n}^{\otimes r+1}$ such that
    $$h(v,w) :=\la vw|\XLDUICLDUI{A}{B}|vw\ra\|w\|_2^{2(r-1)} = \sum_{j=0}^r \la v\bar{w}^{\otimes j}w^{\otimes r-j}|Z_j|v\bar{w}^{\otimes j}w^{\otimes r-j}\ra, \quad v,w\in \Comp^n.$$
Setting $v = w = x \in \mathbb{R}^n$ and by the relations \cref{eq:LDUI-duality,eq:CLDUI-duality}, we have that the polynomial
\begin{align*}
    h(x,x) = \la x^{\odot 2}|A+\mathring{B}|x^{\odot 2}\ra \|x\|_2^{2(r-1)} =
     \frac{1}{2}\la x^{\odot 2} |A+A^{\top}+2{\rm Re}(\mathring{B})| x^{\odot 2} \ra \|x\|_2^{2(r-1)}
\end{align*}
is SOS. This concludes $A+A^{\top}+2{\rm Re}(\mathring{B})\in \mathsf{K}_n^{(r-1)}$.
\end{proof}
Later in \cref{thm:Ext-from-SOS}, we will establish the precise correspondence between $\XLDUI{A}{B}$ satisfying $A+A^{\top}+2{\rm Re}(\mathring{B})\in \K_n^{(r-1)}$ and the notion of \emph{complete graph extendibility}.

{\medskip
On the other hand, it was questioned in \cite[Section VI]{DPS04} whether the set $\mathsf{Sep}(\Comp^n\otimes \Comp^n)^{\circ}$ can be saturated by the cones $(\mathsf{PPTBExt}_n^{(r)})^{\circ}$, i.e., whether the equality
\begin{equation} \label{eq:eq-sep-dual}
    \mathsf{Sep}(\Comp^n\otimes \Comp^n)^{\circ}=\bigcup_{r\geq 1} (\mathsf{PPTBExt}_n^{(r)})^{\circ}
\end{equation}
holds. Since $\mathsf{Sep}(\Comp^n\otimes \Comp^n)^{\circ}= \big(\bigcap_{r\geq 1}\mathsf{PPTBExt_n^{(r)}}\big)^{\circ}=\overline{\bigcup_{r\geq 1} (\mathsf{PPTBExt}_n^{(r)})^{\circ}}$, the question mathematically asks whether the closure is essential or not. It was shown in \cite[Theorem 2]{DPS04} that the interior of $\mathsf{Sep}(\Comp^n\otimes \Comp^n)^{\circ}$ is contained in $\bigcup_{r\geq 1} (\mathsf{PPTBExt}_n^{(r)})^{\circ}$. In other words, every entanglement witness $W$ which is ``strictly'' block-positive,
    $$\la vw|W|vw\ra>0,\quad \forall\, v,w\neq 0,$$
is certifiable by $r$-SOS property for some finite $r$. Conversely, any witness $W\in \mathsf{Sep}(\Comp^n\otimes \Comp^n)^{\circ} \setminus \bigcup_{r\geq 1} (\mathsf{PPTBExt}_n^{(r)})^{\circ}$, if exists, could be used to detect entanglement of quantum states that are $r$-PPT bosonic extendible for each $r\geq 1$. Moreover, $W$ would be on the boundary in the sense that $\Tr(W\rho)=0$ for some $\rho\in \mathsf{Sep}(\Comp^n\otimes \Comp^n)$, as discussed in \cite{DPS04}. Note that the set $\bigcup_{r\geq 1} (\mathsf{PPTBExt}_n^{(r)})^{\circ}$ is not precisely the interior of $\mathsf{Sep}(\Comp^n\otimes \Comp^n)^{\circ}$: $W= \XLDUI{J_4}{\widetilde{H}-\mathring{J}_4}\in \mathsf{PPT}(\Comp^n\otimes \Comp^n)^{\circ}=(\mathsf{PPTBExt}_n^{(1)})^{\circ}$ from Example \ref{ex:HornPCOP2} while we have
    $$\Tr(W\, |v\ra\la v|^{\otimes 2})=\la v^{\otimes 2}|\XLDUI{J_4}{\widetilde{H}-\mathring{J}_4}|v^{\otimes 2}\ra= \la v^{\odot 2}|\widetilde{H}|v^{\odot 2}\ra =0$$
for $v=\frac{1}{\sqrt{2}}(0,1,1,0)^{\top}\in \Comp^4$, and hence $W\notin {\rm int}\big(\mathsf{Sep}(\Comp^n\otimes \Comp^n)^{\circ}\big)$.

Now we provide the following \emph{explicit ($(n^2-n)/2$-dimensional) counterexamples} to the question in \cref{eq:eq-sep-dual} which strengthens Corollary \ref{cor:PosIndecomp}.

\begin{theorem} \label{thm:EWnotSOS}
Suppose $n\geq 6$ and let $M$ be a copositive matrix which does not have $r$-SOS property for any $r$, i.e., $M\in \COP_n\setminus \bigcup_{r\geq 0}\mathsf{K}_n^{(r)}$. Then for any $N\in \EWP_n^{\sa}$ with $\diag(N)=\diag(M)$ and $N_{ij}\geq \frac{1}{2} M_{ij}$ $\forall\,i\neq j \in [n]$, we have
    $$\XLDUICLDUI{N}{M-\mathring{N}}\in \mathsf{Sep}(\Comp^n\otimes \Comp^n)^{\circ} \setminus \bigcup_{r\geq 1} (\mathsf{PPTBExt}_n^{(r)})^{\circ}.$$
Furthermore, all these entanglement witnesses are on the boundary of $\mathsf{Sep}(\Comp^n\otimes \Comp^n)^{\circ}$.
\end{theorem}
\begin{proof}
The fact $\XLDUICLDUI{N}{M-\mathring{N}}\in \mathsf{Sep}(\Comp^n\otimes \Comp^n)^{\circ}$ follows from Lemma \ref{lemma:CJiso}, Proposition \ref{prop:COPCP-postivie}, and \cref{thm:COPCP-from-COP}. Furthermore, Proposition \ref{prop:Ext-to-SOS} implies that $\XLDUICLDUI{N}{M-\mathring{N}}\notin (\mathsf{PPTBExt}_n^{(r)})^{\circ}$ for any $r\geq 0$. The last assertion follows a priori from \cite[Theorem 2]{DPS04}, as mentioned earlier.
\end{proof}

\begin{remark}
The existence of copositive matrices in $\COP_n\setminus \bigcup_{r\geq 0}\K_n^{(r)}$ for arbitrary $n\geq 6$ was first shown in \cite{laurent2023exactness}. One such example is
    $$M=H\oplus 0_{(n-5)\times (n-5)}\in \COP_n\setminus \bigcup_{r\geq 0}\K_n^{(r)}, \quad n\geq 6,$$
where $H\in \COP_5$ is the Horn matrix defined in \cref{eq:Horn}. For other examples, we refer to \cite[Theorem 3, Example 1 and 3]{laurent2023exactness}. In particular, this implies that
\begin{align*}
    \XLDUICLDUI{J_n}{M-\mathring{J}_n} &\in \mathsf{Sep}(\Comp^n\otimes \Comp^n)^{\circ} \setminus \bigcup_{r\geq 1} (\mathsf{PPTBExt}_n^{(r)})^{\circ}, \\
    \XLDUICLDUI{J_5\oplus 0_{(n-5)\times (n-5)}}{(H-\mathring{J}_5)\oplus 0_{(n-5)\times (n-5)}} &\in \mathsf{Sep}(\Comp^n\otimes \Comp^n)^{\circ} \setminus \bigcup_{r\geq 1} (\mathsf{PPTBExt}_n^{(r)})^{\circ}.
\end{align*}
Note that the latter matrix is a (trivial) lifting of $5\otimes 5$ Hermitian matrix $\XLDUICLDUI{J_5}{H-\mathring{J_5}}$ into $n\otimes n$ system (we also refer to \cite[Proposition 4.8]{singh2021diagonal}). Since
\begin{align*}
    &\la vw|\XLDUI{J_5\oplus 0_{(n-5)\times (n-5)}}{(H-\mathring{J}_5)\oplus 0_{(n-5)\times (n-5)}}|vw\ra \|w\|_2^{2(r-1)} \\
    &=\Big(\sum_{1\leq i<j\leq 5} |v_i\overline{w}_j-v_j\overline{w}_i|^2+\sum_{i,j=1}^5 (v_i\overline{v}_jw_i\overline{w}_{j}H_{ij})\Big)\Big(\sum_{j=1}^n|w_j|^2\Big)^{r-1}
\end{align*}
is SOS for no $n\geq 6$ and $r\geq 1$, this actually shows that
    $$\mathsf{Sep}(\Comp^{n_A}\otimes \Comp^{n_B})^{\circ}\supsetneq \bigcup_{r\geq 1}(\mathsf{PPTBExt}_{n_A,n_B}^{(r)})^{\circ}$$
whenever $n_A\geq 5$ and $n_B\geq 6$, by embedding $\XLDUI{J_5}{H-\mathring{J_5}}$ into proper spaces. 
\end{remark}
}

Recall that we constructed positive maps by lifting copositive matrices to the matrix pairs of the form $(N,M-\mathring{N})$ in \cref{thm:COPCP-from-COP}. We leave open the question whether a similar “lifting” construction of $r$-extendibility witnesses can be carried out from the hierarchies of copositive matrices.
\begin{question} \label{question:ExtLifting}
For $r\geq 2$ and $M\in \K_n^{(r-1)}$, do we have $\XLDUICLDUI{N}{M-\mathring{N}}\in (\mathsf{PPTExt}_n^{(r)})^{\circ}$ for some choice of $N\in \EWP_n^{\rm sa}$? Moreover, is it possible to choose $\mathring{N}=\max(M)\mathring{J}$ similar to the construction in \cref{thm:COPCP-from-COP}?
\end{question}

{Note that the above assertion for $r=1$ already follows from \cref{thm:PDEC-from-SPN}.} Although the complete answer to the above question remains unclear, we show in the next section that another type of hierarchy of extensions, called \emph{complete graph extendibility}, can be fully characterized in terms of the copositive sum of squares hierarchy, generalizing some results from the literature. 

\subsection{Complete extendibility}
Let $\mathcal{K}_{r}$ denotes the complete graph with $r$ vertices. A bosonic quantum state is called \emph{$\mathcal{K}_{r}$-PPT bosonic extendible} (or \textit{$r$-PPT bosonic exchangeable}) if there exists a bosonic operator $\tilde{\rho}\in \mathsf{PSD}\big(\bigvee^{r}\Comp^n\big)$ such that $(\id_n^{\otimes 2}\otimes \Tr^{\otimes r-2})(\tilde{\rho})=\rho$ and $\tilde{\rho}$ is also PPT across every bipartition of the $r$ systems. We refer to \cite{ACG+23+} for general notion of graph extendibility.

\begin{figure}[htb]
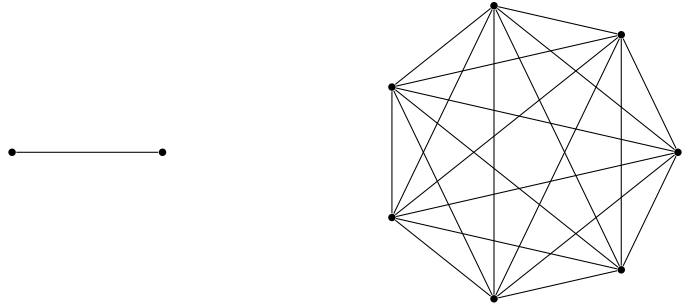

\centering
\completegraph{2}{1} \hspace{7em} \completegraph{7}{2}
\caption{On the left, we denote a bipartite bosonic state $\rho$, and on the right, the the notion of $r$-extension that we consider in Definition \ref{def:star-graph-extendibility}. The marginal over every edge of the graph should be equal to $\rho$. Moreover, we demand that the global state corresponding to the complete graph should live in the bosonic subspace. For PPT extendible states $\mathsf{PPTBExt}^{(\mathcal{K}_r)}_n$, we impose the PPT condition across all bipartitions.}
\end{figure}

For $r\geq 2$, let us denote by $\mathsf{PPTBExt}_n^{(\mathcal{K}_r)}$ the convex cone of $n\otimes n$ bosonic positive-semidefinite operators $\rho$ which are {$\mathcal{K}_{r}$-PPT bosonic extendible}. It is clear that $\mathsf{PPTBExt}_n^{(\mathcal{K}_2)}$ is simply the set of PPT bipartite bosonic operators $\mathsf{PPT}(\Comp^n\vee \Comp^n)$, and we also have an inclusion chain for $\mathcal{K}_r$-bosonic extendible states.
\begin{equation} \label{eq:PPTBExt}
    \mathsf{Sep}(\Comp^n\vee \Comp^n)\subseteq \mathsf{PPTBExt}_n^{(\mathcal{K}_r)}\subseteq \cdots \subseteq\mathsf{PPTBExt}_n^{(\mathcal{K}_3)}\subseteq \mathsf{PPTBExt}_n^{(\mathcal{K}_2)}=\mathsf{PPT}(\Comp^n\vee \Comp^n).
\end{equation}
\begin{figure}[htb]
\centering
$\ldots$
\hspace{2em}
\completegraph{7}{1}
\hspace{2em}
$\ldots$
\hspace{2em}
\completegraph{3}{1}
\hspace{2em}
\completegraph{2}{1}
\end{figure}

Indeed, the set of separable bosonic states are exactly equal the cone generated by product bosonic states \cite{ichikawa2008exchange}, i.e $$\mathsf{Sep}(\Comp^n\vee \Comp^n) = \operatorname{cone}\{\ketbra{v}{v} \otimes \ketbra{v}{v}\}$$ Therefore a separable bosonic state $\rho = \sum_k \ketbra{v_k}{v_k} \otimes \ketbra{v_k}{v_k}$ has a $\mathcal{K}_r$-bosonic extension $\sum_k \ketbra{v_k}{v_k}^{\otimes r}$ for all $r \geq 2$.

Furthermore, $\mathsf{PPTBExt}_n^{(\mathcal{K}_{r+1})}\subseteq \mathsf{PPTBExt}_n^{(r)} \subseteq \mathsf{BExt}_n^{(r)}$ for $r\geq 1$.
Thus, from the previous hierarchy of $(\mathsf{PPTBExt}_n^{(r)})_{r\geq 1}$ or simply from \emph{Quantum de Finetti theorem} \cite{caves2002unknown,christandl2007one}, we again have
    $$\mathsf{Sep}(\Comp^n\vee \Comp^n)=\bigcap_{r\geq 2} \mathsf{PPTBExt}_n^{(\mathcal{K}_r)}.$$
It is natural to ask whether the inclusion in \cref{eq:PPTBExt} would be finite or not. That is, given some local Hilbert space dimension $n$, is there a finite $r\geq 2$ such that
    $$\mathsf{PPTBExt}_n^{(\mathcal{K}_r)}=\mathsf{Sep}(\Comp^n\vee \Comp^n)?$$
If $n=2$ the question is obvious: $\mathsf{PPTBExt}_2^{(\mathcal{K}_2)}=\mathsf{PPT}(\Comp^2\vee \Comp^2)= \mathsf{Sep}(\Comp^2\vee \Comp^2)$. When $n\geq 3$, it is shown that $\mathsf{PPT}(\Comp^n\vee \Comp^n) \supsetneq \mathsf{Sep}(\Comp^n\vee \Comp^n)$ \cite{PhysRevLett.102.170503,marconi2021entangled}. However, explicit construction of examples beyond this, i.e., entangled bose-symmetric states having high degree of extendibility, remains to be a challenging problem.

This section is mainly devoted to finding an explicit element in $\mathsf{PPTBExt}_n^{(\mathcal{K}_r)} \setminus \mathsf{Sep}(\Comp^n\vee \Comp^n)$ for every $n\geq 5$ and $r\geq 2$, that is, an extendible and highly symmetric PPT entangled state. We also show that such a state exists within the class of (mixtures of) bipartite \emph{Dicke states}, \cite{yu2016separability, tura2018separability}. Our result also allows us to construct $r$-partite entangled bosonic states of which are PPT with respect to all bipartitions, by looking at the bosonic extension of a matrix $\mathsf{PPTBExt}_n^{(r-1)}\setminus \mathsf{Sep}_n$ . Due to the fact that it is an extension, the entanglement of this bosonic state can be even verified by checking that the \emph{two-body} marginals are entangled.

\begin{remark} \label{rmk:MultiEntanglement}
    Any $r$-partite bosonic state $\rho\in \mathsf{PSD}(\bigvee^r \Comp^n)$ have only two types of entanglement \cite{augusiak2012entangled,PhysRevLett.102.170503}:
    \begin{enumerate}
        \item $\rho$ is \emph{fully separable}, i.e., $\rho=\sum_{i} |v_i\ra\la v_i|^{\otimes r}$ for some vectors $v_i\in \Comp^n$,
        
        \item $\rho$ is  \emph{genuine entangled} i.e., $\rho$ cannot be written as a convex combination of states which are separable with respect to (possibly different) bipartitions.
    \end{enumerate}
    In particular, any bosonic state which is not fully separable is entangled across any bipartitions of the system. Hence, we do not make a distinction within various notions of multipartite entanglement for this class of states. 
\end{remark}

\medskip

Simiarly as in \cref{thm:ppt-extension-bosonic}, we first derive a SOS correspondence to the dual of $\mathcal{K}_{r}$-PPT bosonic extendibility (i.e the set of witnesses for $\mathcal{K}_{r}$-PPT bosonic extendibility), which we prove later has a correspondence with the SOS hierarchy for LDUI/CLDUI witnesses. 

\begin{theorem} \label{thm:PPTExtDual}
For a Hermitian operator $W\in \Msa{n}\otimes \Msa{n}$ and for $r\geq 2$, the following are equivalent.
\begin{enumerate}
    \item $W\in (\mathsf{PPTBExt}_n^{(\mathcal{K}_r)})^{\circ}$;

    \item There exist positive semidefinite matrices $Z_0,Z_1,\ldots, Z_{\lfloor r/2\rfloor}\in \M{n}^{\otimes r}$ such that, for all vectors $v\in \Comp^n$,
    \begin{equation} \label{eq:PPTExtDual}
        \la v^{\otimes 2}|W|v^{\otimes 2}\ra \|v\|_2^{2(r-2)} = \la v^{\otimes r}|Z_0+Z_1^{\Gamma_{[1]}}+\cdots + Z_{\lfloor r/2\rfloor}^{\Gamma_{\left[\lfloor r/2 \rfloor\right]}}| v^{\otimes r} \ra =\sum_{j=0}^{\lfloor r/2 \rfloor} \la \overline{v}^{\otimes j}v^{\otimes r-j}| Z_j|\overline{v}^{\otimes j}v^{\otimes r-j}\ra,
    \end{equation}
    where $\Gamma_{[j]}=\top^{\otimes j}\otimes \id_n^{\otimes r-j}$ is the partial transpose operation with respect to the tensor bipartition $\{[j], [r]\setminus[j]\}$.
\end{enumerate}
\end{theorem}
\begin{proof}
{
\textbf{((1)$\Rightarrow$(2))} If $W\in (\mathsf{PPTBExt}_n^{(\mathcal{K}_r)})^{\circ}$, then by definition we have
    $$\Tr\big((W\otimes I_n^{\otimes r-2})\tilde{\rho}\big) = \Tr(W\rho)\geq 0,$$
where $\rho\in \mathsf{PPTBExt}_n^{(\mathcal{K}_r)}$ and $\tilde{\rho}$ is any $\mathcal{K}_r$-PPT bosonic extension of $\rho$.
Note that $\tilde{\rho}\in \mathcal{B}\big( \bigvee^r \Comp^n\big)^+$ can be chosen to have a constraint that $\tilde{\rho}^{\top_{[j]}}\geq 0$ for all $j=1, \ldots, \lfloor r/2\rfloor$, that is indeed the the condition that $\tilde{\rho}$ is PPT with respect to every bipartition. We claim that there exist PSD operators $Z_0,\ldots, Z_{\lfloor r/2\rfloor}\in \M{n}^{\otimes r}$ such that
    $$\Pi_S(W\otimes I_n^{\otimes r-2})\Pi_S= \Pi_S \Big(Z_0+\sum_{j=1}^{\lfloor r/2\rfloor} Z_j^{\Gamma_{[j]}} \Big)\Pi_S,$$
where $\Pi_S$ is a projection of $(\Comp^n)^{\otimes r}$ onto the symmetric space $(\Comp^n)^{\vee r}$. 
This will imply that
    $$\la v^{\otimes 2}|W|v^{\otimes 2}\ra\|v\|_2^{2(r-2)} = \la v^{\otimes r}|\Pi_S(W\otimes I_n^{\otimes r-2})\Pi_S|v^{\otimes r}\ra = \la v^{\otimes r}|Z_0+Z_1^{\Gamma_{[1]}}+\cdots + Z_{r-1}^{\Gamma_{[\lfloor r/2\rfloor]}}| v^{\otimes r} \ra$$
for $v\in \Comp^n$. 
Note that the extension $\tilde \rho \in \bigcap_{j=0}^{\lfloor r/2\rfloor} \mathcal{C}_j$ where $\mathcal{C}_0=\mathsf{PSD}\big( \bigvee^r \Comp^n\big)$ and
    $$\mathcal{C}_j=\{X\in \Msa{n}^{\otimes r}: X^{\Gamma_{[j]}}\geq 0\}.$$
Therefore, the required condition says that $W\otimes I_n^{\otimes {r-2}}\in \big(\bigcap_{j=0}^{\lfloor r/2\rfloor} \mathcal{C}_j\big)^{\circ}=\overline{\sum_{j=0}^{\lfloor r/2\rfloor}\mathcal{C}_j^{\circ}}$. In other words, there exists $Z_0,\ldots, Z_{\lfloor r/2\rfloor}$ such that 
$$W\otimes I_n^{\otimes r-2}=Z_0+\sum_{j=1}^{\lfloor r/2\rfloor} Z_j^{\Gamma_{[j]}}$$ where $\Pi_SZ_0\Pi_S\geq 0$, $Z_j \geq 0$ for $j\geq 1$ (since  $Z_j^{\top_{[j]}} \in \mathcal{C}_j^{\circ}$), showing the claim
{(The closure can be ignored as we can show that the set of Hermitian polynomials of the type
    $$h_{Z_0,\ldots, Z_{\lfloor r/2\rfloor}}(z)=\sum_{j=0}^{\lfloor r/2 \rfloor} \la \overline{z}^{\otimes j}z^{\otimes r-j}| Z_j|\overline{z}^{\otimes j}z^{\otimes r-j}\ra, \quad z\in \Comp^n,$$
is closed. Indeed, this set is the conic hull of the set of Hermitian forms $z\mapsto |\la \varphi|\overline{z}^{\otimes j}z^{\otimes r-j}\ra|^2$ for $\varphi\in (\Comp^n)^{\otimes r}$ and $j=0,1,\ldots, \lfloor r/2\rfloor$.)}

\medskip

\textbf{((2)$\Rightarrow$(1))} The condition in \cref{eq:PPTExtDual} is equivalent to
\begin{align*}
    \la v^{\otimes r}|\big(W\otimes I_n^{\otimes r-2}-(Z_0+\sum_{j=1}^{\lfloor r/2\rfloor}Z_j^{\Gamma_{[j]}})\big)|v^{\otimes r}\ra = 0 &\quad \forall\,v\in \Comp^n\\
    \iff \Tr\big(X \big(W\otimes I_n^{\otimes r-2}-(Z_0+\sum_{j=1}^{\lfloor r/2\rfloor}Z_j^{\Gamma_{[j]}})\big) \big)=0 &\quad \forall\, X\in {\rm span}\{|v^{\otimes r}\ra \la v^{\otimes r}|:v\in \Comp^n\}.
\end{align*}
It is well-known that ${\rm span}\{|v^{\otimes r}\ra \la v^{\otimes r}|:v\in \Comp^n\}=\mathcal{B}\big( \bigvee^r \Comp^n \big)$ \cite{chiribella2010quantum,harrow2013church}, 
and hence one has
    $$\Pi_S(W\otimes I_n^{\otimes r-2})\Pi_S = \Pi_S(Z_0+Z_1^{\Gamma_{[1]}}+\cdots Z_{\lfloor r/2\rfloor}^{\Gamma_{[\lfloor r/2\rfloor]}})\Pi_S.$$
Now for $\rho\in \mathsf{PPTBExt}_n^{(\mathcal{K}_r)}$ and for $\tilde{\rho}$ a $\mathcal{K}_r$-PPT bosonic extension of $\rho$, it is straightforward that
    $$\Tr(W\rho)=\Tr\big(\Pi_S(Z_0+Z_1^{\top_{[1]}}+\cdots Z_{\lfloor r/2\rfloor}^{\Gamma_{[\lfloor r/2\rfloor]}})\Pi_S\tilde{\rho}\big)= \Tr(Z_0\tilde{\rho})+\sum_{j=1}^{\lfloor r/2\rfloor} \Tr\big(Z_j\tilde{\rho}^{\Gamma_{[j]}}\big)\geq 0,$$
and therefore, $W\in (\mathsf{PPTBExt}_n^{(\mathcal{K}_r)})^{\circ}$.
}
\end{proof}

Next, let us recall the following characterization of SOS polynomials from \cite[Proposition 9]{ZVP06}.

\begin{proposition} \label{prop:SOS}
Let $p(x)$ be a real homogeneous polynomial in $x\in \Real^n$ of degree $r$ such that $p(x^{\odot 2})$ is a sum of squares. Then $p$ admits a decomposition of the form
\begin{equation} \label{eq:SOS}
    p(x)=\sum_{j=0}^{\lfloor r/2\rfloor}\sum_{\substack{\alpha\in \mathbb{N}_0^n\\ |\alpha|=r-2j}} x^{\alpha}\psi_{j,\alpha}(x)^2
\end{equation}
where $|\alpha|:=\alpha_1+\cdots+\alpha_n$, $x^{\alpha}:=x_1^{\alpha_1}\cdots x_n^{\alpha_n}$, and each $\psi_{j,\alpha}$ is a homogeneous polynomial of degree $j$.
\end{proposition}

{Recall from Proposition \ref{prop:Ext-to-SOS} that a necessary condition for $\XLDUI{A}{B}\in (\mathsf{PPTBExt}_n^{(r-1)})^{\circ}$ is 
    $$A+A^{\top}+2 {\rm Re}(\mathring{B})\in \K_n^{(r-2)}.$$
The next result shows that this is actually also sufficient when we consider the dual of $\mathcal{K}_r$-PPT bosonic extendibility (note that $(\mathsf{PPTBExt}_n^{(r-1)})^{\circ}\subseteq (\mathsf{PPTBExt}_n^{(\mathcal{K}_r)})^{\circ}$).
}

\begin{theorem} \label{thm:Ext-from-SOS}
Let $(A,B)\in \Mreal{n} \underset{\R{n}}{\times} \Msa{n}$ and let $r\geq 2$. Then $\XLDUI{A}{B}\in (\mathsf{PPTBExt}_n^{(\mathcal{K}_r)})^{\circ}$ if and only if
    $$A+A^{\top}+2{\rm Re}(\mathring{B})\in \K_n^{(r-2)}.$$
In particular, we have $\XLDUI{N}{M-\mathring{N}}\in (\mathsf{PPTBExt}_n^{(\mathcal{K}_r)})^{\circ}$ for any $M\in \mathsf{K}_n^{(r-2)}$ and $N\in \Mrealsa{n}$ such that ${\rm diag}(N)={\rm diag}(M)$.
\end{theorem}
\begin{proof}
Let us consider a Hermitian polynomial
\begin{align*}
    h(z) &:=\la z^{\otimes 2}|\XLDUI{A}{B}|z^{\otimes 2}\ra\|z\|_2^{2(r-2)} \\
    &=\la z\odot \overline{z}|A+\mathring{B}|z\odot \overline{z}\ra \|z\|_2^{2(r-2)} \\
    &= \frac{1}{2}\la z\odot \overline{z} |A+A^{\top}+2{\rm Re}(\mathring{B})| z\odot \overline{z} \ra \|z\|_2^{2(r-2)}
\end{align*}
for $z\in \Comp^n$, where the second and third equalities follow from the duality relation \cref{eq:LDUI-duality} and the observation $z\odot\overline{z}\in \Real^n$. If $\XLDUI{A}{B}\in (\mathsf{PPTBExt}_n^{(\mathcal{K}_r)})^{\circ}$, then
\cref{thm:PPTExtDual} implies that the real polynomial $x\in \Real^n \mapsto h(x)$ is a sum of squares, and hence $A+A^{\top}+2{\rm Re}(\mathring{B}) \in \mathsf{K}_n^{(r-2)}$. 

Conversely, if $A+A^{\top}+2{\rm Re}(\mathring{B})\in \mathsf{K}_n^{(r-2)}$, then Proposition \ref{prop:SOS} implies that $h(z)=p(z\odot \overline{z})$ where $p=\sum_{j=0}^{\lfloor r/2\rfloor}\sum_{\alpha} x^{\alpha}\psi_{j,\alpha}^2$ is written as in \cref{eq:SOS}. Let us further write
    $$\psi_{j,\alpha}(x)=\sum_{\beta\in \mathbb{N}_0^n,\,|\beta|=j} c_{\beta}^{(j,\alpha)}x^{\beta}$$
for coefficients $c_{\beta}^{(j,\alpha)}\in \Real$. Now for $0\leq j\leq \lfloor r/2\rfloor$ and $\alpha\in \mathbb{N}_0^n$, one has
\begin{align*}
    (z\odot \overline{z})^{\alpha}\big(\psi_{j,\alpha}(z\odot \overline{z})\big)^2 &= |z_1|^{2\alpha_1}\cdots |z_n|^{2\alpha_n}\Big(\sum_{|\beta|=j} c_{\beta}^{(j,\alpha)}|z_1|^{2\beta}\cdots |z_n|^{2\beta_n}\Big)^2\\
    &=\Big|\sum_{|\beta|=j} c_{\beta}^{(j,\alpha)} z_1^{\alpha_1+\beta_1}\cdots z_n^{\alpha_n+\beta_n}\overline{z}_1^{\beta_1} \cdots \overline{z}_n^{\beta_n}\Big|^2\\
    &=\big|\la \varphi_{j,\alpha}|\overline{z}^{\otimes j} z^{\otimes r-j}\ra\big|^2,
\end{align*}
where $|\varphi_{j,\alpha}\ra := \sum_{|\beta|=j}c_{\beta}^{(j,\alpha)}\big(|1\ra^{\otimes \beta_1}\cdots|n\ra^{\otimes \beta_n}\otimes |1\ra^{\otimes \alpha_1+\beta_1}\cdots |n\ra^{\otimes \alpha_n+\beta_n}\big) \in (\Comp^n)^{\otimes r}$. Consequently, we obtain $\displaystyle h(z)=\sum_{j=0}^{\lfloor r/2 \rfloor} \la \overline{z}^{\otimes j}z^{\otimes r-j}| Z_j|\overline{v}^{\otimes j}z^{\otimes z-j}\ra$ for positive semidefinite matrices
    $$Z_j=\sum_{|\alpha|=r-2j}|\varphi_{j,\alpha}\ra \la \varphi_{j,\alpha}|\in \mathsf{PSD}\big(\Comp^{\otimes r}\big), \quad 0\leq j\leq \lfloor r/2\rfloor.$$
By comparison with \cref{eq:PPTExtDual}, we conclude that $\XLDUI{A}{B}\in (\mathsf{PPTBExt}_n^{(\mathcal{K}_r)})^{\circ}$.
\end{proof}

We provide further characterizations of the dual cones of separable / PPT bosonic states.

\begin{corollary}
\label{cor:bosonic-witness-cop}
For a pair $(A,B) \in \Mreal{n} \underset{\R{n}}{\times} \Msa{n}$, 
\begin{enumerate}
    \item $\XLDUI{A}{B}\in \mathsf{Sep}(\Comp^n\vee \Comp^n)^{\circ}$ if and only if $A+A^{\top}+2{\rm Re}(\mathring{B})\in \mathsf{COP}_n$,

    \item $\XLDUI{A}{B}\in \PPT(\Comp^n\vee \Comp^n)^{\circ}$ if and only if $A+A^{\top}+2\Re(\mathring{B})\in \K_n^{(0)} = \SPN_n$.
\end{enumerate}
In particular, given $M\in \COP_n\setminus \SPN_n$, we can generate an LDUI operator $\XLDUI{N}{M-\mathring{N}}$, where $\diag(N)=\diag(M)$ and $N\in \Mrealsa{n}$, as a witness for PPT entanglement of (bipartite) bosonic states.
\end{corollary}
\begin{proof}
Since $\mathsf{Sep}(\Comp^n\vee \Comp^n)={\rm conv}\{|v^{\otimes 2}\ra\la v^{\otimes 2}|: v\in \Comp^n \}$, $\XLDUI{A}{B}\in \mathsf{Sep}(\Comp^n\vee \Comp^n)^{\circ}$ if and only if
\begin{align*}
    & \la z^{\otimes 2}|\XLDUI{A}{B}|z^{\otimes 2}\ra = \frac{1}{2}\la z\odot \overline{z} |A+A^{\top}+2{\rm Re}(\mathring{B})| z\odot \overline{z} \ra \geq 0 \quad \forall\,z\in \Comp^n \\
    & \iff \la x^{\odot 2} |A+A^{\top}+2{\rm Re}(\mathring{B})| x^{\odot 2} \ra \geq 0 \quad \forall\,x\in \Real^n \\
    & \iff A+A^{\top}+2{\rm Re}(\mathring{B})\in \COP_n,
\end{align*}
showing (1). The assertion (2) simply follows from \cref{thm:Ext-from-SOS} with $r=2$.
\end{proof}

\begin{remark}
The class of bosonic witness introduced in Corollary \ref{cor:bosonic-witness-cop} generalize the witnesses considered in  \cite{marconi2021entangled}. Their ``indecomposable bosonic" witness corresponds to setting $(A,B) = (\operatorname{diag}(M),M)$, and $\XLDUI{\diag(M)}{M} = \sum_{ij} M_{ij} \ketbra{ij}{ji}$. Moreover, our results in \cref{thm:Ext-from-SOS} show that these witnesses can be generalized using copositive hierarchies to detect \textbf{complete} bosonic extendibility of quantum states.   
\end{remark}

{
Let us also record the result analogous to \cref{thm:EWnotSOS}, which follows directly from \cref{thm:Ext-from-SOS,cor:bosonic-witness-cop}.
\begin{corollary}
If $n\geq 6$ and $M\in \COP_n\setminus \bigcup_{r\geq 0} \K_n^{(r)}$, then we have for any $N\in \Mrealsa{n}$,
    $$\XLDUI{N}{M-\mathring{N}}\in \mathsf{Sep}(\Comp^n\vee\Comp^n)^{\circ}\setminus \bigcup_{r\geq 2}(\mathsf{PPTBExt}_n^{(\mathcal{K}_r)})^{\circ}.$$
In particular, $\mathsf{Sep}(\Comp^n\vee\Comp^n)^{\circ}=\overline{\bigcup_{r\geq 2}(\mathsf{PPTBExt}_n^{(\mathcal{K}_r)})^{\circ}}\supsetneq \bigcup_{r\geq 2}(\mathsf{PPTBExt}_n^{(\mathcal{K}_r)})^{\circ}$ whenever $n\geq 6$.
\end{corollary}

Now we shift our focus to the characterization of (complete graph) extendibility of \emph{Dicke states} introduced in \cref{sec:LDUI}.} First of all, the following proposition is an easy application of \cite[Proposition 3.2]{PY24}, but we include its proof for completeness.

\begin{proposition} \label{prop:TwirlingExt}
For $\pi:G\to \mathcal{U}_n$ a unitary representation of a compact group $G$, consider the $\pi\otimes \pi$-twirling operation
    $$\mathcal{T}_{\pi\otimes \pi}(X):=\int_{G}(\pi(g)\otimes \pi(g))X(\pi(g)\otimes \pi(g))^* dg, \quad X\in \M{n}\otimes \M{n},$$
where $dg$ denotes the (normalized) Haar measure on $G$. Then $\mathcal{T}_{\pi\otimes \pi}$ preserves the properties regarding extendibility: if $X\in \mathsf{PPTBExt}_n^{(\mathcal{K}_r)}$ (resp. $(\mathsf{PPTBExt}_n^{(\mathcal{K}_r)})^{\circ}$), then $\mathcal{T}_{\pi\otimes \pi}(X)\in \mathsf{PPTBExt}_n^{(\mathcal{K}_r)}$ (resp. $(\mathsf{PPTBExt}_n^{(\mathcal{K}_r)})^{\circ}$).
\end{proposition}
\begin{proof}
It is straightforward to check that, if $\tilde{\rho}$ is a $\mathcal{K}_r$-PPT bosonic extension of $\rho\in \mathsf{PPTBExt}_n^{(\mathcal{K}_r)}$, then $\mathcal{T}_{\pi^{\otimes r}}(\tilde{\rho})=\int \pi(g)^{\otimes r}\tilde{\rho}(\pi(g)^*)^{\otimes r}dg$ is a $\mathcal{K}_r$-PPT bosonic extension of $\mathcal{T}_{\pi\otimes \pi}(\rho)$. On the other hand, if $W\in (\mathsf{PPTBExt}_n^{(\mathcal{K}_r)})^{\circ}$, then
    $$\Tr(\mathcal{T}_{\pi\otimes \pi}(W)\rho)=\Tr(W\mathcal{T}_{\pi\otimes \pi}(\rho))\geq 0$$
since $\mathcal{T}_{\pi\otimes \pi}(\rho)\in \mathsf{PPTBExt}_n^{(\mathcal{K}_r)}$. This gives $\mathcal{T}_{\pi\otimes \pi}(W)\in (\mathsf{PPTBExt}_n^{(\mathcal{K}_r)})^{\circ}$.
\end{proof}

Since $\mathsf{K}_n^{(r)}$ a proper subset of $\mathsf{COP}_n$ for all $n\geq 5$ and $r\geq 0$, we obtain the following main result in this section, which provides examples for quantum states in $\mathsf{PPTBExt}_n^{(\mathcal{K}_r)}\setminus \mathsf{Sep}_{n}$.

\begin{theorem} \label{thm:DickeExt}
The Dicke state $\XLDUI{P}{P}$ is $\mathcal{K}_r$-PPT bosonic extendible ($r \geq 2$) if and only if $P\in (\mathsf{K}_n^{(r-2)})^{\circ}$. In particular, for $n \geq 5$, any matrix $P\in (\mathsf{K}_n^{(r-2)})^{\circ}\setminus \mathsf{CP}_n$ gives a $\mathcal K_{r}$-PPT bosonic extendible entangled state $\XLDUI{P}{P}$.
\end{theorem}
\begin{proof}
Suppose $P\in (\mathsf{K}_n^{(r-2)})^{\circ}$ and let $W\in (\mathsf{PPTBExt}_n^{(\mathcal{K}_r)})^{\circ}$. Then by \cref{thm:Ext-from-SOS,prop:TwirlingExt}, $\XLDUI{A}{B}:=\mathcal{T}_{\mathsf{LDUI}}W\in (\mathsf{PPTBExt}_n^{(\mathcal{K}_r)})^{\circ}$ and $(A+\mathring{B})+(A+\mathring{B})^{\top}\in \mathsf{K}_n^{(r-2)}$. Therefore, one has
\begin{align*}
    \la \XLDUI{P}{P}, W\ra &= \la \mathcal{T}_{\mathsf{LDUI}}\XLDUI{P}{P},W\ra = \la \XLDUI{A}{A}, \mathcal{T}_{\mathsf{LDUI}}W\ra \\
    &= \la \XLDUI{P}{P}, \XLDUI{A}{B}\ra = \la P, A+\mathring{B}\ra \\
    &= \frac{1}{2}\la P, (A+\mathring{B})+(A+\mathring{B})^{\top}\ra\geq 0.
\end{align*}
This shows that $\XLDUI{P}{P}\in \mathsf{PPTBExt}_n^{(\mathcal{K}_r)}$. Conversely, suppose $\XLDUI{P}{P}\in \mathsf{PPTBExt}_n^{(\mathcal{K}_r)}$ and let $M\in \mathsf{K}_n^{(r-2)}$. Then \cref{thm:Ext-from-SOS} again implies that $\XLDUI{Y}{Y}\in (\mathsf{PPTBExt}_n^{(\mathcal{K}_r)})^{\circ}$ where $Y=\frac{1}{2}(M+{\rm diag}(M))$. Therefore, we have
    $$\la P,M\ra=\la P,Y+\mathring{Y}\ra=\la \XLDUI{P}{P},\XLDUI{Y}{Y}\ra\geq 0,$$
establishing $P\in (\mathsf{K}_n^{(r-2)})^{\circ}$.
\end{proof}

\begin{remark}
This result gives examples of $\mathcal{K}_r$-PPT bosonic extendible entangled states for arbitrary $r \geq 2$. On the other hand, \cref{thm:DickeExt} also implies that the CLDUI state $\XCLDUI{P}{P}$ is $\mathcal{K}_r$-PPT (non-bosonic) extendible for $P\in (\mathsf{K}_n^{(r-2)})^{\circ}$. The main difference is that $\XCLDUI{P}{P}=(\XLDUI{P}{P})^{\Gamma}$ is not a bosonic state, but is symmetric in the sense of commuting with the flip operator $F:=\sum_{i,j}|ij\ra\la ji|$, i.e
    $$F\XCLDUI{P}{P} = \XCLDUI{P}{P}F.$$
\end{remark}

\begin{corollary} \label{cor:MultiDickeEnt}
For $n\geq 5$ and $r\geq 2$, there exists a bosonic state $\tilde{\rho}\in \mathsf{PSD}(\bigvee^r \Comp^n)$ that is PPT across every bipartition while being genuinely entangled (Remark \ref{rmk:MultiEntanglement}). Moreover, such a state also exists within the class of \emph{$r$-partite diagonal symmetric (or mixture of Dicke) states}, i.e., bosonic states with $\mathcal{DU}_n^{\otimes r}$-symmetry.
\end{corollary}

\begin{proof}
Let $P\in (\mathsf{K}_n^{(r-2)})^{\circ}\setminus \mathsf{CP}_n$ and let $\tilde{\rho}$ be a $\mathcal{K}_r$-PPT bosonic extension of a bipartite Dicke state $\XLDUI{P}{P}$. We claim that $\tilde{\rho}$ satisfies the advertised property. Indeed, since $(\id_n^{\otimes \{1,2\}}\otimes \Tr^{\otimes [r]\setminus \{1,2\}})(\tilde{\rho})=\XLDUI{P}{P}$
is entangled, and since partial trace operations are separability preserving, this implies that $\tilde{\rho}$ is not fully separable. Finally, the same property holds if we consider the \emph{$\mathcal{DU}^{\otimes r}$-twirling} 
    $$\mathcal{T}_{D\mathcal{U}_n^{\otimes r}}(\tilde{\rho}):=\int_{D\mathcal{U}_n}U^{\otimes r}\tilde{\rho}U^{*\otimes r} dU$$
of $\tilde{\rho}$, which is now a diagonal symmetric extension of $\XLDUI{P}{P}$.
\end{proof}

\begin{remark}
The same conclusion as in Corollary \ref{cor:MultiDickeEnt} was obtained in \cite{romero2025multipartite}. A careful inspection of the arguments in that paper reveals that the existence of an $r$-partite PPT Dicke state having a PPT entangled bipartite reduction is assumed without proof. \cref{thm:DickeExt,cor:MultiDickeEnt} fills in this gap in a mathematically rigorous way.
\end{remark}

We also refer to \cite{augusiak2012entangled} for a construction of $r$-qubit ($r \geq 4$) entangled states that are PPT across all bipartitions. Since we do not have the complete characterization of the hierarchy $\mathsf{PPTBExt}_n^{(r-1)}$, we do not know whether two hierarchies we consider are equivalent for bosonic states, or particularly mixtures of Dicke states. We pose this as the open question. 

\begin{question}
For mixtures of Dicke states, are the following two sets equivalent? 
    $$\mathsf{PPTBExt}_n^{(\mathcal{K}_r)} = \mathsf{PPTBExt}_n^{(r-1)}, \quad r\geq 2$$
\end{question}

The sets are equivalent to if  Question \ref{question:ExtLifting} admits a affirmative solution, in which case the arguments in \cref{thm:DickeExt} can be applied analogously to star $r$-extendibility. We would also like to state the following question regarding the $\K_n^{(r)}$ hierarchy and its potential implications for the theory of entanglement. 

\begin{question}
For $n \geq 5$, is the inclusion $\K_n^{(r)}\subseteq \K_n^{(r+1)}$ strict for every $r\geq 1$? The affirmative answer would imply that all the inclusions in \cref{eq:PPTBExt} are strict. Furthermore, the same strict inclusions would hold within the class of (mixtures of) Dicke states. 
\end{question}

\section{Discussion and open problems}

This work introduces a unified and powerful framework that connects four key areas: the theory of \emph{positive linear maps}, the \emph{optimization theory of copositive matrices}, \emph{algebraic graph theory}, and the rich structure of \emph{quantum entanglement} in mixed states.

The impact of this unification is quite broad. For quantum information theorists, we provide a complete characterization of positivity for a major class of maps and a systematic method for producing new indecomposable maps and highly extendible entangled states. The established ``lifting'' theorems create a direct bridge from the mature, classical theory of copositive matrices to the quantum world, allowing for the construction of vast families of maps and states having useful properties for entanglement detection. Furthermore, our graph-theoretic construction, which links the properties of maps to fundamental graph parameters, offers a rich and structured source of these crucial examples.

We also believe this work will spur significant new directions of research. The newly introduced families of maps and states invite further investigation, particularly, it would be interesting to have a generalization of the copositive hierarchies to hierarchies for $\COPCP$. The explicit connection between SOS hierarchies and state extendibility opens the door to applying powerful semidefinite programming tools developed in optimization theory to unsolved problems in entanglement theory. The new classes of positive-indecomposable maps might play a role in answering questions around PPT entanglement, particularly the PPT$^2$ conjecture \cite{PPTsq,Christandl2018,singh2022ppt} and the Schmidt number of PPT entangled states \cite{huber2018high,pal2019class,krebs2024high}.

\bigskip

\noindent\textbf{Acknowledgments.} The authors thank the organizers of the workshop \href{https://homepages.laas.fr/henrion/pop23/}{POP23 - Future Trends in Polynomial OPtimization} where this project was started. This research was supported by the ANR project \href{https://esquisses.math.cnrs.fr/}{ESQuisses}, grant number ANR-20-CE47-0014-01. A.G received support from the University Research School EUR-MINT
(State support managed by the National Research Agency for Future Investments
program bearing the reference ANR-18-EURE-0023)

\bibliography{references}
\bibliographystyle{alpha}

\end{document}